\documentclass[12pt]{iopart}

\usepackage{iopams}  

\usepackage{hyperref}
\usepackage{graphicx}
\usepackage{curves}
\usepackage{enumerate}
\expandafter\let\csname equation*\endcsname\relax

\expandafter\let\csname endequation*\endcsname\relax

\usepackage{amsmath,amsthm,amssymb}
\usepackage{enumerate}
\usepackage{amsmath}
\usepackage{amssymb}
\usepackage{amsfonts}
\usepackage{amsthm}
\usepackage{epsfig}
\usepackage{graphicx}
\usepackage{soul}
\usepackage[normalem]{ulem}

\usepackage{mathrsfs}

\usepackage[square,numbers,sort&compress]{natbib}

\usepackage{pgf,tikz,pgfplots}
\pgfplotsset{compat=1.15}
\usepackage{mathrsfs}
\usetikzlibrary{arrows}
\newtheorem{thm}{Theorem}[section]
\newtheorem{lemma}{Lemma}[section]
\newtheorem{rem}{Remark}[section]
\newtheorem{prop}{Proposition}[section]
\newtheorem{RHP}{Riemann-Hilbert problem}
\newtheorem{Cor}{Corollary} [section]
\newtheorem{Assumption}{Assumption}

\newcommand{\ee}{\mathrm{e}}
\renewcommand{\rmi}{i}
\newcommand{\imi}{i}

\newcommand{\D}{\displaystyle}
\newcommand{\la}{\lambda}
\newcommand{\ii}{i}
\renewcommand{\i}{i}
\newcommand{\w}{\mathrm{w}}
\newcommand{\ord}{\mathcal{O}}
\renewcommand{\Im}{\operatorname{Im}}
\renewcommand{\Re}{\operatorname{Re}}

\newcommand{\mC}{{\mathbb C}}

\renewcommand{\Re}{\operatorname{Re}}
\newcommand{\ol}{\overline}

 \providecommand{\D}{\mathbb}

\renewcommand{\(}{\left(}
\renewcommand{\)}{\right)}

\usepackage{geometry}
\usepackage{color,xcolor}
\usepackage[section,above]{placeins}
\usepackage[square,numbers,sort&compress]{natbib}

\usepackage{pgf,tikz,pgfplots}
\pgfplotsset{compat=1.15}
\usepackage{mathrsfs}
\usetikzlibrary{arrows}
\usetikzlibrary{patterns, decorations.markings, arrows, calc}

\usepackage{etoolbox}

\makeatletter
\newrobustcmd{\fixappendix}{%
  \patchcmd{\l@section}{1.5em}{7em}{}{}%
  \patchcmd{\l@subsection}{2.3em}{7em}{}{}%
}
\makeatother

\begin{document}

\title{Maxwell-Bloch equations without spectral broadening: the long-time asymptotics of an input pulse in a long two-level laser amplifier}

	\author{Volodymyr Kotlyarov$^{\square}$ and  Oleksandr Minakov$^{\triangle}$}
		
	\address{${ }^{\square}$\it\small  B. Verkin Institute for Low Temperature Physics and Engineering \\
		\it\small  of the National Academy of Sciences of Ukraine, 47, Nauky Ave., Kharkiv, 61103 Ukraine\\
			${}^{\triangle}$\it\small  Department of Mathematical Analysis, Faculty of Mathematics and Physics, Charles University, 					Sokolovska 8, Prague 8, 186\,75 Czech Republic\\
	}

	\ead{kotlyarov@ilt.kharkov.ua ;\; minakov(at)karlin(dot)mff(dot)cuni(dot)cz
	}
	\vspace{10pt}
	
	\begin{indented}
		\item[]October 2022
	\end{indented}

\begin{abstract}
We study the problem of propagation of an input electromagnetic pulse through a long two-level laser amplifier under trivial initial conditions. In this paper, we consider an unstable model described by the Maxwell-Bloch equations without spectral broadening. Previously, this model was studied by S.V. Manakov in \cite{Manakov82} and together with V.Yu. Novokshenov in \cite{MN86}. We consider this model in a more natural formulation as an initial-boundary (mixed) problem using a modern version of the inverse scattering transform method in the form of a suitable Riemann-Hilbert (RH) problem. The RH problem arises as a result of applying the Fokas-Its method of simultaneous analysis of the corresponding spectral problems for the Ablowitz-Kaup-Newell-Segur (AKNS) equations. This approach makes it possible to obtain rigorous asymptotic results at large times, which differ significantly from the previous ones. Differences take place both near the light cone and in the tail region, where a new type of solitons is found against an oscillating background. These solitons are physically relevant, their velocities are smaller than the speed of light. The number of such solitons can be either finite or infinite (in the latter case, the set of zeros has a condensation point at infinity). Such solitons can not be reflectionless, they are generated by zeros of the reflection coefficient of the input pulse (and not by poles of the transmission coefficient).

Thus our approach shows the presence of a new phenomenon in soliton theory, namely, the boundary condition (input pulse) of a mixed problem under trivial initial conditions can generate solitons due to the zeros of the reflection coefficient, while the poles of the transmission coefficient do not contribute to the asymptotics of the solution.

\vskip 0.5cm
\noindent Keywords: Maxwell-Bloch equations, spectral broadening, two-level laser amplifier, Riemann-Hilbert problem, unified spectral method, rigorous asymptotics\\ 
Mathematics Subject Classification numbers: 
35B40, 35Q15, 35Q51, 35Q60, 35M13, 30E20, 37K10, 37K40, 45E05

\end{abstract}

\newpage
\tableofcontents
\title[]{}

\maketitle

\makeatletter

\clearpage

\newpage
\section{Introduction and results}

The integrable Maxwell-Bloch (MB) equations  have the following form (cf.\cite{GZM85})
\begin{align}
\label{MB1}
& \frac{\partial \cal E}{\partial t}+\frac{\partial \cal E}{\partial x}=\int_{-\infty}^\infty\rho(t,x,\lambda)n(\lambda)d\lambda, \qquad
\int_{-\infty}^\infty n(\lambda)d\lambda=1,\\
\label{MB2}
& \frac{\partial \rho}{\partial t}+2\rmi\lambda\rho=\cal N\cal E,\\
\label{MB3}
&\frac{\partial \cal N}{\partial t}=-\,\frac{1}{2}\left(\ol{\cal E}\rho+\cal E \ol{\rho}\right).
\end{align}
Here ${\cal{E}}={\cal E}(t,x)$ is a complex-valued function of the time variable $t$ and spatial variable $x$, $\rho=\rho(t,x,\lambda)$ and $\mathcal{N} = {\cal N}(t,x,\lambda)$ are respectively complex-valued and real-valued functions of $t$, $x$ and spectral parameter $\lambda$, and the bar denotes the complex conjugation. 
	
Equations \eqref{MB1}--\eqref{MB3} arise in a few physical models and their studying was launched in \cite{L1}-\cite{L4}.
The next very important step was done in \cite{AKN}, where the inverse scattering transform method was developed for a self-induced transparency model. In this paper, we are interested in a model of quantum laser amplifier, which was studied in \cite{Manakov82, MN86}. For these models, ${\mathcal E}(t,x)$ is the complex-valued envelope of an electromagnetic wave with a fixed polarization, ${\mathcal N}(t,x,\lambda)$ and ${\rho}(t,x,\lambda)$ are entries of the density matrix of the atom subsystem
\begin{equation}\label{MB5}
\begin{pmatrix}{\cal N}(t,x,\lambda)&\rho(t,x,\lambda)\\
\ol{\rho(t,x,\lambda)} & -{\cal N}(t,x,\lambda)\end{pmatrix}.
	\end{equation}
Parameter $\lambda$ denotes a deviation of the passage frequency from its mean value.
The weight function $n(\lambda)$ in equation \eqref{MB1} characterizes  the inhomogeneous broadening. 

In recent years, interest in various problems related to the Maxwell-Bloch equations has grown noticeably.
For short reviews on the MB equations and applications of the inverse scattering transform method  to them  see \cite{AKN, AS}, \cite{GZM83}-\cite{GZM85}, \cite{K13}, \cite{Zakh80}.
We note the work of \cite{LM2022} where the authors study the Cauchy problem for the Maxwell-Bloch equations  of light-matter interaction, under assumptions that prevent the generation of solitons. It concerns the aftereffect of the passage of an optical pulse in an active (stable and unstable) medium.

In this paper, we study  the case of the infinitely narrow spectral line, i.e. without the spectral broadening, when $n(\lambda)=\delta(\lambda)$, where $\delta(.)$ is the Dirac $\delta$-function. Then the system \eqref{MB1}--\eqref{MB3} is reduced to the form ($\la=0$)
\begin{equation}\label{MB1a}
\frac{\partial \mathcal E}{\partial t}+\frac{\partial \mathcal E}{\partial x}= \rho,\qquad
\frac{\partial \rho}{\partial t}={\mathcal N\mathcal E},\qquad
\frac{\partial {\mathcal N}}{\partial t}=-\,\frac{1}{2}\left(\ol{\mathcal E}\rho+\mathcal E\ol{\rho}\right)
\end{equation}
where
$\mathcal{E} = \mathcal{E}(t, x),$ $\mathcal{N} = \mathcal{N}(t, x) := \mathcal{N}(t, x, 0),$ $\rho = \rho(t, x) := \rho(t,x,0),$ 
and the initial-boundary value (mixed) problem is defined by the following  conditions:
\begin{equation}\label{IBC}
{\cal E}(0,x)={\cal E}_0(x),\quad \rho(0,x )=\rho_0(x ),\quad {\cal N}(0,x)={\cal N}_0(x), \quad {\cal E}(t,0)={\cal E}_1(t),
\end{equation}
where $x\in[0,l)$ ($l\le\infty$) and $t\in(0, +\infty).$

Note that the functions $\rho(t,x)$, ${\cal N}(t,x)$ are not independent; indeed, equations \eqref{MB2}, \eqref{MB3} imply $\frac{\partial}{\partial t} \left(|\rho(t,x)|^2+{\cal N}(t,x)^2\right)=0,$ 
and hence without loss of generality we can assume
$$
|\rho(t,x)|^2+{\cal N}(t,x)^2\equiv1.
$$
Then $\mathcal{N}(0,x)$ is given by
\begin{equation}\label{N0}
	{\cal N}(0,x)=\mp\sqrt{1-|\rho(0,x)|^2}.
\end{equation}		
The sign ``minus" corresponds to a stable medium (attenuator). The sign ``plus"  corresponds to an unstable medium (for example, a quantum two-level laser amplifier), which is the subject of our study.

An approach to the study of the mixed problem for the Maxwell-Bloch equations, based on the formalism of the matrix Riemann-Hilbert (RH) problem, was proposed in \cite{K13} for the case of an arbitrary spectral broadening and in \cite{MK06}, \cite{KM14}, \cite{FKM17} for the case without spectral broadening.
Furthermore, the full linearization of the mixed problem is established in \cite{K13}.
The corresponding matrix RH problems were formulated in terms of spectral functions defined through given initial and boundary conditions for the MB equations by using the Fokas-Its method of simultaneous spectral analysis of the corresponding AKNS equations \cite{Fok97} - \cite{FI92}, \cite{BFS03} - \cite{BKS11}. 

Our goal is to study the asymptotic behaviour of a solution of the mixed problem for the MB equations \eqref{MB1a}.
More precisely, we study the problem of propagation of an input electromagnetic pulse 
\begin{equation}
\label{input_pulse}
\mathcal{E}(t,0)=\mathcal{E}_{1}(t),\quad t>0,
\end{equation}
through a long two-level laser amplifier under trivial initial conditions, i.e.: 
\begin{equation}\label{trivial}
{\cal E}(0,x)=\rho(0,x )\equiv0 ,\quad {\cal N}(0,x)\equiv 1, \quad  x\geq0
\end{equation}
(existence and uniqueness of such a solution is established in Proposition \ref{prop_ibv} and 
\ref{sect_uniqueness}).

Such a problem was earlier treated in \cite{Manakov82} by S.V.~Manakov  and  in \cite{MN86}  together with V.Yu.~Novokshenov, but in a different formulation: they considered the Cauchy problem on the whole $t$-axis with an input pulse equal to zero for negative $t.$ 

We consider this model in a more natural formulation as an initial-boundary (mixed) problem using a modern version of the inverse scattering transform method in the form of a suitable Riemann-Hilbert (RH) problem. The corresponding matrix Riemann-Hilbert problem is formulated on a contour that is the union of the continuous spectra of the Lax operators (generated by AKNS spectral problems) for the Maxwell-Bloch equations and which consists of the real axis and the circle of radius 1/2 centred at the origin of the complex plane. In this case, the jump matrices are exponentially growing  on the circle  and the corresponding phase function has saddle points on the imaginary axis  in the absence of stationary points on the real axis. These features of the RH problem lead to the fact that the asymptotic behaviour of the solution at large times near the light cone is a train of pulses of unboundedly growing amplitude and contracting width. Each pulse has a speed that approaches the speed of light. 

Our results share some qualitative features with the ones obtained by S.V. Manakov \cite{Manakov82} in 1982, but differ from them. This is due to the difference in approaches, which is that we use the inverse scattering transform method in the form of a matrix RH problem, the  Fokas-Its unified method  of simultaneous spectral analysis of the corresponding AKNS equations and the rigorous Deift-Zhou steepest descent method \cite{DIZ93} - \cite{DZ93},  while S.V. Manakov  did not use the true reflection coefficient, but its approximation in the form of the Fourier transform  of the input pulse in the assumption of its smallness. At the same time, he claimed that ``{\it the long-time solution becomes  essentially (and in a certain sense, extremely) nonlinear.}" Therefore
it seems that replacing the reflection coefficient with the Fourier transform of the input pulse is not completely justifiable.
For a comparison between our results and those obtained in \cite{Manakov82} and \cite{MN86}, see Remark~\ref{rem_1.2} (for the region near the light cone), and Remark \ref{rem_1.3} (for the region of a rapidly oscillating self-similar wave).

\subsection*{Formulation of results.}

To formulate the results, let us introduce two functions, $r(k)$ and $a(k)^{-1},$ which are called the reflection and transmission coefficients, respectively, associated with the initial and boundary conditions of the problem \eqref{MB1a}, \eqref{input_pulse}, \eqref{trivial} (they are defined later in Section \ref{sect_scattering_data}), and let $b(k) = r(k)a(k).$
Function $r(k)$ is analytic in $\Im k\geq 0,$ and decays as $k\to\infty.$ 

\medskip
Our main results are summarised in Theorems \ref{thm_really_close_light_cone} and \ref{thm_tail} (see also Figure \ref{Fig_regions}). Theorem \ref{thm_really_close_light_cone} covers the region
\begin{equation}\label{ineq_as_sol}
 x  < t \leq x + \frac{1}{4x}\(m^2 \ln^2 x + C \ln x\cdot\ln\ln x\),
\end{equation}
where $C\in\mathbb{R}$ is an arbitrary number, and where the parameter $m$ characterizes the behaviour of the reflection coefficient $r(.)$ and is specified more precisely in Assumptions \ref{assumption_1+}, \ref{assumption_1} below.
Together with the causality principle (Theorem \ref{thm_caus}) it covers completely the region $0 \leq t \leq x + \frac{1}{4x}\(m \ln x + C_1 \ln\ln x\)^2, \quad C_1\in\mathbb{R}.$

Theorem \ref{thm_tail} covers the region 
$$
(1-\sigma)^{-1} x \leq t \leq \sigma^{-1} x,
\qquad \mbox{or, equivalently, } \qquad \sigma \leq \frac{x}{t} \leq 1-\sigma,
$$
where $\sigma$ is an arbitrary number in the interval $(0,\frac12).$

\medskip
In Theorem \ref{thm_really_close_light_cone}, we need to make one of the following assumptions on the behaviour of the reflection coefficient $r(k)$ in the Zakharov-Shabat spectral problem for the Dirac operator with potential defined by the input pulse $\mathcal{E}_1(t)$. We assume the fulfilment of Assumption \ref{assumption_1+} in parts I, II of Theorem~\ref{thm_really_close_light_cone},
and the slightly weaker Assumption~\ref{assumption_1} in parts III, IV.

\begin{Assumption}\label{assumption_1+}
Let the reflection coefficient $r(.)$ satisfy the following condition: there exist a real number $m\geq2$ and a nonzero complex number $C\in\mathbb{C}\setminus\left\{0\right\}$ such that
\[
r(k) = \frac{C}{k^m} + \ord(k^{-m-1})\mbox{ as } k\to\infty, \mbox{ uniformly in } \Im k\geq0.
\]
\end{Assumption}
\begin{Assumption}\label{assumption_1}
Let the reflection coefficient $r(.)$ satisfy the following condition: there exists a real number $m\geq2$ such that 
\[
r(k)\asymp k^{-m}\quad \mbox{ as } k\to\infty,\quad \mbox{uniformly in }\Im k \geq 0\] (here the symbol $\asymp$ means `of the same order', i.e. there exist two positive constants $0<C_1<C_2$ such that $C_1|k|^{-m} \leq |r(k)| \leq C_2|k|^{-m}$  as $|k|\to\infty,$ uniformly in $\Im k\geq 0$).
\end{Assumption}
\begin{rem}Assumption \ref{assumption_1+} is satisfied for instance in the case of trivial initial data \eqref{trivial} and a smooth and fast decaying for $t\to\infty$ input pulse \eqref{input_pulse} with the following behaviour at $t=0:$ $\mathcal{E}_1(t) = c_1 t^{m-1}(1+\ord(t)), \ t\to+0,$ for some $c_1\neq0.$
Note though that it is not satisfied by functions $\mathcal{E}_1$ from the Schwartz class with support on $[0, +\infty)$, since for them $r(k)$ decays faster than any power of $k$ as $k\to\infty$ (and also $\mathcal{E}_1(t)$ decays faster than any power of $t$ as $t\to0+$).
\end{rem}

\begin{thm}\label{thm_really_close_light_cone}[Near the light cone.]
Let an input pulse $\mathcal{E}_1(t)$ be not identically equal to zero and be integrable with the first moment
\begin{equation}\label{first_moment}
\int_0^{\infty}(1+t)|\mathcal{E}_1(t)|\,dt<\infty.
\end{equation}
Then the solution of the initial-boundary value (ibv) problem \eqref{MB1a}, \eqref{input_pulse}, \eqref{trivial} (which exists and is unique in view of Proposition \ref{prop_ibv} and \ref{sect_uniqueness} below) has the following behaviour:
\begin{enumerate}[I.]
\item
under Assumption \ref{assumption_1+},
in the limit as $k_0\to\infty,$ where 
\begin{equation}\label{k_0}
k_0=k_0(t/x)\equiv\frac12\sqrt{\frac{x}{t-x}},
\end{equation}
uniformly in the domain
\[
\left\{(t,x):\quad x\ <\ t\ \leq\ x + \frac{1}{x}\right\}
\]
we have
\begin{align}\label{E_thm_1}
&\mathcal{E}(t,x) = 4  k_0\, r(i k_0) \, I_{m-1}\(2\sqrt{x(t-x)}\) + \ord(k_0^{-m}),
\\\label{N_thm_1}
&\mathcal{N}(t,x) = 1 - 2|r(ik_0)|^2 \left(I_m\(2\sqrt{x(t-x)}\)\right)^2 + \ord(k_0^{-2m-1}),
\\\label{rho_thm_1}
&\rho(t,x) = 2\, r(i k_0)\, I_m\(2\sqrt{x(t-x)}\) + \ord(k_0^{-m-1}).
\end{align}
Here $I_\nu$ is the modified Bessel function of the first kind of the order $\nu, \ \nu = m-1, m.$
\item 
Let $\varepsilon_1>0$ be a fixed number. 
Denote
\begin{equation}\label{p}
p_1(t,x) = m \ln x - m \ln\sqrt{x(t-x)} - 2\sqrt{x(t-x)}.
\end{equation}
Then under Assumption \ref{assumption_1+}, in the limit as $x\to\infty$, uniformly for $(t,x)$ in the domain
\begin{equation}\label{domain_part_II}
\left\{(t,x):\ x+\frac{1}{x}\ \leq\  t \ \leq \ x + \frac{1}{4x}\(m\ln x - (m+\varepsilon_1)\ln\ln x\)^2
\right\}
\end{equation}
we have $e^{-p_1(t,x)} = \ord\((\ln x)^{-\varepsilon_1}\)$ and 
\[
\mathcal{E}(t,x) = 4 \, k_0\cdot r(i k_0) \cdot I_{m-1}\(2\sqrt{x(t-x)}\) + \ord(e^{-p_1(t,x)} + k_0 e^{-2p_1(t,x)}),
\]
\[
\mathcal{N}(t,x) = 1 - 2|r(ik_0)|^2 \(I_m(2\sqrt{x(t-x)})\)^2 + \ord(e^{-2p_1(t,x)}k_0^{-1} + e^{-3p_1(t,x)}),
\]
\[
\rho(t,x) = 2\,r(i k_0)\, I_m\(2\sqrt{x(t-x)}\) + \ord(e^{-p_1(t,x)}k_0^{-1} + e^{-2p_1(t,x)}).
\]
(Note that the estimate $e^{-p_1(t,x)} = \ord((\ln x)^{-\varepsilon_1})$ is the worst possible in the region \eqref{domain_part_II} estimate, but it might be better.)
\item Let $\varepsilon_2\in(0,\frac12)$ and $K\geq m + \varepsilon_1$ be fixed numbers, where $\varepsilon_1$ is a constant from part II. 
Denote
\begin{equation}\label{p12}
p_2(t,x) = m \ln x - 2\sqrt{x(t-x)} - (m-\frac12)\ln\sqrt{x(t-x)}.
\end{equation}
Then under Assumption \ref{assumption_1}, in the limit as $x\to\infty,$ uniformly for $(t,x)$ in the domain 
\begin{equation}\label{domain_part_III}
\left\{(t,x):\
x + \frac{1}{4x}\(m\ln x - K \ln\ln x\)^2
\leq
 t \leq
  x + \frac{1}{4x}\(\!m\ln x - (m+\varepsilon_2-{1/2})\ln\ln x\)^2\right\}
\end{equation}
we have $e^{-p_2(t,x)} = \ord\((\ln x)^{-\varepsilon_2}\)$  
and
\[
\mathcal{E}(t,x) = \frac{2k_0\cdot r(ik_0)\cdot e^{2\sqrt{x(t-x)}}}
{\sqrt{\pi}\cdot\sqrt[4]{x(t-x)}}
\(1 + \ord\((\ln x)^{-1} + e^{-p_2(t,x)}\)\),
\]
\[
\mathcal{N}(t,x) = 1 - \frac{|r(ik_0)|^2e^{4\sqrt{x(t-x)}}}{2\pi\sqrt{x(t-x)}}
+\ord\(e^{-2p_2(t,x)}(\ln x)^{-1} + e^{-3p_2(t,x)}\),
\]
\[
\rho(t,x) = \frac{r(ik_0)\,e^{2\sqrt{x(t-x)}}}{\sqrt{\pi}\,\sqrt[4]{x(t-x)}}
+ \ord\(e^{-p_2(t,x)}(\ln x)^{-1} + e^{-2p_2(t,x)}\).
\]

\item let Assumption \ref{assumption_1} be satisfied, and let $n=0,1,2,3, \ldots$ be an integer. Then in the limit as $x\to+\infty$ (or, equivalently, $t\to+\infty$) uniformly for $(t, x)$ in the domain
\begin{equation}\label{domain_part_IV}
\left\{(t,x):\
x + \frac{1}{4x}\(m\ln x+(n-m)\ln\ln x\)^2 \leq \, t\, \leq \, x + \frac{1}{4x}\(m\ln x+(n+1-m)\ln\ln x\)^2\right\},
\end{equation}
the solution of the problem \eqref{MB1a}, \eqref{input_pulse}, \eqref{trivial} has the following asymptotics:
\begin{align}\label{E_res_as_sol}
&\mathcal{E}(t, x)=2\sqrt{ \frac{x}{t-x}}\(\frac{(-1)^{n} \cdot e^{\ii\arg r(i k_0)}}
{\cosh\Theta_n(t, x)}+
\ord \( 
\frac{1}{\sqrt{\ln x}}
\)
\),
\\\label{N_res_as_sol}
&\mathcal{N}(t, x) = 1 - \frac{2}{\cosh^2\Theta_n(t, x)} + \ord\(\frac{1}{\sqrt{\ln x}}\),
\\\label{rho_res_as_sol}
&\rho(t, x) = \frac{2\,(-1)^{n-1}\ e^{i\arg r(i k_0)}\ \tanh\Theta_n(t, x)}{\cosh\Theta_n(t, x)} + \ord\(\frac{1}{\sqrt{\ln x}}\),
\end{align}
where
\[\Theta_n(t, x) = 2\sqrt{x(t-x)} - \(n+\frac12\)\ln\sqrt{x(t-x)} + \chi_n(k_0(t/x)),\]
and $k_0$ is defined in \eqref{k_0}, and
\[\chi_n(k_0)=
\ln\dfrac{|r(i k_0)|\cdot n!}{\sqrt{\pi}\cdot 2^{3n+2}}\ .
\]
\end{enumerate}
\end{thm} 
\begin{rem}
Formula \eqref{E_res_as_sol} shows that the output field $\mathcal{E}(t, x)$ is a sequence of pulses of unboundedly growing amplitude and contracting width, and formula \eqref{N_res_as_sol} shows that $\mathcal{N}(t, x)$ is close to $1$ away from the peaks and is close to $-1$ near the peaks.
\end{rem}
\begin{rem}
Theorem \ref{thm_really_close_light_cone}, part I, covers both the situations when $x$ is bounded and when $x$ might grow.
\end{rem}

\begin{rem}
Note that the results of parts II, III of Theorem \ref{thm_really_close_light_cone} are consistent, since $I_{\nu}(\xi) = \frac{e^{\xi}}{\sqrt{2\pi\xi}}(1+\ord(\xi^{-1}))$ as $\xi\to\infty,$ $|\arg\xi|<\frac{\pi}{2}$ (\cite[formula (9.7.1)]{abramowitz}).
\end{rem}

\begin{rem}
In their domain of overlap, part III and part IV with $n=0$ of Theorem \ref{thm_really_close_light_cone} are consistent.
\end{rem}

\begin{figure}\begin{center}
\begin{tikzpicture}
\fill[fill=blue!10](6,0)--(0,0)--(5,5)--(6, 5);
\draw[->] (-0,0) -- (6.1,0);
\draw[->] (0,0) -- (0,5);
\node at (6.4, 0.1) {$x$};
\node at (-0.3, 5) {$t$};
\draw[dashed](0,0) -- (5.5, 5.5);
\node at (6, 5.3) {$x=t$};
\node at (4, 1) {\color{blue}Causality region};
\draw[dashed, fill=green!10] (0.5, 5) -- (0, 0)-- (3.5,5);
\node at (1.2, 2.9) {\rotatebox{65}{Solitons due to $b(k)=0$}};
\draw[dashed, thin, fill=orange!50] (5, 5) to (0,0) [out=50, in=-135] to (4.2, 5);
\node at (2.3, 2) {\rotatebox{45}{\color{orange}Peaks of growing amplitude}};
\end{tikzpicture}\end{center}
\caption{Different regions of $x>0,t>0$ quarter plane. The causality region (in purple) is described in Theorem \ref{thm_caus}, the tail solitonic sector of the light cone (in green) is described in Theorem \ref{thm_tail} and the region near the boundary of the light cone (in orange) is described in Theorem \ref{thm_really_close_light_cone}.}
\label{Fig_regions}
\end{figure}
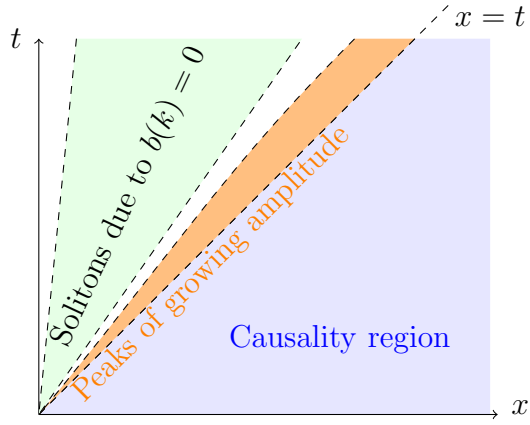

\begin{rem}\label{rem_1.2}

Formula \eqref{E_res_as_sol} can be rewritten in the form
\[
\mathcal{E}(t, x) =4k_0\(  \frac{(-1)^n \,e^{i\,\arg r(i k_0)}}{\cosh\left[2k_0 t - (2k_0 - \frac{1}{2k_0})x - x_0(t, x)\right]} + \mathcal{O}\(\frac{1}{\sqrt[4]{x(t-x)}}\) \),
\]
where 
$k_0 = k_0(\frac{t}{x})$ and $x_0(t, x) = (n+\frac12)\ln\sqrt{x(t-x)} - \chi(k_0(\frac{t}{x})).$
Note that it takes the form of a one soliton solution for the MB equation in an unstable medium,
\begin{equation*}
\begin{split}
\mathcal{E}(t, x) = \frac{-4k_2\,e^{-i\varphi_0}}
{\cosh\left[2k_2 t - (2k_2 - \frac{1}{2k_2})x - x_0
\right]},
\quad
\mathcal{N}(t, x) = 1 - \frac{2} {\cosh^2\left[2k_2 t - (2k_2 - \frac{1}{2k_2})x-x_0 \right]}.
\end{split}
\end{equation*}
We see that the above soliton has velocity $\(1-\frac{1}{4k_2^2}\)^{-1} >1,$ i.e. bigger than the velocity of the light, hence it does not have physical meaning.
It is remarkable that despite this fact, with the modulated parameters $x_0, \varphi_0, k_2,$ this soliton has a velocity smaller than the speed of light and represents the asymptotics of a physically meaningful problem.
Indeed, taking into account that $2\sqrt{x(t-x)}=4k_0(t-x)$, $(n+\frac12)\ln \sqrt{x(t-x)}= (n+\frac12)\ln(\frac{m}{2}\ln x) [1+o(1)]$, $\chi_n(k_0)=-m\ln k_0[1+o(1)]= (m\ln(\frac{m}{2}\ln x) - m \ln x)[1+o(1)]$  as $x\to\infty$, it is easy to verify that \eqref{E_res_as_sol} can be rewritten in the form:
\begin{equation}\label{Man}
\mathcal{E}(t, x) =4k_0\(  \dfrac{(-1)^n \,e^{\ii\,\arg\,r(\ii k_0)}} {\cosh\left[4k_0(t - x) -m\ln x -(n-m+1/2)\ln(\frac{m}{2}\ln x) \(1+o(1)\)\right]} +\mathcal{O}\(\frac{1}{\sqrt{\ln x}}\) \).	
\end{equation}
Formula \eqref{Man} agrees qualitatively with the one obtained by Manakov (\cite[formulae (33), (36)]{Manakov82}), but it does not coincide precisely in terms of the amplitude and width of the pulse, which in \cite{Manakov82} additionally depend on the pulse number, while in our case such dependence is absent.

Moreover, in Manakov's case the field $\mathcal{E}(t, x)$ is real-valued, while in our case the field $\mathcal{E}(t, x)$ is complex-valued, as it should be for the envelope of an electromagnetic wave.

\end{rem}

\medskip
To formulate our second main result, we denote by $Z_b$ the set of zeros of $b(.)$ (the $b(k)$ is introduced in Section \ref{sect_scattering_data} below) in the half-plane $\Im k\geq 0.$ For simplicity, we make the following assumption on their mutual location:

\begin{Assumption}\label{assumption_2}
We assume that all zeros of $b(.)$ in $\Im k\geq0$ are simple, do not lie on the real line, and all their absolute values are pairwise distinct.
\end{Assumption}

\noindent We thus can parametrize $Z_b = \left\{k_{j}\right\}_{j=1}^{N},$ where $N\in\mathbb{N}\cup\left\{\infty\right\},$ $|k_{j}| < |k_{j+1}|, $ $\Im k_{j} > 0$ for all $j.$

\begin{thm}\label{thm_tail}
Let an initial pulse $\mathcal{E}_1(.)$ be a compactly supported \color{black} locally integrable function, \color{black} not identically equal to zero, let Assumption \ref{assumption_2} be satisfied and let $b(k)=\mathcal{O}(k^{-2})$ as $k\to\infty$. 
Let $\sigma\in(0, \frac12)$ be any fixed number, and let $\varepsilon>0$ be so small that for any $t, x$ satisfying $\sigma \leq \frac{x}{t} \leq 1-\sigma$ there might be at most one $k_j\in Z_b$ such that
\[\left|\frac{x}{t} - \frac{4|k_j|^2}{1 + 4|k_j|^2}\right|<\varepsilon.\]
Then in the limit 
$$\tau\equiv t-x \to \infty,$$ uniformly in $\sigma\leq \frac{x}{t} \leq 1 -\sigma,$ 
the asymptotics of the solution of the problem \eqref{MB1a}, \eqref{input_pulse}, \eqref{trivial} (which exists and is unique in view of Proposition \ref{prop_ibv} and \ref{sect_uniqueness} below) take the following form:
\begin{enumerate}
\item [I.] Away from solitons: let $\left| \frac{x}{t} - \frac{4|k_j|^2}{1 + 4|k_j|^2} \right| \geq \varepsilon$ for all $k_j\in Z_b.$ Then

\begin{align*}
\mathcal{E}(t, x) &= \frac{2\,k_0^{1/2}}{\tau^{1/2}}\(\sqrt{\nu_l}\, e^{i\,\omega_l(t, x)} + \sqrt{\nu_r}\, e^{i\,\omega_r(t, x)}\) + \mathcal{O}\(\tau^{-1}\),
\qquad
\mathcal{N}(t, x) = -1 + \mathcal{O}\(\tau^{-1}\),
\end{align*}
and
\begin{align*}
\rho(t, x) &= \frac{1}{\tau^{1/2}k_0^{1/2}}\(\sqrt{\nu_l}\,e^{i(\omega_l(t, x)+\frac{\pi}{2})} - \sqrt{\nu_r}\, e^{i(\omega_r(t, x)+\frac{\pi}{2})}\) + \mathcal{O}\(\tau^{-1}\),
\end{align*}
where $k_0 = k_0(t/x)$ is defined in \eqref{k_0}, 
$$ 
\nu_l = \nu_l(k_0) = \frac{1}{2\pi}\ln\(1 + \frac{1}{|r(-k_0)|^2}\),	\quad
\nu_r = \nu_r(k_0) = \frac{1}{2\pi}\ln\(1 + \frac{1}{|r(k_0)|^2}\),
$$
and (here, subscripts $l, r$ stand for `left', `right')
\begin{multline}\label{omega_l}
\omega_l(t, x) = 4\tau k_0 - \nu_l \ln(16\tau k_0) - \frac{1}{\pi}\int_{-k_0}^{k_0}\frac{\ln\frac{1 + |r(s)|^{-2}}{1 + |r(-k_0)|^{-2}}\, ds}{s+k_0} + \arg\(a(-k_0)b(-k_0)\) + \arg\Gamma(i\nu_l)
\\
+2\sum\limits_{|k_j| < k_0}\arg\frac{k_0 + \ol{k_j}}{k_0 + k_j}
- \frac{\pi}{4},
\end{multline}
\begin{multline}\label{omega_r}
\omega_r(t, x) = - 4 \tau k_0 + \nu_r \ln(16\tau k_0) - \frac{1}{\pi}\int_{-k_0}^{k_0}\frac{\ln\frac{1 + |r(s)|^{-2}}{1 + |r(k_0)|^{-2}}\, ds}{s-k_0} + \arg (a(k_0) b(k_0)) - \arg\Gamma(i\nu_r) 
\\
+2\sum\limits_{|k_j|<k_0}\arg\frac{k_0 - \ol{k_j}}{k_0-k_j}
+ \frac{\pi}{4}.
\end{multline}

\item [II.] Near the solitons: let $\left|\frac{x}{t} - \frac{4|k_j|^2}{1 + 4|k_j|^2}\right| < \varepsilon,$ for some $j.$
Then
\begin{multline*}
\mathcal{E}(t, x)
=
4B_j(t, x)
+
\frac{2\sqrt{k_0\,\nu_l}}{\tau^{1/2}}
\(
\(1 - \frac{i\,A_j(t, x)}{k_0 + k_j}\)^2 \cdot e^{i\,\omega_l(t, x)}
+
\frac{B_j(t, x)^2\, e^{-i\,\omega_l(t, x)}}{(k_0 + \ol{k_j})^2}
\)
\\
+
\frac{2\,\sqrt{k_0\,\nu_r}}{\tau^{1/2}}
\(
\(1 + \frac{i A_j(t, x)}{k_0 - k_j}\)^2 \cdot e^{i\,\omega_r(t, x)}
+
\frac{B_j(t, x)^2\,e^{-i\,\omega_r(t, x)}}{(k_0 - \ol{k_j})^2}
\)
+
\mathcal{O}(\tau^{-1}),
\end{multline*}
and
\begin{align*}
&\mathcal{N}(t, x) = -P(t, x) + Q(t, x)\ol{Y(t, x)} + \ol{Q(t, x)}Y(t, x) + \mathcal{O}(\tau^{-1}),
\\
&\rho(t, x) = Q(t, x) + 2Y(t, x)P(t, x) + 2X(t, x) Q(t, x) + \mathcal{O}(\tau^{-1}),
\end{align*}
where
\[
P(t, x) = 1 - \frac{2|B_j(t, x)|^2}{|k_j|^2},
\qquad
Q(t, x) = \frac{-2i B_j(t, x)}{\ol{k_j}}\(1 - \frac{i A_j(t, x)}{k_j}\),
\]
\begin{align*}
X(t, x) &= 
\frac{\sqrt{\nu_l}}{2\sqrt{k_0\,\tau}}
\(
\(1 + \frac{i\,A_j(t, x)}{k_0 + \ol{k_j}}\)\frac{B_j(t, x)\,e^{-i\,\omega_l(t, x)}}{k_0 + \ol{k_j}}
-
\(1 - \frac{i\,A_j(t, x)}{k_0 + k_j}\)\frac{\ol{B_j(t, x)}\,e^{i\,\omega_l(t, x)}}{k_0+k_j}
\)
\\
&+ 
\frac{\sqrt{\nu_r}}{2\sqrt{k_0\,\tau}}
\(
\(1 - \frac{i\,A_j(t, x)}{k_0 - \ol{k_j}}\)\frac{B_j(t, x)\,e^{-i\,\omega_r(t, x)}}{k_0 - \ol{k_j}}
-
\(1 + \frac{i\,A_j(t, x)}{k_0 - k_j}\)\frac{\ol{B_j(t, x)}\,e^{i\,\omega_r(t, x)}}{k_0-k_j}\right)
\end{align*}
and
\begin{align*}
Y(t, x) = \frac{i\,\sqrt{\nu_l}}{2\,\sqrt{k_0\,\tau}}
\(
e^{i\,\omega_l(t, x)}\(1 - \frac{i\,A_j(t, x)}{k_0 + k_j}\)^2
+
\frac{B_j(t, x)^2\,e^{-i\,\omega_l(t, x)}}{(k_0 + \ol{k_j})^2}
\)
\\
-\frac{i\,\sqrt{\nu_r}}{2\,\sqrt{k_0\,\tau}}
\(
e^{i\,\omega_r(t, x)}\(1 + \frac{i\,A_j(t, x)}{k_0 - k_j}\)^2
+
\frac{B_j(t, x)^2\,e^{-i\,\omega_r(t, x)}}{(k_0 - \ol{k_j})^2}
\).
\end{align*}
Here 
\[
A_j = A_j(t, x) = \frac{|\w_j|^2\cdot 2\Im k_j}{1 + |\w_j|^2},
\quad 
B_j = B_j(t, x) = \frac{-\ol{\w_j}\cdot 2\Im k_j}{1 + |\w_j|^2},
\]
and $\w_j = |\w_j|e^{i\arg \w_j},$ where (below, $\dot b(k_j)$ is the derivative of $b(.)$ at the point $k_j$)
\begin{multline*}
|\w_j| = \frac{1}
{2\Im k_j\cdot |a(k_j)\dot{b}(k_j)|}
\exp\left\{-2\Im k_j\(t-x-\frac{x}{4\left[(\Re k_j)^2 + (\Im k_j)^2\right]}\)\right\}
\cdot\\\cdot
\exp\left[
\frac{-\Im k_j}{\pi}\int_{-k_0}^{k_0}\frac{\ln\(1 + |r(s)|^{-2}\)\, ds}
{(s-\Re k_j)^2 + (\Im k_j)^2}\right]
\prod\limits_{p: |k_p|<|k_j|}\left|\frac{k_j - k_p}{k_j - \ol{k_p}}\right|^2\,,
\end{multline*}
and
\begin{multline*}
\arg \w_j = -\arg\(a(k_j)\dot{b}(k_j)\) + 2\Re k_j\cdot\(t-x+\frac{x}{4[(\Re k_j)^2 + (\Im k_j)^2]}\)
+
\\
+\frac{1}{\pi}\int_{-k_0}^{k_0}
\frac{(s - \Re k_j)\ln\(1 + |r(s)|^{-2}\)\, ds}{(s-\Re k_j)^2 + (\Im k_j)^2}
+
2\sum\limits_{p: |k_p|<|k_j|}\arg\(\frac{k_j - k_p}{k_j - \ol{k_p}}\).
\end{multline*}
\end{enumerate}
\end{thm}
\begin{rem}
Note that in previous studies (cf. \cite{GZM85}) it was believed that in the unstable medium, zeros of $a(.)$ generate solitons whose speed is higher than the speed of light, and thus are physically impossible. However, Theorem \ref{thm_tail} states that zeros of $a(.)$ do not contribute in any way to the asymptotics. On the contrary, it is zeros of $b(.)$ that generate solitons of the problem, and these solitons are physically relevant, i.e. they have speeds less than the speed of light.
\end{rem}
\begin{rem}\label{rem_1.3}
Comparing the formulae of the paper \cite{MN86} and Theorem \ref{thm_tail}, 
we see that they agree qualitatively, but do not agree quantitatively. For instance, in the unnumbered formula after formula (3.14) in \cite{MN86}, the amplitude of $\mathcal{E}$ is proportional to $x^{1/4}(t-x)^{-3/4},$
while in our formula it is proportional to $\ln(\frac{t}{x} - 1)x^{1/4}(t-x)^{-3/4}.$

Besides, our approach shows the presence of a new phenomenon in the theory of solitons, namely, the boundary condition (input pulse) of a mixed problem under trivial initial conditions can generate solitons due to zeros of the reflection coefficient, while the poles of the transmission coefficient do not contribute to the asymptotics of the solution. 
\end{rem}

\begin{rem}\label{rem_matching_initial_boundary_values}
Note that the condition $m\geq2$ from Assumptions \ref{assumption_1+}, \ref{assumption_1} is used in the proof of Theorem \ref{thm_really_close_light_cone} to ensure that the left-hand side in \eqref{estimates_J} belongs to the class $L_1.$

Note also that if an initial pulse $\mathcal{E}_1(t)$ has a limit $\mathcal{E}_1(0)$ as $t\to+0$, and Assumption \ref{assumption_1+} or \ref{assumption_1} with $m>1$ is satisfied, then $\mathcal{E}_1(0)=0$ and thus the initial and boundary conditions for $\mathcal{E}(t,x)$ match at $x=t=0.$ 
Indeed, 
if $\mathcal{E}_1(0)\neq0,$ then reflection coefficient $r(k)$ vanishes not faster than $1/k$ as $k\to\infty$, as follows from the integral representations for $a(k)$ and $b(k)$ (see Remark 2.2) by using Riemann-Lebesgues lemma and the known formula:
$-4K_{12}(t,t)\vert_{t=0}=\mathcal{E}_1(t)\vert_{t=0}\neq 0$.

\end{rem}

\medskip
The paper is organised as follows. In Sections \ref{sect_scattering_data}, \ref{sect_matrix_RH}
we give some preliminary information about the MB equations and the corresponding RH problem. Most of the material there follows \cite{FKM17}. In addition, in Section \ref{sect_matrix_RH} we prove the causality principle (Theorem \ref{thm_caus}).
Section \ref{sect_thm_1} is devoted to the proof of Theorem \ref{thm_really_close_light_cone}, and Section \ref{sect_thm_2} is devoted to the proof of Theorem \ref{thm_tail}.

\medskip\noindent
\paragraph{Notations.} Throughout the paper, we use the following notation. For a function $f(t, x; k)$ that depends on the real variables $t, x\in\mathbb{R}$ and a complex variable $k\in\mathbb{C},$ we denote
\[
f^*(t, x; k) := \ol{f(t, x; \ol k)},
\]
where the bar denotes the complex conjugate.

\section{Basic solutions of the Ablowitz-Kaup-Newell-Segur linear equations}\label{sect_scattering_data}

In this preparatory section only, we assume the existence of the solution of the initial-boundary value problem \eqref{MB1a}, \eqref{IBC}, and then derive a meaningful Riemann-Hilbert problem, which is fully determined by the initial and boundary values.

Then in the next Section \ref{sect_matrix_RH}, we drop the assumption of the existence of the solution of the initial-boundary problem, and instead start directly from the RH problem, which is fully determined by the initial and boundary conditions. Based on that, we then prove the existence of a solution to the initial-boundary value problem.

Most of the constructions of this Section \ref{sect_scattering_data} and the next Section \ref{sect_matrix_RH} are taken from \cite{FKM17}, and we sketch them for the convenience of the reader.

\subsection*{Lax-pair representation of the MB equations.}
The Ablowitz-Kaup-Newell-Segur (AKNS)  equations for the  Maxwell-Bloch   equations without spectral broadening have the form \cite{AKN, AS, GZM85}:
\begin{align}\label{teq}
\Phi_t=&U(t,x;k)\Phi,  &U(t,x; k)&=-(\rmi k\sigma_{3}+H(t,x)),\\
\Phi_x=&V(t,x;k)\Phi,   &V(t,x;k)&= \rmi k\sigma_3+H(t,x)+\frac{\rmi F(t,x)}{4k},\label{xeq}
\end{align}
where $\Phi = \Phi(t, x; k),$
$\sigma_{3}=\begin{pmatrix}
		1&0\\ 0&-1
	\end{pmatrix}$, and
\[
H(t,x)=\frac{1}{2} \begin{pmatrix}
0&{\cal E}(t,x)\\ - \ol{{\cal E}(t,x)} & 0
\end{pmatrix},
	\quad
F(t,x)=\begin{pmatrix}
{\cal N}(t,x)&\rho(t,x)\\
\ol{\rho(t,x)} & -{\cal N}(t,x )
\end{pmatrix}.
	\]
It is  well known \cite{AS} that the overdetermined system of differential equations \eqref{teq}, \eqref{xeq}  is compatible if and only if the compatibility condition
\begin{equation}\label{UV}
		U_x(t, x; k)-V_t(t, x; k)+[U(t, x; k), V(t, x; k)]=0
\end{equation}
holds (here, $[U,V] = UV - VU$ is the matrix commutator). It is equivalent to the system of nonlinear equations
\begin{align}\label{HF}
\frac{\partial H(t, x)}{\partial t}+\frac{\partial H(t, x)}{\partial x}=\frac{1}{4}[\sigma_3, F(t, x)], \qquad \frac{\partial F(t, x)}{\partial t}=[F(t, x), H(t, x)],
	\end{align}
which are the matrix form of the MB equations \eqref{MB1a}.

\subsection*{Jost solutions.}

We suppose here that the solution (${\cal E}(t,x )$, ${\cal N}(t,x)$,  $\rho(t,x)$) of the mixed problem \eqref{MB1a}, \eqref{IBC} for the Maxwell-Bloch equations in the domain $t\in\mathbb{R}_+$, $0\le x\le l\le\infty$ does exist, unique, smooth and tends to its limits fast enough for large $x$ and large $t.$

Then the AKNS linear equations \eqref{teq} and \eqref{xeq}  are compatible.

Define solutions $W(t, x; k), \Phi(t; k), \Psi(t, x; k), w(x; k), Y(t, x; k), Z(t, x; k)$ of the $x$- or (and) $t$-equations \eqref{teq}, \eqref{xeq} as follows.

Let $W(t,x;k)$ satisfy the $x$-equation \eqref{xeq} (for all $t$) together with the initial condition $W(t,0; k)=I$, and let $\Phi(t; k)$ satisfy the $t$-equation \eqref{teq} for $x=0$ under the initial condition $$\lim\limits_{t\to\infty}\Phi(t; k)e^{\ii k t\sigma_3}=I.$$

Let $\Psi(t,x;k)$ be the solution of the $t$-equation \eqref{teq} (for all $x$), which also satisfies the initial condition $\Psi(0,x; k)=I$, and let $w(x; k)$ satisfy the $x$-equation \eqref{xeq} for $t=0,$ with the following initial condition:
$$\lim\limits_{x\to l-} w(x; k)e^{-\ii x\mu( k)\sigma_3}=I,$$
where 
$$\mu( k)=k+\frac{1}{4k}\,.$$
(If $l<\infty$, then this is equivalent to $w(l; k)=e^{\ii l\mu( k)\sigma_3}$,  and if $l=\infty$, the initial condition takes the form  $\lim\limits_{x\to\infty} w(x; k)e^{-\ii x\mu( k)\sigma_3}=I$).

\noindent Next, define the {\it Jost solutions} $Y(t,x;k)$ and $Z(t,x;k)$ as the matrix products
\begin{equation}\label{YZ}
Y(t,x;k)=W(t,x;k)\Phi(t; k),\qquad Z(t,x;k)=\Psi(t,x;k)w(x; k).
\end{equation}

Note that automatically $Y$ satisfies the $x$-equation \eqref{xeq}, and $Z$ satisfies the $t$-equation \eqref{teq}.
It is a direct consequence of the following Lemma \ref{lem 2.1} that in fact the functions $Y(t, x; k),$ $Z(t, x; k)$ satisfy both the $t$- and $x$-equations \eqref{teq}, \eqref{xeq}.

\begin{lemma}\label{lem 2.1}[\cite{BK00}, Lemma 2.1]
Let equations \eqref{teq} and \eqref{xeq} be compatible for all $t,x,k\in\mathbb{R}$. Let $\mathcal{F}(t,x;k)$ be a matrix satisfying the $t$-equation \eqref{teq} for all $x$ (the $x$-equation \eqref{xeq} for all $t$). Assume that $\mathcal{F}(t_0,x; k)$ satisfies the $x$-equation \eqref{xeq} for some $t=t_0\le\infty$ (the $t$-equation \eqref{teq} for some $x=x_0\le\infty$). Then $\mathcal{F}(t,x; k)$ satisfies the $x$-equation \eqref{xeq} for all $t$ (satisfies the $t$-equation \eqref{teq} for all $x$).
\end{lemma}
The proof can be found, for example, in \cite{BK00} (Lemma 2.1).

\subsection*{Properties of the Jost solutions.}\label{sect_Jost}
To formulate the properties of the Jost solutions $Y, Z,$ let us introduce the notations: 
\begin{equation}\label{Sigma}\Omega_\pm = \{k\!\in\! \mathbb{C}_\pm :\, |k|>1/2\},
\qquad
D_\pm = \{k\!\in\! \mathbb{C}_\pm :\, |k| < 1/2\},
\qquad
\Sigma = \mathbb{R}\cup C_{u}(1/2)\cup
C_{d}(1/2),\end{equation}
where 
$\mathbb{C}_{\pm} = \left\{k\in\mathbb{C}:\ \pm\Im k > 0\right\}$ and where
$C_{u}(1/2)$ and $C_{d}(1/2)$ are the semicircles
$$C_{u}(1/2)=\{k\in\mathbb{C} :\, |k| =1/2,\, \arg k\! \in\! (\pi, 0)\}, \qquad C_{d}(1/2)=\{k\!\in\! \mathbb{C} :\, |k|=1/2,\, \arg k\!\in\! (-\pi, 0)\}$$ (the subscripts $u, d$ stand for `up' and `down'). 
The orientation on the contour $\Sigma$ is from the left to the right on the real line $\mathbb{R} $ and on the half-circles $C_{u}(1/2)$, $C_{d}(1/2)$  and is depicted in Figure \eqref{domains}.
Let $\overline\Omega _\pm$ and $\overline D_\pm$ be the closures of the domains $\Omega _\pm$ and  $D_\pm$, respectively. Note that the contour $\Sigma$ is the set where $\Im\mu(k)=0$:
$$
\Sigma=\left\{k\in\mathbb{C}: \,\Im\left(k+\frac{1}{4k}\right)=0\right\}=\mathbb{R}\cup C_{u}(1/2)\cup C_{d}(1/2).
$$

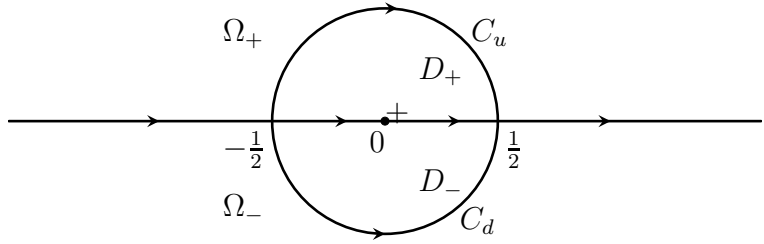
\begin{figure}
	\centering
	\begin{tikzpicture}[line cap=round,line join=round,>=triangle 45,x=1.0cm,y=1.0cm] 
		\clip(-5,-2.5) rectangle (5, 2.5);
		\draw [line width=1pt, decoration = {markings, mark=at position 0.2 with {\arrow{stealth}}}, 
decoration = {markings, mark=at position 0.45 with {\arrow{stealth}}}, 
decoration = {markings, mark=at position 0.6 with {\arrow{stealth}}}, 
decoration = {markings, mark=at position 0.8 with {\arrow{stealth}}}, postaction={decorate}] (-5,0) -- (5,0);
\draw [line width=1 pt, decoration = {markings, mark=at position 0.25 with {\arrow{stealth reversed}}}, 
decoration = {markings, mark=at position 0.75 with {\arrow{stealth}}}, postaction={decorate}] (0,0) circle (1.5cm);
\draw (1,1.5) node[anchor=north west] {$C_{u}$};
\draw (0.85,-1) node[anchor=north west] {$C_{d}$};
		\draw [fill=black] (-0,0) circle (1.5pt);
		\draw (1.45,0.01) node[anchor=north west] {$\frac{1}{2}$};
		\draw (-0.35, 0.01) node[anchor=north west] {$0$};
		\draw (-2.3, 0.01) node[anchor=north west] {$-\frac{1}{2}$};

		\draw (0.3,1) node[anchor=north west] {$D_+$};
		\draw (0.3,-0.5) node[anchor=north west] {$D_-$};

\draw (-2.3, 1.5) node[anchor=north west] {$\Omega_+$};
\draw (-2.3, -1.5) node[anchor=south west] {$\Omega_-$};

\put(100, 50){$\Omega_+$}
\put(100,-50){$\Omega_-$};
\unitlength 1mm
\put(-30,0){$+$}  \put(0,1){$+$}   \put(28,1){$+$}   \put(-30,-3){$-$}  \put(0,-3){$-$}   \put(28,-3){$-$}
\put(-6,15){$+$}   \put(-6,-17){$-$}   \put(-6,11){$-$}   \put(-6,-13){$+$}

\end{tikzpicture}
\caption{The domains $\Omega_\pm$, $D_\pm$ and the oriented contour $\Sigma=\mathbb{R}\cup C_{u}\cup C_{d}$.} \label{domains}
	\label{fig1}
\end{figure}

\noindent
The function $Y(t,x;k)=(Y_{[1]}(t,x;k), Y_{[2]}(t,x;k)),$ defined in \eqref{YZ}, has the following properties (\cite{FKM17}):
\begin{lemma}[\cite{FKM17}]\label{lem_Y_prop}
\begin{enumerate}[1)]
\item
$Y(t,x;k)$ ($ k\neq0$) satisfies the $t$- and $x$-equations \eqref{teq},~\eqref{xeq};
	
\item
$Y(t,x;k)=\Lambda  Y^*(t,x;k) \Lambda^{-1}$,\; $ k\in\mathbb{R}\setminus\{0\}$, where $\Lambda= \begin{pmatrix} 0&1\\ -1&0 \end{pmatrix}$;

\item $\det Y(t,x;k) \equiv1, \quad  k\in\mathbb{R}\setminus\{0\}$;

\item the map $(t, x)\longmapsto Y(t,x;k)$ ($ k\neq0$) is smooth in $t$ and $x$;

\item the maps $k\longmapsto Y_{[1]}(t,x;k)e^{\ii kt-\ii \mu(k)x}$, $k\longmapsto Y_{[1]}(t,x;k)e^{-\ii kt+\ii \mu(k)x}$, $k\longmapsto Y_{[2]}(t,x;k)e^{-\ii kt+\ii \mu(k)x}$, $k\longmapsto Y_{[2]}(t,x;k)e^{\ii kt-\ii \mu(k)x}$  are analytic in $\Omega_-$, $D_-$, $\Omega_+$, $D_+,$ respectively;

\item the vector functions $Y_{[1]}(t,x;k)e^{\ii kt-\ii \mu(k)x}$,\, $Y_{[1]}(t,x;k)e^{-\ii kt+\ii \mu(k)x}$ and \linebreak
$Y_{[2]}(t,x;k)e^{-\ii kt+\ii \mu(k)x}$,\, $Y_{[2]}(t,x;k)e^{\ii kt-\ii \mu(k)x}$ are analytic in $\mathbb{C}_-$  and $\mathbb{C}_+,$ respectively, continuous up to the boundary with exception of $k=0$ and have the following asymptotic behaviour:
\begin{align*}
&Y_{[1]}(t,x;k)e^{i\theta(t, x;k)}= \begin{pmatrix} 1\\ 0\end{pmatrix}+ \ord(k^{-1}),\quad &k\in\overline \Omega_-,\; k\to\infty,&\\
&Y_{[1]}(t,x;k)e^{-i\theta(t, x;k)}= \ord(1) +\ord(k),\quad &k\in\overline  D_-\setminus\{0\},\; k\to 0,&
\end{align*}
\begin{align*}
&Y_{[2]}(t,x;k)e^{-i\theta(t, x;k)}= \begin{pmatrix} 0\\ 1\end{pmatrix}+ \ord(k^{-1}),\quad &k\in\overline  \Omega_+,\; k\to\infty,&\\
&Y_{[2]}(t,x;k)e^{i\theta(t, x;k)}= \ord(1) +\ord(k),\quad &k\in\overline  D_+\setminus\{0\},\; k\to 0,&
\end{align*}
where
\begin{equation}\label{theta}
\theta(t,x;k) = kt-\mu(k)x = (t-x)k-\frac{x}{4k}.
\end{equation}
\end{enumerate}
\end{lemma}
	
\noindent The function $Z(t,x;k)=(Z_{[1]}(t,x;k), Z_{[2]}(t,x;k))$, defined in \eqref{YZ}, has the following properties (\cite{FKM17}):
\begin{lemma}[\cite{FKM17}]\label{lem_Z_prop}
\begin{enumerate}[1)]
\item $Z(t,x;k)$ ($ k\neq 0$) satisfies the $t$- and $x$-equations \eqref{teq},~\eqref{xeq};
	
\item $Z(t,x;k)=\Lambda Z^*(t,x;k)\Lambda^{-1}$,\; $ k\in\mathbb{R}\setminus\{0\}$;

\item $\det Z(t,x;k) \equiv 1,\;  k\in\mathbb{R}\setminus\{0\}$;

\item the map $(t, x)\longmapsto Z(t,x;k)$ ($ k\neq0$) is smooth in $t$ and $x$;

\item the maps $k\longmapsto Z_{[1]}(t,x;k)$ and $k\longmapsto Z_{[2]}(t,x;k)$ are analytic in $\Omega_+\cup D_-$ and $\Omega_-\cup D_+,$ respectively, and the asymptotic behaviour of $Z_{[1]}(t,x;k)e^{\ii kt-\ii x\mu(k)}$, $Z_{[2]}(t,x;k)e^{-\ii kt+\ii x\mu(k)}$ is as follows:
\begin{align*}
Z_{[1]}(t,x;k)e^{i\theta(t, x;k)}=& \begin{pmatrix}1\\0\end{pmatrix}+ \ord(k^{-1}),\quad &k\in \overline\Omega_+,\; k\to\infty,\\
Z_{[1]}(t,x;k)e^{i\theta(t, x;k)} =& \ord(1)+ \ord(k),\quad &k\in\overline D_-\setminus\{0\}\; k\to 0,
\end{align*}
\begin{align*}
Z_{[2]}(t,x;k)e^{-i\theta(t, x;k)}=& \begin{pmatrix}0\\1\end{pmatrix}+ \ord(k^{-1}),\quad &k\in \overline\Omega_-,\; k\to\infty,\\
Z_{[2]}(t,x;k)e^{-i\theta(t, x;k)}=& \ord(1)+ \ord(k),\quad &k\in\overline D_+\setminus\{0\},\; k\to 0.
\end{align*}
\end{enumerate}
\end{lemma}

\begin{rem}\label{remZ}
In the case of trivial initial conditions \eqref{trivial} we have $w(x,k) = e^{i x (k + 1/(4k))\sigma_3}$ and hence the matrix $Z(t,x;k) = (Z_{[1]}(t,x;k),  Z_{[2]}(t,x;k))$ is analytic in $k\in \mathbb{C}\setminus\{0\}.$
\end{rem}

\subsection*{Spectral coefficients $a(k), b(k), A(k), B(k), \alpha(k), \beta(k).$}\label{sect_ab}

Since the matrices  $Y(t,x;k)$ and $Z(t,x;k)$  are solutions of the $t$- and $x$-equations \eqref{teq}, \eqref{xeq}, they are linearly dependent. Consequently, there exists a transition matrix $T( k)$, independent of $t$ and $x$, such that
\begin{equation}\label{sc}
Y(t,x;k)=Z(t,x;k) T( k).
\end{equation}
The transition matrix is equal to
\[
T( k)=Z^{-1}(0,0; k)Y(0,0; k)=w^{-1}(0; k)\Phi(0; k),
\]
and, hence,  $T( k)=\Lambda T^*( k)\Lambda^{-1}$, $ k\in\mathbb{R}\setminus\{0\}$, i.e. $T( k)$ has the form
\[
T( k)=\begin{pmatrix}
		 a^*( k) & b( k) \\
		-b^*( k) & a( k)
	\end{pmatrix}.
\]
The scattering relation \eqref{sc} can be written in the form
\begin{equation}\label{ZYrelations}
\begin{split}
&Y_{[1]}(t,x;k)= a^*( k)Z_{[1]}(t,x;k)- b^*(k)Z_{[2]}(t,x;k),\quad  k \in \Sigma\setminus\{0\},\\
&Y_{[2]}(t,x;k)=a( k)Z_{[2]}(t,x;k)+b( k)Z_{[1]}(t,x;k),\quad  k \in \Sigma\setminus\{0\},
\end{split}.
\end{equation}
where $\Sigma = \mathbb{R}\cup C_{u}\cup C_{d}$
was introduced above in Fig.\ref{domains}. 
From these relations we obtain that
\begin{align*}	&a( k)=\det(Z_{[1]}(t,x;k),Y_{[2]}(t,x;k)),\quad
a^*( k)=\det(Y_{[1]}(t,x;k), Z_{[2]}(t,x;k)),\\
&b( k)=\det(Y_{[2]}(t,x;k), Z_{[2]}(t,x;k)), \quad b^*( k)=\det(Y_{[1]}(t,x;k), Z_{[1]}(t,x;k)).
\end{align*}
	
\noindent 
To study the properties of $a(k), b(k)$, it is convenient to introduce the matrix
\begin{equation}\label{Phi0}
\Phi(0; k)=\begin{pmatrix}
A^*(k) & B( k)\\
-B^*(k) & A( k) \end{pmatrix},
\end{equation}
which is determined by the boundary condition $\mathcal{E}_1(t) = {\cal E}(t, 0)$
(here,  the functions $A(k), B(k)$ are called the {\it spectral functions of the $t$-equation for $x=0$}),
and the matrix
\[
w(0; k)=\begin{pmatrix}
\alpha( k) & -\beta^*( k)\\
\beta( k) & \alpha^*( k) \end{pmatrix},
\]
which is determined by the initial functions ${\cal E}(0,x)$, $\rho(0,x)$ and ${\cal N}(0,x)$
(here,  the functions $\alpha, \beta$ are called the {\it spectral functions of the $x$-equation for $t=0$}).

The functions $\alpha(k)$, $\beta(k)$ and  $ \alpha^*(k)$, $\beta^*(k)$ can be extended analytically in $\Omega_+\cup D_-$ and $\Omega_-\cup D_+$, respectively, the functions $A(k)$, $B(k)$ and
$A^*(k)$, $B^*(k)$ can be extended analytically in $\mathbb{C}_+$ and $\mathbb{C}_-$, respectively.
They have the following asymptotic behaviour:
\begin{align*}
&\alpha(k)=1+\ord(k^{-1}), && \beta(k)=\ord(k^{-1}), &&  k\to\infty,  \quad k\in\overline\Omega_+;\\
&\alpha(k)=\ord(1), && \beta(k)=\ord(1), && k\to0, \quad k\in \overline D_-;\\
&\alpha^*(k)=1+\ord(k^{-1}), && \beta^*(k)=\ord(k^{-1}), && k\to\infty, \quad k\in\overline\Omega_-;\\
&\alpha^*(k)=\ord(1), && \beta^*(k)=\ord(1), && k\to0, \quad k\in \overline D_+;
\end{align*}
\begin{align*}
&A(k)=1+\ord(k^{-1}), &&B(k)=\ord(k^{-1}),&  & k\to\infty, \quad k\in\mathbb{C}_+;\\
&A^*(k)=1+\ord(k^{-1}), &&B^*(k)=\ord(k^{-1}), & & k\to\infty, \quad k\in\mathbb{C}_-;\\
&A(k)=\ord(1), &&B(k)=\ord(1),&  &k\to0, \quad k\in\mathbb{C}_+;\\
&A^*(k)=\ord(1), && B^*(k)=\ord(1),&  &k\to0, \quad k\in\mathbb{C}_-.
\end{align*}

\noindent The entries of the transition matrix $T( k)$ in the domains of their analyticity are equal to
\begin{center}
$a(k)= \alpha(k)A(k)-\beta(k)B(k)$,\;\; $k\in\Omega_+$;\quad
$b(k)= \alpha^*(k)B(k)+\beta^*(k)A(k)$,\;\; $k\in D_+$;
		
$a^*(k)= \alpha^*(k) A^*(k) - \beta^*(k) B^*(k)$,\;\; $k\in\Omega_-$;\quad
$b^*(k)= \alpha(k) B^*(k)+ \beta(k) A^*(k)$,\;\; $k\in D_-$\,.
\end{center}
The spectral functions $a(k)$ and $b(k)$ are defined and smooth for $ k \in \Sigma\setminus\{0\}$. The matrix $T$ is unimodular, $T( k)\equiv1$  and, hence,
\begin{equation}\label{aabb}a( k) a^*(k) + b(k) b^*(k) \equiv1.
\end{equation}
The spectral functions have the following asymptotics:
\begin{center}
$a(k)=1+\ord(k^{-1})$ as $k\to\infty$, $k\in\overline\Omega_+,
\qquad
b(k)=\ord(1)$ as $k\to 0$, $k\in \overline D_+$;\;
		
$a^*(k)=1+\ord(k^{-1})$ as $k\to\infty$, $k\in\overline\Omega_-,
\qquad
b^*(k)= \ord(1)$ as $k\to 0$, $k\in\overline D_-.$
	\end{center}

\begin{rem}\label{rem_prop_abr_trivial}
In the case of trivial initial data \eqref{trivial} the ibv problem is only defined by the input pulse which in turn determines the spectral functions in the form \cite{FKM17}: $\alpha(k)\equiv1, \beta(k)\equiv0,$
\begin{align*}
a(k)=A(k)=1+\int_0^T K_{22}(0,\tau)\ee^{\rmi k \tau }d\tau, \qquad  & b(k)=B(k)=\int_0^T K_{12}(0,\tau)\ee^{\rmi k \tau }d\tau,
\end{align*}
where $K_{lm}(t, \tau)$ ($l,m=1,2$)  are the entries of a transformation operator \cite{FKM17}
$$ \Phi(t; k)=e^{-\ii k t\sigma_3}+\int_t^T K(t,\tau)e^{-\ii k\tau\sigma_3}d\tau, \qquad  k\in\mathbb{R}, \qquad   T\le\infty, 
$$
which satisfy $\int_{t}^T|K_{lm}(t, \tau)|d\tau<\infty.$ 
Here $T\leq\infty$ is the supremum of the support of $\mathcal{E}_{1}(t),$ i.e.
$\mathcal{E}_1(t) = 0$ for $t\in(T,\infty).$
These formulae show that $a(k)$ and $b(k)$ admit analytic continuation in the domain $\Im k\geq0$ in the case $T=\infty$, and even are entire functions in the case $T<\infty$,
with $b^*(.)$ satisfying the estimate
\begin{equation}\label{b_estimate}
b^*(k) = \ord(k^{-1}e^{T\Im k}), \quad k\to\infty,\ \Im k\geq 0.
\end{equation}
In either case, the function $r(k)=\frac{b(k)}{a(k)}$
admits an analytic continuation from the real axis.

Note also that for trivial initial data \eqref{trivial}, the functions $b(.), r(.)$ do not vanish identically provided that $\mathcal{E}_1(t)$ does not vanish identically.
\end{rem}

\subsection*{Zeros of $a, b.$}

\begin{lemma}\label{lem_zeros_ab}
Let $\left\{k_j^{_{(a)}}\right\}_j$ be the set of zeros of $a(.)$ in $\Omega_+,$ and $n_j^{_{(a)}}$ be the multiplicity of the zeros, i.e. $a(k_j^{_{(a)}}) = \ldots = a^{(n_j^{_{(a)}}-1)}(k_j^{_{(a)}}) = 0,$ $a^{(n_j^{_{(a)}})}(k_j^{_{(a)}})\neq 0.$ 
\\Similarly, let
 $\left\{k_j^{_{(b)}}\right\}_j$ be the set of zeros of $b(.)$ in $D_+,$ and $n_j^{_{(b)}}$ be the multiplicity of the zeros, i.e. $b(k_j^{_{(b)}}) = \ldots = b^{(n_j^{_{(b)}}-1)}(k_j^{_{(b)}}) = 0,$ $b^{(n_j^{_{(b)}})}(k_j^{_{(b)}})\neq 0.$
 Then
\begin{enumerate}[I.]
\item
for each zero $k_j^{_{(a)}}$ of $a(.)$ there exist constants $\mu_{0,j}^{_{(a)}}, \ldots, \mu_{n_j^{_{(a)}}-1, j}^{_{(a)}}\in\mathbb{C}$ (independent of $t, x$) such that
\[
\frac{d^p}{dk^p}Y_{[2]}(t, x; k)\Big|_{k_j^{_{(a)}}} = \sum\limits_{q=0}^{p}\binom{p}{q}\mu_{p-q,j}^{_{(a)}}\frac{d^{q}}{d k^q}Z_{[1]}(t, x; k)\Big|_{k = k_j^{_{(a)}}},
\qquad p = 0, \ldots, n_j^{_{(a)}} - 1,
\]
where $\binom{p}{q}=\frac{p!}{q!(p-q)!}$ is the binomial coefficient, 
and\\
for each zero $k_j^{_{(b)}}$ of $b(.)$ there exist constants $\mu_{0,j}^{_{(b)}}, \ldots, \mu_{n_j^{_{(b)}}-1, j}^{_{(b)}}\in\mathbb{C}$ (independent of $t, x$) such that
\[
\frac{d^p}{dk^p}Y_{[2]}(t, x; k)\Big|_{k_j^{_{(b)}}} = \sum\limits_{q=0}^{p}
\binom{p}{q}\mu_{p-q,j}^{_{(b)}}\frac{d^{q}}{d k^q}Z_{[2]}(t, x; k)\Big|_{k = k_j^{_{(b)}}},
\qquad p = 0, \ldots, n_j^{_{(b)}} - 1.
\]

\item
Let $k_j^{_{(a)}}$ (respectively $k_j^{_{(b)}}$) be a zero of $a(.)$ (resp. $b(.)$) and 
\\
let $T_{1,j}^{_{(a)}}, \ldots, T_{n_j^{_{(a)}},j}^{_{(a)}}$ (resp. $T_{p, j}^{_{(b)}}, p=1, \ldots, n_j^{_{(b)}}$) be the coefficients in the Taylor expansion of the function 
$
\frac{1}{a(k)}
\mbox{(resp. $\frac{1}{b(k)}$)}
$
at the point $k=k_j^{_{(a)}}$ ($k=k_j^{_{(b)}}$),
i.e.
\\
$
\frac{1}{a(k)}
=
\sum_{q=1}^{n_j^{_{_{(a)}}}}
\frac{T_{q,j}^{_{(a)}}}{\(k-k_{j}^{_{(a)}}\)^{_q}}
+
\mathcal{O}(1),\quad k\to k_j^{_{(a)}}
\hfill\qquad
\mbox{(resp. }
\frac{1}{b(k)}
=
\sum_{q=1}^{n_j^{_{(b)}}}
\frac{T_{q,j}^{_{(b)}}}{\(k-k_{j}^{_{(b)}}\)^{_q}}
+
\mathcal{O}(1),\quad k\to k_j^{_{(b)}}
$).
\\
Define the constants
\\
$
A^{_{(a)}}_{p, j} = \sum_{q=0}^{n_j^{_{(a)}}-p}T^{_{(a)}}_{q+p, j} \frac{\mu_{q, j}^{_{(a)}}}{q!},
\ p = 1, \ldots, n_j^{_{(a)}}$

$\hfill
(\mbox{resp.} A^{_{(b)}}_{p, j} = \sum_{q=0}^{n_j^{_{(b)}}-p}T^{_{(b)}}_{q+p, j} \frac{\mu_{q, j}^{_{(b)}}}{q!},
\
p = 1, \ldots, n_j^{_{(b)}}).
$

Then\\
\[
\frac{1}{a(k)}Y_{[2]}(t, x; k) e^{i\theta(t, x; k)} 
-
\left[
\sum_{q=1}^{n^{_{(a)}}_j}
\frac{A_{q,j}^{_{(a)}}\,e^{2i\theta(t,x;k)}}{\(k-k_j^{_{(a)}}\)^q}
\right]
Z_{[1]}(t, x; k) e^{-i\theta(t, x; k)} 
=
\ord(1),\quad k\to k_j^{_{(a)}}
\]
(resp.
$
\frac{1}{b(k)}Y_{[2]}(t, x; k) e^{i\theta(t, x; k)} 
- 
\left[
\sum_{q=1}^{n^{_{(b)}}_j}
\frac{A_{q,j}^{_{(b)}}\,e^{2i\theta(t,x;k)}}{\(k-k_j^{_{(b)}}\)^{_q}}
\right]
Z_{[2]}(t, x; k) e^{-i\theta(t, x; k)} 
=
\ord(1),\ k\to k_j^{_{(b)}}
$).
\end{enumerate}
\end{lemma}
\begin{proof}
{\it I.} Indeed, $a(k) = \det[Z_{[1]}(t,x;k), Y_{[2]}(t,x;k)]$ and $b(k) = \det[Y_{[2]}(t,x;k), Z_{[2]}(t,x;k)],$
and hence the argument from \cite[Lemma D.1]{GM2020} is readily applied. The fact that $\mu_{q,j}^{_{(a)}}$, $\mu_{q,j}^{_{(b)}}$ are independent of $t, x$ follows from the fact that $Z_{[1]}, Z_{[2]}, Y_{[2]}$ satisfy the $t$- and $x$-equations~\eqref{teq},~\eqref{xeq}.

{\it II.} The proof is a tedious multiplication of Taylor expansions of $Y_{[1]}, Z_{[2]}$ in a manner similar to \cite[Lemma D.2]{GM2020} and we omit it.

(Note however that in \cite[Lemma D.2]{GM2020} instead of the sum 
$\sum_{q=1}^{n^{_{(b)}}_j}
\frac{A_{q,j}^{_{(b)}}\,e^{2i\theta(t,x;k)}}{\(k-k_j^{_{(b)}}\)^{_q}}$
another one was considered, namely
$\sum_{q=1}^{n^{_{(b)}}_j}
\frac{\widetilde A_{q,j}^{_{(b)}}(t,x)\,e^{2i\theta(t,x;k_j^{_{(b)}})}}{\(k-k_j^{_{(b)}}\)^{_q}},$
and hence the corresponding quantities $\widetilde A_{q,j}^{_{(b)}}(t,x)$ depended on $t$ and $x.$
For our RH problem it is more convenient to consider quantities $A_{q,j}^{_{(b)}}$ which are not dependent on $t,x$).
\end{proof}

\subsection*{Riemann-Hilbert problem.}\label{sect_RH}
\noindent Let us define the matrix
\begin{equation}\label{M}
M(t,x;k)\!=\!\begin{cases}
\!\begin{pmatrix}
Z_{[1]}(t,x;k)e^{\ii kt-\ii x\mu(k)} ,& \D\frac{Y_{[2]}(t,x;k)}{a(k)}e^{-\ii kt+\ii x\mu(k)}
\end{pmatrix}\!\!, \quad k \in\! \Omega_+,\!\smallskip\\
\!\begin{pmatrix}	\D\frac{Y_{[1]}(t,x;k)}{a^*(k)}e^{\ii kt-\ii x\mu(k)} ,& Z_{[2]}(t,x;k)e^{-\ii kt +\ii x\mu(k)}
\end{pmatrix}\!\!,\quad k \!\in\! \Omega_-,\!\smallskip\\
\!\begin{pmatrix}
\D\frac{Y_{[2]}(t,x;k)}{b(k)}e^{\ii kt-\ii x\mu(k)} ,& Z_{[2]}(t,x;k)e^{-\ii kt+\ii x\mu(k)}
\end{pmatrix}\!\!,\quad k \!\in\! D_+,\!\smallskip\\
\begin{pmatrix}
Z_{[1]}(t,x;k)e^{\ii kt-\ii x\mu(k)} ,& \D\frac{-Y_{[1]}(t,x;k)}{b^*(k)}e^{-\ii kt+\ii x\mu(k)}
\end{pmatrix}\!\!,\, k \!\in\! D_-.\!
\end{cases}
\end{equation}

\begin{prop}\label{prop_RH}
The function $M$ defined in \eqref{M} solves the following RH problem:
\begin{RHP}\label{mainRHP}
Find a $2\times2$ matrix $M(t,x;k)$ that satisfies the following properties:
\begin{enumerate}
\item\label{RH1} analyticity: $M(t,x;k)$ is analytic in $k\in\mathbb{C}\setminus
\(\Sigma\cup\left\{k_{j}^{_{(a)}}\right\}_{j}
\cup
\{\,\ol{k_{j}^{_{(a)}}}\,\}_{j}
\cup
\left\{k_{j}^{_{(b)}}\right\}_{j}
\cup
\{\,\ol{k_{j}^{_{(b)}}}\,\}_j\)
$
and continuous up to the boundary,
where $k_j^{_{(a)}}$ and $k^{_{(b)}}_j$ are zeros of $a(.)$ in $\Omega_+$ and $b(.)$ in $D_+,$ respectively;

\item \label{RH2} pole conditions: at the zeros $k_j^{_{(a)}}$ of the function $a(.)$ in the domain $\Omega_+$ and their complex conjugates, the following pole conditions are satisfied:
		
$M(t,x;k)\begin{pmatrix}1 & \sum_{q=1}^{n_j^{_{(a)}}}\frac{-A_{q,j}^{_{(a)}}\ e^{-2i\theta(t,x;k)}}{\(k-k_j^{_{(a)}}\)^{_q}}
\\ 0 & 1\end{pmatrix} = \ord(1), \quad k\to k_j^{_{(a)}},$

$	
\hfill M(t,x;k)\begin{pmatrix}1 & 0 \\ \sum_{q=1}^{n_j^{_{(a)}}}\frac{\ol{A_{q,j}^{_{(a)}}}\ e^{2i\theta(t,x;k)}}{\(k-\ol{k_j^{_{(a)}}}\)^{_q}}
& 1\end{pmatrix} = \ord(1), \quad k\to \ol{k_j^{_{(a)}}};$

\item \label{RH3} pole conditions: at the zeros $k_j^{_{(b)}}$ of the $b(.)$ in the domain $D_+$ and their complex conjugates, the following pole conditions are satisfied:

$M(t,x;k)\begin{pmatrix}1 & 0 \\ \sum_{q=1}^{n_j^{_{(b)}}}\frac{-A_{q,j}^{_{(b)}}\ e^{2i\theta(t,x;k)}}{\(k-k_j^{_{(b)}}\)^{_q}} & 1\end{pmatrix} = \ord(1), \quad k\to k_j^{_{(b)}},$

$	
\hfill M(t,x;k)\begin{pmatrix}1 & \sum_{q=1}^{n_j^{_{(b)}}}\frac{\ol{A_{q,j}^{_{(b)}}}\ e^{-2i\theta(t,x;k)}}{\(k-\ol{k_j^{_{(b)}}}\)^{_q}}
\\ 0 & 1\end{pmatrix} = \ord(1), \quad k\to \ol{k_j^{_{(b)}}};$

\item jump conditions: $M_{-}(t,x;k)=M_{+}(t,x;k)J(t,x;k), \qquad
k\in\Sigma\setminus\{-\frac{1}{2}, 0, \frac{1}{2}\}$, 
where
\begin{align}\label{J1}
J(t,x;k)=&
\begin{pmatrix}	1+|r(k)|^2 & -r(k)e^{-2\ii\theta(t,x;k)} \\\\
-{r^*(k)}e^{2\ii\theta(t,x;k)} & 1
\end{pmatrix}, & k \in \mathbb{R},\; |k|>\frac{1}{2}, \nonumber\\
=&\begin{pmatrix}
1& -{(r^*(k))}^{-1}e^{-2\ii\theta(t,x;k)} \\\\
-r(k)^{-1}e^{2\ii\theta(t,x;k)} & 1+|r(k)|^{-2}
\end{pmatrix},& \quad k \in \mathbb{R},\; |k|<\frac{1}{2}, \; k\neq0;\nonumber\\
\end{align}
\begin{align}
=&\begin{pmatrix}
	0 & -r(k)e^{-2\ii\theta(t,x;k)} \\\\
	r(k)^{-1}e^{2\ii\theta(t,x;k)} & 1
\end{pmatrix}, &\quad k \in C_{u}(1/2), \nonumber\\
=&\begin{pmatrix}
0 & (r^*(k))^{-1}e^{-2\ii\theta(t,x;k)} \\\\
-r^*(k)e^{2\ii\theta(t,x;k)} & 1
\end{pmatrix},&\quad k \in C_{d}(1/2), \label{J2}
\end{align}
where  $r(k):=b(k)/a(k)$ is defined on $\mathbb{R}\cup C_{u}(1/2)$, and $\theta(t,x;k)= k t - x \mu(k)$ is defined in \eqref{theta}.
		
\item $M(t,x;k)$ is bounded in the neighbourhoods of the points
$\{-\frac{1}{2}, 0, \frac{1}{2}\}$;

\item normalisation: $ M(t,x;k)=I+O(k^{-1}),\quad |k|\to\infty.$
\end{enumerate}
\end{RHP}
\end{prop}

\begin{proof}
The proof is standard and boils down to direct verification of all the properties.
\end{proof}

\begin{rem}
If the zeros of $a(.), b(.)$ lie on the circle $C_u(\frac12),$ the corresponding limits in the conditions \eqref{RH2}, \eqref{RH3} are understood within the respective domain.

Note that RH problem \ref{mainRHP} is fully determined by the initial and boundary data \eqref{IBC}.

Note also that the RH problem \ref{mainRHP} is formulated on the contour $\Sigma,$ which is the union of the continuous spectra of both Lax operators for the Maxwell-Bloch equations.
\end{rem}

\section{Matrix Riemann-Hilbert problem}\label{sect_matrix_RH}

In this section, we drop the assumption of the previous Section \ref{sect_scattering_data} that there exists a solution to the initial value boundary problem \eqref{MB1a}, \eqref{IBC}.

From now on we restrict ourselves to the case of {\it the input pulse}, i.e. we assume that the initial conditions are trivial, i.e. given by formulae \eqref{trivial}, and a boundary function \eqref{input_pulse} satisfies condition \eqref{first_moment}.
Recall (see Remark \ref{rem_prop_abr_trivial}) that the functions $b(.), r(.)$ do not vanish identically provided that $\mathcal{E}_1(t)$ does not vanish identically.

\begin{prop}\label{prop_ibv}
Let $\mathcal{E}_{1}(t)$ be a locally integrable function satisfying \eqref{first_moment} and not identically equal to zero, let $a(.), b(.)$ be corresponding to $\mathcal{E}_{1}(t)$ spectral functions of the $t$-equation defined by formulae \eqref{Phi0} (i.e. $a(k)=A(k), b(k)=B(k)$) and let $r(k) = \frac{b(k)}{a(k)}$ and let $r(k)=\mathcal{O}(k^{-2})$ as $k\to\infty.$
Then
\begin{enumerate}[I.]
\item RH problem \ref{mainRHP} has a unique solution.
\item Functions $\mathcal{E}, \rho, \mathcal{N}$ defined by the following formulae
\begin{align}\label{E1}
{\cal E}(t,x)=&-\lim\limits_{k\to\infty}4\ii k M_{12}(t,x;k),
\end{align}
\begin{equation}\label{F1}
\begin{pmatrix} {\cal N}(t,x) & \rho(t,x) \\
\ol{\rho(t,x)} & -{\cal N}(t,x) \end{pmatrix}=M(t,x;+\ii 0)\sigma_3 M^{-1}(t,x;+\ii 0),
\end{equation}
satisfy the MB equations \eqref{MB1a} together with the initial and boundary conditions \eqref{input_pulse},\eqref{trivial}.
\item Given a solution $M(t,x;k)$ of the RH problem \ref{mainRHP} (which exists and is unique in view of part I of this Proposition), define functions $Y_{[1]}, Y_{[2]}, Z_{[1]}, Z_{[2]}$ by formula \eqref{M}. Then these functions satisfy the properties of Lemmas \ref{lem_Y_prop}, \ref{lem_Z_prop} 
and, moreover, the function $Z = (Z_{[1]}, Z_{[2]})$ satisfies the property in Remark \ref{remZ}, i.e. $Z$ is analytic in $k\in\mathbb{C}\setminus\{0\}$.

\end{enumerate}

\end{prop}

\begin{proof}
I. Note that the RH problem satisfies the Schwartz reflection symmetry \cite{Zhou89}:
$\ol{J^T(\ol k)} = J(k)$ for $k\in \Sigma\setminus\mathbb{R},$ $J(k) + \ol{J^T(k)}$ is positive definite for $k\in\mathbb{R}$ (here superscript $T$ denotes transposition, and $J(k) = J(t,x;k)$), which implies \cite[Theorem 9.3]{Zhou89} the solvability of the RH problem.

Furthermore, it is easy to see that $\det J(t,x;k) \equiv 1$, $k\in\Sigma\setminus\{-\frac12, 0, \frac12\},$ and hence $\det M(t,x;k)\equiv1$ for $\forall k\in\mathbb{C}\setminus\Sigma.$
The uniqueness then follows by a standard argument: assuming that there exists another solution $\widetilde M,$ it follows that the function $\widetilde M M^{-1}$ 
\begin{enumerate}
\item  has identity jumps across $\Sigma,$ 
\item does not have poles at the points $k_j^{_{(a)}},$ $k_j^{_{(b)}}$ and their complex conjugates, 
\item is bounded everywhere, 
\item tends to the identity matrix at infinity.
\end{enumerate}
It then follows from the Liouville theorem that $\widetilde M M^{-1}$ is the identity matrix.

	\medskip
II. {\it Step 1: differentiability of $M(t,x;k)$ in $t, x.$} In view of Remark \ref{rem_prop_abr_trivial}, the function $r(k)=\frac{b(k)}{a(k)}$
admits an analytic continuation from the real axis.
It then follows that the solution of the RH problem is differentiable in $t, x$ (cf. \cite[Theorem 2.2]{KM19}).
Indeed, first of all we note that  differentiation of the jump matrices $J(t,x,k)$ with respect to $t$ and $x$ multiply their entries on $k$ and $k-\frac{1}{4k}$ correspodently. This forces us to pass to the equivalent Riemann-Hilbert problem. This can be done due to the analyticity of the reflection coefficient by reformulating the RH problem for a new contour, which should bypass only the origin of the complex plane because at infinity 
$r(k)$ vanishes like $k^{-m}$ ($m\ge2$).  
\begin{figure}[ht]
\hskip2cm	\begin{tikzpicture}
		\setlength{\unitlength}{0.60mm}
		\linethickness{1pt}
		\draw[thick,->] (-7,0) to (-3,0);
		\draw[thick,->] (3,0) to (7,0);
		\draw[very thick, dashed] (-3,0) to (3,0);
		\draw [fill=black] (0,0) circle (2pt);	
\draw[thick,postaction=decorate, decoration={markings, mark = at position 0.5 with {\arrow{<}}}](3,0)[out=150, in =30] to (-3,0);
\draw[thick,postaction=decorate, decoration={markings, mark = at position 0.5 with {\arrow{<}}}](3,0)[out=-150, in =-30] to (-3,0);
		\node at (0, -1.3) {$\ol{\hat L}$};
		\node at (2.6, 2.3) {$C_u$};
		\node at (2.6, -2.3) {$C_d$};
		
		\node at (1,0.4 ) {$D$};
		\node at (1,-0.4 ) {$\ol{D}$};
		\node at (3.4,-0.4 ) {$1/2$};
		\node at (-3.5,-0.4 ) {$-1/2$};
		\node at (0,-0.3) {$0$};	\draw [line width=1pt,] (0,0) circle (0.04cm);
		\draw [ line width=0.8pt,] (0,0) circle (3cm);
	\end{tikzpicture}
	\caption{All branch of the jump contour $\Sigma^{(1)}$ are oriented from the left to the right.} \label{sigma1}
\end{figure}
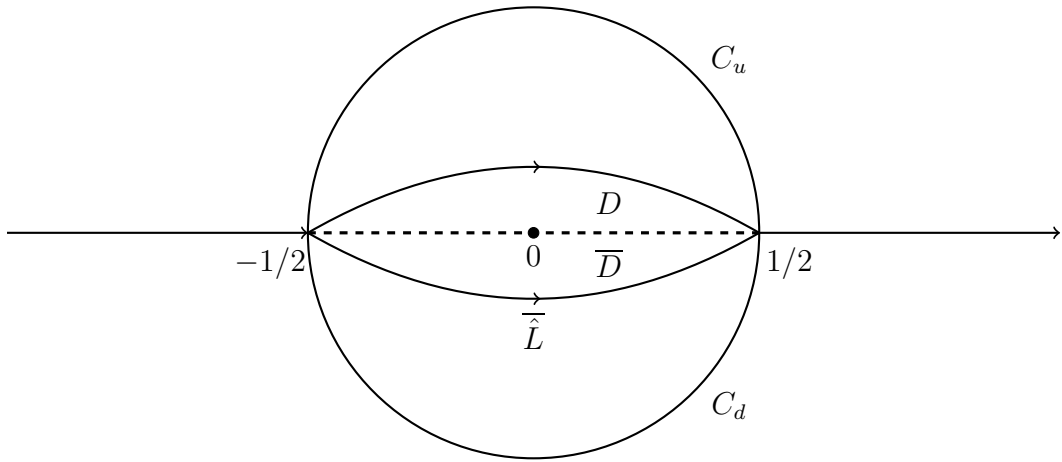
The factorization of the jump matrix given on the interval $(-\frac{1}{2}, \frac{1}{2})$
\begin{align*}
	J(t,x,k)=&\begin{pmatrix}1&0\\-r^{-1}(k)\ee^{2\imi\theta(t,x;k)}&1 \end{pmatrix}
	\begin{pmatrix}1&(-r^*(k))^{-1}\ee^{-2\imi\theta(t,x;k)}\\0&1
	\end{pmatrix}
\end{align*}
suggests a suitable transformation of the RH problem, namely:	
\[ {M}^{(1)}(t,x,k)=M(t,x,k) G^{(1)}(t,x,k),\]
where (see figure \ref{sigma1})
\begin{align*}
	G^{(1)}(t,x,k)=&\begin{pmatrix}1&0\\-r(k)^{-1}\ee^{2\imi\theta(t,x;k)}&1 \end{pmatrix},
	& k\in D,\\
	=&\begin{pmatrix}1&(r^*(k))^{-1}\ee^{-2\imi\theta(t,x;k)}\\0&1
	\end{pmatrix},
	&k\in\overline{D},\\
	=&\begin{pmatrix}1&0\\0&1
	\end{pmatrix},& k\notin D \cup\overline{D},
\end{align*}
This transformation implies the following RH problem:
\[{M}^{(1)}_-(t,x,k)={M}^{(1)}_+(t,x,k){J}^{(1)}(t,x,k), \qquad k\in\hat{\Sigma}^{(1)},\]
\[{M}^{(1)}(t,x,k)\rightarrow I,\ k\rightarrow\infty,\]
where $\Sigma^{(1)}=
(-\infty,-\frac{1}{2})\cup (\frac{1}{2}, \infty)\cup C_{u}(\frac{1}{2})\cup C_d(\frac{1}{2}) \cup\hat L \cup\overline{\hat L}$ and the jump matrix $J^{(1)}(t,x,k)=(G^{(1)}_+(t,x,k))^{-1}J(t,x,k)G^{(1)}_-(t,x,k)=$
\begin{align*}
	=&\begin{pmatrix}1&0\\0&1
	\end{pmatrix},& k\in (-\frac{1}{2}, \frac{1}{2})\\
	=&\begin{pmatrix}1&0\\-r^{-1}(k)\ee^{2\imi\theta(t,x;k)}&1 \end{pmatrix}, & k\in\hat L,\\
	=&\begin{pmatrix}1&(r^*(k))^{-1}\ee^{-2\imi\theta(t,x;k)}\\0&1
	\end{pmatrix}, & k\in\ol{\hat L}\end{align*}
and $J^{(1)}(t,x,k)=J(t,x,k)$ for $k\in(-\infty,-\frac{1}{2})
\cup (\frac{1}{2}, \infty)\cup C_{u}(\frac{1}{2})\cup C_d(\frac{1}{2}). $	
We have to note that matrix $M^{(1)}(t,x,k)$  is analytic everywhere with the exception of the contour $\Sigma^{1}$. Indeed, at the point $k=0$ matrices $M(t,x,k)$ and $G^{(1)}(t,x,k)$  are bounded in upper and lower half vicinities of zero.  Hence the matrix   $M^{(1)}(t,x,k)$, is also bounded there and continuous in  punchered  vicinity and therefore it analytic and at the point zero.

As usual, let us put $M^{(1)}_+(t,x,s)=I+N^{(1)}(t,x,s)$ when $ s\in \Sigma_+^{(1)}$ and pass to the singular integral equation:
$$
N^{(1)}(t,x,s) - \dfrac{1}{2\pi\imi}\int\limits_{\Sigma^{(1)}}\dfrac{N^{(1)}(t,x,z)[I-J^{(1)}(t,x,z)]}{(z-s)_+}d z=
\dfrac{1}{2\pi\imi}\int\limits_{\Sigma^{(1)}}\dfrac{[I-J^{(1)}(t,x,z)]}{(z-s)_+}d z.
$$ 
After X. Zhou \cite{Zhou89} it is well known that such type of singular integral equations are uniquely solvable in Hilbert space $L^2(\Sigma^{(1)})$. In particular, the unique solvability of the mixed problem to Maxwell-Bloch equations can be found in  \cite{K13}. Thus the above singular integral equation  has a unique solution $N^{(1)}(t,x,s)\in L^2(\Sigma^{(1)})$. Now we can differentiate the last equation in $t$ and $x$ as many times as it allows the number $m$. Indeed, to differentiate these equations and matrix $N^{(1)}(t,x,s)$ it is sufficient that its formal derivatives are convergent.
So the derivative of order $p$ leads to the multiplication of the integrand by the factor $k^q(k+\frac{1}{4k})^{p-q}$  where ($q\le p$). 
Thus, the integrands have the order $O(k^{p-m})$ when $k\to\infty$ since the reflection coefficient $r(k)$ decreases as $O(k^{-m})$. The power factors on the finite part of the contour $\Sigma^{(1)}$ do not affect the convergence of the integrals, while at infinity the $L^2$-convergence takes place under the condition
$2(m-p)>1$ that means $1\le p<m-1/2$. Therefore we have to take into account such inequalities: $1\le p\le m-1$ and, hence, $m\ge2$. This provides a unique solvability and existence of the partial derivatives of $N^{(1)}(t,x,k)$ with respect to $t$ and $x$. Hence for $k\neq0$ matrices $M^{(1)}(t,x,k)$ and $M(t,x,k)$ have partial derivatives of order $p\le m-1$ under condition that $m\ge2$. 

{\it Step 2: Lax pair equations from the RH problem solution.}
It is known (see, for example, \cite{K13}, \cite{FKM17}) that the next theorem is valid.
\begin{thm}\label{AKNS}
	Let
	${\Phi}(t,x,k):={M}^{(1)}(t,x,k)\ee^{-\imi k(t-x)\sigma_3 -\imi\frac{x\sigma_3}{4k}} $. Then  $ {\Phi}(x,t,k)$ satisfies\\ the 
	Ablowitz-Kaup-Newell-Segur  system of equations:
	\begin{align*}
		\Phi_t=&-\left(\imi k\sigma_3 +H(t,x)\right)\Phi, \\  \label{xequ1}
		\Phi_x=&\left(\imi k\sigma_3 +H(t,x)+\frac{\imi F(t,x)}{4k}\right)\Phi, 
	\end{align*}
	The matrices $H(t,x)$ and $F(t,x)$ are defined as
	\begin{equation}\label{H1}\nonumber
		H(t,x)=-\imi [\sigma_3, m^{(1)}(t,x)],\qquad m^{(1)}(t,x):=-\dfrac{1}{\pi}\int\limits_{\Sigma^{(1)}} \left[I+N^{(1)}(t,x,k)\right]\left[I-J^{(1)}(t,x,k)\right]  d k,
	\end{equation}
	\begin{equation}\label{F1}\nonumber
		F(t,x):=\begin{pmatrix} {\cal N}(t,x) & \rho(t,x) \\
			\overline {\rho (t,x)} & -{\cal N}(t,x) \end{pmatrix}=-M^{(1)}(t,x,0)\sigma_3\left(M^{(1)}(t,x,0)\right)^{-1}.
	\end{equation}
\end{thm}
\begin{Cor}
	AKNS equations are evidently compatible, i.e. $\Phi_{tx}=\Phi_{xt}$, and hence the Maxwell-Bloch equations in the form \eqref{HF} are satisfied.
\end{Cor}

{\it Step 3: initial and boundary conditions.}
 It remains to check that the initial and boundary conditions are fulfiled.
The verification that the $\mathcal{E}, \rho, \mathcal{N}$ satisfy the trivial initial conditions \eqref{trivial} follows in the same way as the proof of Theorem \ref{thm_caus}.

To check the fulfilment of the boundary conditions,
set $x=0$ in RH problem \ref{mainRHP} and make the following transformation:
\\$M^{(1)}(t;k) = M(t,x=0;k)G^{(1)}(t;k),$
where 
\[
G^{(1)}(t;k) = \begin{pmatrix}1& 0\\
-r(k)^{-1} e^{2\ii t k}& 1
\end{pmatrix},  k\in D_+,
\quad
G^{(1)}(t;k) = \begin{pmatrix}1& (r^*(k))^{-1} e^{-2\ii t k}\\
0& 1 \end{pmatrix}, k\in D_-,
\]
and
$G^{(1)}(t;k) = I\mbox{ elsewhere.}$
We then obtain an equivalent RH problem
$$M^{(1)}_-(t;k)=M^{(1)}_+(t;k)J^{(1)}(t;k), \qquad k\in \mathbb{R},
$$
with the jump matrix
$
J^{(1)}(t;k)=
\begin{pmatrix}	1+|r(k)|^2 & -r(k)e^{-2\ii t k} \\
-{r^*(k)}e^{2\ii t k} & 1
\end{pmatrix}, k \in \mathbb{R},
$
and the pole conditions as in RH problem \ref{mainRHP} at the zeros $k_j^{_{(a)}}$ of the function $a(.)$ in the domain $\Im k\geq0$ and their complex conjugates.

Note that this RH problem has a unique solution that is given as follows:
\[
M^{(1)}(t;k)
=
\begin{cases}
\begin{pmatrix}
\Psi_{1}(t;k)e^{ikt}, & \frac{1}{a(k)}\Phi_2(t;k)e^{-ikt}
\end{pmatrix}, \Im k>0,
\\
\begin{pmatrix}
\frac{1}{a^*(k)}\Phi_1(t;k)e^{ikt}, & \Psi_{2}(t;k)e^{-ikt}
\end{pmatrix}, \Im k<0,
\end{cases}
\]
where $\Phi(t;k), \Psi(t;k)$ are the solutions of the $t$-equation \eqref{teq},
$\Phi_t(t;k) + i k \sigma_3 \Phi(t;k) = \begin{pmatrix}0 & -\frac12\mathcal{E}_1(t) \\ \frac12\ol{\mathcal{E}_1(t)} & 0\end{pmatrix}\! \Phi(t;k),$ which are 
uniquely determined by the boundary conditions
$
\lim\limits_{t\to\infty}\Phi(t; k)e^{i k t \sigma_3} = I,
$
$
\Psi(t=0; k) = I,
$
respectively.
It then readily follows from \eqref{E1} that 
$\mathcal{E}(t,x=0;k) = -4i\lim\limits_{k\to\infty}M_{12}(t,x=0;k)=
-4i\lim\limits_{k\to\infty}M^{(1)}_{12}(t;k)=\mathcal{E}_1(t),$ 
which is what we aimed to prove.

\medskip
III. The properties of the columns follow from the fact that $M(t,x;k)$ satisfies the RH~problem~\ref{mainRHP}.
Indeed,
we can write matrix $M(t,x,k)$ as follows 
\[
M(t,x;k)\!=:\!\begin{cases}
	\!\begin{pmatrix}
		\hat Z_{[1]}(t,x;k)e^{\ii\theta(t,x,k)},&\D\frac{\hat Y_{[2]}(t,x;k)} {a(k)}e^{-\ii \theta(t,x,k)}
	\end{pmatrix}\!\!, \quad k \in\! \Omega_+,\!\smallskip\\
	\!\begin{pmatrix}
		\D\frac{\tilde Y_{[2]}(t,x;k)} {b(k)}e^{\ii \theta(t,x,k)} ,& \tilde Z_{[2]}(t,x;k)e^{-\ii \theta(t,x,k)}
	\end{pmatrix}\!\!,\quad k \!\in\! D_+,\!\smallskip\\
	\!\begin{pmatrix}
		\D\frac{\hat Y_{[1]}(t,x;k)} {a^*(k)}e^{\ii\theta(t,x,k)} ,& \hat Z_{[2]}(t,x;k)e^{-\ii\theta(t,x,k)}
	\end{pmatrix}\!\!,\quad k \!\in\! \Omega_-,\!\smallskip\\
	\begin{pmatrix}
		\tilde Z_{[1]}(t,x;k)e^{\ii\theta(t,x,k)} ,& \D\frac{-\tilde Y_{[1]}(t,x;k)}{b^*(k)}e^{-\ii \theta(t,x,k)}
	\end{pmatrix}\!\!,\,\quad k \!\in\! D_-,\!
\end{cases}
\]
where, in view of property \eqref{RH1} of RH problem \ref{mainRHP}, the functions $\hat Z_{[1]}(t,x;k)$ and 
$\hat Y_{[2]}(t,x;k)$ ($\hat Z_{[2]}(t,x;k)$ and 
$\hat Y_{[1]}(t,x;k)$) are meromorphic in $k\in\Omega_+$ ($k\in\Omega_-$),
 while $\tilde Y_{[2]}(t,x;k)$ and
$\tilde Z_{[2]}(t,x;k)$ ($\tilde Z_{[1]}(t,x;k)$ and 
$\tilde Y_{[1]}(t,x;k)$) are meromorphic in $k\in D_+$ ($z\in D_-$). 
Moreover, the residual conditions \eqref{RH2}, \eqref{RH3} of RH problem \ref{mainRHP} imply that $\hat Z_{[1]}(t,x;k), \tilde Z_{[2]}(t,x;k), \hat Z_{[2]}(t,x;k), \tilde Z_{[1]}(t,x;k)$ do not have poles in their domains.
Furthermore, expansions for $a(k)^{-1}, b(k)^{-1}$ from Lemma \ref{lem_zeros_ab} and the residual conditions \eqref{RH2}, \eqref{RH3} show that the rest of the columns, i.e. $\widehat Y_{[2]}(t,x;k),$ $\tilde Y_{[2]}(t,x;k), \tilde Y_{[1]}(t,x;k), \hat Y_{[1]}(t,x;k)$ also do not have poles in their domains (otherwise the products in the conditions \eqref{RH2}, \eqref{RH3} would not be regular).

From the jump conditions on $C_u(\frac12)$ it follows that $\hat Y_{[2]}(t,x;k) = \tilde Y_{[2]}(t,x;k)$ for $k\in C_u(\frac12),$ and hence in view of analyticity $\hat Y_{[2]}(t,x;k) \equiv \tilde Y_{[2]}(t,x;k) =: Y_{[2]}(t,x;k)$ is analytic in $k\in \mathbb{C}_+.$
Similarly, from the jump condition on $C_d(\frac12)$ it follows that $\hat Y_{[1]}(t,x;k) \equiv \tilde Y_{[1]}(t,x;k) =: Y_{[1]}(t,x;k)$ is analytic in $k\in\mathbb{C}_-.$

Furthermore, the jump condition on $C_u(\frac12)$ also implies
\begin{equation}
\label{jump_Cu}
\hat Y_{[2]}(t,x;k) = b(k)\hat Z_{[1]}(t,x;k) + a(k) \tilde Z_{[2]}(t,x;k),\qquad k\in C_u(1/2),
\end{equation}
and hence $\tilde Z_{[2]}(t,x;k)$ has an analytic continuation into $k\in \Omega_+,$ and thus is analytic in $\mathbb{C}_+.$ Besides, $\hat Z_{[1]}(t,x;k)$ admits an analytic continuation into $D_+,$ and thus is analytic in $\mathbb{C}_+.$
Similarly, from the jump condition on $C_d(\frac12)$ it follows that $\tilde Y_{[1]}(t,x;k) = a^*(k) \tilde Z_{[1]}(t,x;k) - b^*(k) \hat Z_{[2]}(t,x;k)$ and thus $\tilde Z_{[1]}(t,x;k)$ and $\hat Z_{[2]}(t,x;k)$ are analytic in $\mathbb{C}_-.$

The jump condition on $\mathbb{R}\setminus[-\frac12, \frac12]$ implies
\[
\hat Y_{[2]}(t,x;k) = b(k) \hat Z_{[1]}(t,x;k) + a(k) \hat Z_{[2]}(t,x;k), \quad k \in (-\infty, -1/2)\cup(1/2, +\infty).
\]
Comparing that with \eqref{jump_Cu}, we conclude that $\hat Z_{[2]}(t,x;k)\!\! \equiv\!\! \tilde Z_{[2]}(t,x;k)\! =:\!\!Z_{[2]}(t,x;k)$ is analytic in $\mathbb{C}\setminus\{0\}.$

To obtain that $\hat Z_{[1]}(t,x;k) \equiv \tilde Z_{[1]}(t,x;k)$  for all $k\in\mathbb{C}\setminus\{0\},$ the easiest way is to observe that if $M(t,x;k)$ is a solution to the RH problem \ref{mainRHP}, then
$\Lambda \ol{M(t,x;\ol k)}\Lambda^{-1}$ is also a solution (here, $\Lambda$ is defined in Lemma \ref{lem_Y_prop}), and thus from uniqueness we obtain that 
$\Lambda \ol{M(t,x;\ol k)}\Lambda^{-1} = M(t,x;k).$
\end{proof}

\subsection*{Causality principle.}
\begin{thm}{\rm [Causality principle.]}\label{thm_caus}
Let the conditions of Proposition \ref{prop_ibv} be satisfied.
Then for $x\geq t$ the solution of RH problem \ref{mainRHP} is unique and trivial, i.e. $M(t,x;k)\equiv I$ and ${\cal E}(t,x)=\rho(t,x)\equiv 0,$ $\mathcal{N}(t,x)\equiv 1$ for $x\geq t.$
\end{thm}
	
\begin{proof}
Note that the spectral function $r(k)$ is not identically zero, is analytic in $k\in\mC_+$ and $r(k)=O(k^{-1})$ as $k\to\infty.$
Now, we need to treat the cases $x>t$ and $x=t$ differently. For $x>t$
we have that $k_0^2=\dfrac{x}{4\tau}$ and $\tau \equiv t-x$ are negative. The distribution of signs of the phase function 
$\theta=\theta(t,x;k),$ defined in \eqref{theta},
is determined by the equality 
$$
\mathrm{sgn}(\Re(\ii \theta))=-\mathrm{sgn}\Im\theta=
-\tau\,\mathrm{sgn}\left[
\left(1+\dfrac{k_0^2}{|k|^2}\right)\Im k\right]\!,
$$	which means that $e^{\mp2\ii\theta} =O(e^{\pm\tau\Im k})$ as $|k|\to\infty$, $\pm\Im k \geq0.$ On the other hand, for $x=t$ we have $e^{\mp2i\theta} = \ord(e^{\frac{x\Im k}{2|k|^2}}) = \ord(1)$ as $|k|\to\infty,$ $\pm\Im k \geq0,$ for every fixed $x.$ This kind of behaviour of the exponents, a factorization of the jump matrix ($x\geq t$)
\begin{align*}
J(t,x;k)=\begin{pmatrix}1+|r(k)|^2& -r(k)e^{-2\ii\theta}\\
-r^*(k) e^{2\ii\theta}& 1
\end{pmatrix}=
\begin{pmatrix}1 & -r(k) e^{-2\ii\theta}\\
0 & 1 \end{pmatrix}
\begin{pmatrix}1&0 \\ -r^*(k) e^{2\ii\theta} &1 \end{pmatrix}\!,
\end{align*}
and analyticity of the functions $r(k)$ and $r^*(k)$ in $\mC_\pm,$ respectively, allows applying the following transformation to the basic RH problem:
$M^{(1)}(t,x;k)=M(t,x;k) G^{(1)}(t,x;k)$, where
\begin{align*}
G^{(1)}(t,x;k)= &
\begin{pmatrix}
1  & -r(k)\ee^{-2\ii\theta(t,x;k)} \\
0 & 1 \end{pmatrix}, &  k\in \Omega_+,\\
=&\begin{pmatrix} 1 & 0 \\
r^*(k)\ee^{2\ii\theta(t,x;k)} & 1 \end{pmatrix}, &  k\in\Omega_-.
\end{align*}
The function $M^{(1)}$ then satisfies the following jump condition,
$$M^{(1)}_-(t,x;k)=M^{(1)}_+(t,x;k)J^{(1)}(t,x;k), \qquad k\in \mathbb{R}, \quad |k|>\dfrac{1}{2}
		$$
with the jump matrix
		$$
J^{(1)}(t,x;k)= (G^{(1)}_+(t,x;k))^{-1}J(t,x;k)G^{(1)}_-(t,x;k)\equiv I, \qquad k\in\mathbb{R}, \quad |k|>\dfrac{1}{2}.
		$$  
Another factorization of the jump matrix, 
		\begin{align*}
J(t,x;k)=\begin{pmatrix}1& -(r^*(k))^{-1} e^{-2\ii\theta}\\
-r(k)^{-1} e^{2\ii\theta}& 1+|r|^{-2}
\end{pmatrix}=
		\begin{pmatrix}1& 0\\
	-r(k)^{-1} e^{2\ii\theta}& 1
			\end{pmatrix}
		\begin{pmatrix}1&
-(r^*(k))^{-1} e^{-2\ii\theta}\\	0& 1
		\end{pmatrix},
\end{align*}
defines $G^{(1)}(t,x;k)$ in the domains $D_\pm$:
	\begin{align*}
G^{(1)}(t,x;k)=& \begin{pmatrix}1& 0\\
-r(k)^{-1} e^{2\ii\theta}& 1
\end{pmatrix}, &  k\in D_+,\\
=&\begin{pmatrix}1& (r^*(k))^{-1} e^{-2\ii\theta}\\
0& 1 \end{pmatrix}, &  k\in D_-,
		\end{align*} 
and we have
$$
M^{(1)}_-(t,x;k)=M^{(1)}_+(t,x;k)J^{(1)}(t,x;k), \qquad k\in \mathbb{R}, \quad |k|<\dfrac{1}{2},  \quad k\neq 0,
$$
with the jump matrix
		$$
J^{(1)}(t,x;k)= (G^{(1)}_+(t,x;k) )^{-1}(k) J(t,x;k) G^{(1)}_-(t,x;k) \equiv I, \qquad k\in\mathbb{R}, \quad |k|<\dfrac{1}{2},  \quad k\neq 0.
		$$  
It is easy to verify that on the circle $|k|=1/2$
		$$
J^{(1)}(t,x;k)= (G^{(1)}_+(t,x;k))^{-1}J(t,x;k)G^{(1)}_-(t,x;k)\equiv I, \qquad k\in C_{u}(1/2)\cup C_{d}(1/2).
		$$

\noindent
Furthermore, note that in view of formulae \eqref{ZYrelations}, we have 
$$
M^{(1)}(t,x;k) = 
\begin{pmatrix}
Z_{[1]}(t,x;k) e^{i\theta(t,x;k)}, & Z_{[2]}(t,x;k) e^{-i\theta(t,x;k)}
\end{pmatrix},
\quad
k\in\Omega_+\cup\Omega_-.
$$
In view of Remark \ref{remZ}, the function $Z(t,x;k)$ is analytic in $k\in\mathbb{C}\setminus\{0\},$
and hence $M^{(1)}$ does not have poles at the zeros of $a(.), b(.).$

\medskip

Next, let us examine the behaviour of $M^{(1)}$ at the origin. For $\tau<0$,
$e^{\pm2\ii\theta} =O(e^{\pm2\tau\frac{|k_0|}{|k|}\sin\arg k})= o(1)$ when $|k|\to 0$, and $\arg k\in (0,\pi)$ for the sign plus and $\arg k\in (-\pi, 0)$
for the sign minus. Hence $G^{(1)}(t,x;k)= O(1)$ and $M^{(1)}(t,x;k)= O(1)$ for all $|k|\to 0$. At the infinity, since $G^{(1)}(t,x;k)=I+O(k^{-1})$ as $|k|\to \infty$, then also $M^{(1)}(t,x;k)=I+O(k^{-1}).$ Further,  $M^{(1)}(t,x;k)$ is analytic everywhere with exception of the points $-\dfrac{1}{2}$, $0$, $\dfrac{1}{2}$ which are  removable singularities for $M^{(1)}(t,x;k)$. Then $M^{(1)}(t,x;k)\equiv I$ by the Liouville theorem.  Hence, $\lim\limits_{k\to\infty} k(M^{(1)}(t,x;k)-I)\equiv 0$ and therefore 
$${\cal E}(t,x)=-4\ii \lim\limits_{k\to\infty} kM_{12}(t,x;k)=-4\ii \lim\limits_{k\to\infty} \left(kM^{(1)}(t,x;k)(G^{(1)}(t,x;k))^{-1}\right)_{12}\equiv 0.
$$ 
Finally, $M(t,x;+\ii 0)\equiv I$ and therefore $F(t,x)\equiv\sigma_3$, i.e. $\rho(t,x)\equiv 0$, and ${\cal N}(t,x)\equiv 1$.
Theorem \ref{thm_caus} is proved.
\end{proof}

We remind that this trivial result is valid in the causality region $t\leq x$. Thus the Riemann-Hilbert  problem provides the well-known causality principle.

\section{Asymptotic analysis. Near the light cone. Proof of Theorem \ref{thm_really_close_light_cone}}\label{sect_thm_1}

According to the causality principle, Theorem \ref{thm_caus}, the solution of the mixed problem is trivial in the region $x\geq t.$ We thus are only interested in the light cone region, which is defined by the inequality $$\tau\equiv t-x>0.$$
Moreover, in this Section we will only deal with a narrow region near the light cone, which is given by the inequalities \eqref{ineq_as_sol}
\begin{equation}\nonumber
x<t \leq x + \frac{m^2\ln^2x + C\ln x\cdot \ln\ln x}{4x},\ \ C>0,
\end{equation}
where a real number $m\geq1$ is defined either in Assumption \ref{assumption_1+} or \ref{assumption_1}, and $C>0$ is an arbitrary constant.

As usual within the framework of the Deift-Zhou method of steepest descent, a crucial role is played by the signature table of the phase function. 
For convenience, let us rewrite $\ee^{\mp2\rmi \theta(t,x;k)}$ in the form:
$$\ee^{\mp2\rmi \theta(t,x;k)} =\exp\left\{\mp 2\ii(t-x) \left( k- \dfrac{k_0^2}{k}\right)
\right\} =\exp\left\{\mp 2i(t-x)S(k, k_0)\right\},$$
where $k_0=k_0(t/x) = \sqrt{\dfrac{x}{4\tau}}>0$ is defined in \eqref{k_0} and $S(k, k_0) = k - \frac{k_0^2}{k}.$
Note that the phase function $S(k, k_0)$ has two saddle points, for the derivative is equal to zero at the points $\pm i k_0,$ 
$$\frac{d}{dk}S(k, k_0)=1+ \dfrac{k_0^2}{k^2}=0.
$$
It is evident that $\frac{d^2}{dk^2}S(k=\mp\ii k_0, k_0)\neq 0.$  
 Then
\[
\mathrm{sgn}(\Im(\theta(t,x;k)))=\mathrm{sgn}(\Im S(k, k_0))=\mathrm{sgn}\left(1+\dfrac{k_0^2}{|k|^2} \right)\Im k=\mathrm{sgn} \Im k,
\]
i.e. there are two domains $\mathbb{C}_\pm$ where $\Im S(k, k_0)$ is positive and negative, respectively.
Further, the saddle points  $\pm\ii k_0$ are located on the imaginary line.
Then, in the spirit of the steepest descent method, we must move to a new contour where the corresponding requirements
\begin{itemize}
	\item $\Re S(k, k_0)=const$,
	\item $\Im S(-\ii k_0, k_0) < \Im S(k, k_0)< \Im S(\ii k_0, k_0)$
\end{itemize}
are fulfiled. Such a contour is nothing more than the circle $|k|= k_0$,
as follows from the identity
$$\theta(t,x;k)
=(t-x)
\(\Re k\cdot\(1-\frac{k_0^2}{k^2}\)+i\Im k\cdot\(1 + \frac{k_0^2}{|k|^2}\)\).$$ 

\subsection{Transformations of the RH problem}

\subsubsection*{Step 1: moving the circle $|k|=\frac12$ to the circle $|k|=\frac12\sqrt{\frac{x}{t-x}}.$}
Following the logic explained above, our first step is to ``move'' the jump from the circle $|k|=\frac12$ to the circle $|k|=k_0.$
This is done with the help of the following transformation of the RH problem \ref{mainRHP}:
\begin{align*}
&
M^{(1)}(t, x; k) = M(t, x; k)\begin{pmatrix}0 & -r(k)e^{-2i\theta} \\ \frac{1}{r(k)}e^{2i\theta} & 1\end{pmatrix},
\quad \frac12<|k|<\frac12\sqrt{\frac{x}{t-x}},\ \Im k > 0,
\\
&
M^{(1)}(t, x; k) = M(t, x; k)\begin{pmatrix}0 & \frac1{r^*(k)} e^{-2i\theta} \\ -r^*(k) e^{2i\theta} & 1\end{pmatrix}^{-1},
\quad \frac12<|k|<\frac12\sqrt{\frac{x}{t-x}},\ \Im k < 0,
\\
& M^{(1)}(t, x; k) = M(t, x; k)\quad \mbox{ elsewhere,}
\end{align*}
where $r^*(k)=\ol{r(\ol k)}$ and $\theta=\theta(t,x;k).$ This transformation 
exploits the fact that $r(k)$ and $r^*(k)$ are analytic in $k\in\mathbb{C}_\pm$, respectively, and
effectively ``erase'' the jump from the circle $|k|=\frac12$ and redraw it on the circle $|k|=k_0=\frac12\sqrt{\frac{x}{t-x}}.$
This means that RH problem \ref{mainRHP} is reformulated now on the new contour $\Sigma(k_0)=\mathbb{R}\cup C_{u}(k_0)\cup C_{d}(k_0),$
where the semicircles $C_{u}(k_0)$ and $ C_{d}(k_0)$ are the upper and lower parts of the circle $|k|=k_0$, respectively, both of which are oriented from the point $-k_0$ to $k_0,$
\begin{equation}\label{Cud}C_u=C_{u}(k_0)=\{k\in\mathbb{C}\!:\, |k| =k_0,\arg k\! \in\! (\pi, 0)\}, 
\ C_d=C_{d}(k_0)=\{k\!\in\! \mathbb{C}\!:\, |k|=k_0,\, \arg k\!\in\! (0, -\pi)\}.
\end{equation}
Function $M^{(1)}(t,x;k)$ satisfies the following RH problem:
\begin{itemize}
	\item $M^{(1)}(t,x;k)$ is analytic in $k\in \mathbb{C}\setminus \Sigma(k_0)$;
	\item $M^{(1)}(t,x;k)$ has continuous non-tangential boundary values $M^{(1)}_{\pm}(t,x;k)$ ($k\in\Sigma(k_0)$), which satisfy the jump relation:\\
	$M^{(1)}_-(t,x;k)=M^{(1)}_+(t,x;k)J^{(1)}(t,x;k),k\in \Sigma(k_0),$ 
	\item $M^{(1)}(t,x;k)$ is bounded in the neighbourhoods of the points
	$\{-k_0, 0, k_0\}$;
	\item $M^{(1)}(t,x;k)$ satisfies the pole conditions at the zeros of the function $a(.)$ in the domain $|k|\ge k_0,$ which are defined by formulae in property $\eqref{RH2}$ of RH problem \ref{mainRHP}, and satisfies the pole conditions at the zeros of the function $b(.)$ in the domain $|k|\le k_0,$ which are defined by formulae in property $\eqref{RH3}$ of RH problem \ref{mainRHP};
	\item $M^{(1)}(t,x;k)=I+\ord(k^{-1}),\quad k\to \infty.$
\end{itemize}
The jump matrix $J^{(1)}(t,x;k)$ is defined by the same formulae \eqref{J1} and \eqref{J2}, but  on the new contour $\Sigma(k_0)$.

Note that since $a(k)\to 1$ as $k\to\infty,$ for $k_0$ sufficiently large there are no zeros of $a(k)$ in the region $|k|\geq k_0,$ $\Im k\geq0.$ Since in the region \eqref{ineq_as_sol} we have $k_0\geq \frac{4x(1+o(1))}{m^2\ln^2 x}\to\infty$ as $x\to\infty,$ we can assume without loss of generality that $M^{(1)}$ has no poles caused by zeros of the function $a(.).$

\subsubsection*{Step 2: scaling.}
In the regime $\frac{x}{t-x}\to+\infty$ the circle in the jump contour for $M^{(1)}$ is expanding. For convenience, we introduce the scaling 
\[
k = k_0z, \quad \mbox{ where }\quad k_0 = k_0(t/x)=\sqrt{\frac{x}{4(t-x)}},
\]
which transforms that circle to the unit circle in variable $z$.
We have 
\[
\theta(t, x; k) = (t-x)k - \frac{x}{4k}
=
\frac12\sqrt{x(t-x)} \(z-\frac{1}{z}\).
\]
Next, let us parametrize the circle $|z|=1$ in the following way: we introduce one parametrization for the upper part of the circle and another one for the lower part of the circle:
\begin{equation}\label{alpha_ud_z}
\begin{split}
&z = i e^{-i\alpha_u},\quad \alpha_u\in\(\frac{-\pi}{2}, \frac{\pi}{2}\),\quad z\in C_{u}(1),
\\
&z = -i e^{i\alpha_d},\quad \alpha_d\in\(\frac{-\pi}{2}, \frac{\pi}{2}\),\quad z\in C_{d}(1).
\end{split}
\end{equation}
Note that with these parametrizations both halves of the circle are oriented from the point $z=-1$ to the point $z=1.$ Then
\begin{equation}\label{theta_alpha_beta}
\begin{split}
-2i\theta(t, x; k=k_0z) &= 2\sqrt{x(t-x)} - 4\sqrt{x(t-x)} \sin^2\frac{\alpha_u}{2},\quad z\in C_{u}(1),
\\
2i\theta(t, x; k=k_0z) &= 2\sqrt{x(t-x)} - 4\sqrt{x(t-x)} \sin^2\frac{\alpha_d}{2},\quad z\in C_{d}(1).
\end{split}
\end{equation}
It follows in particular that $-2i\theta(t,x;k_0z)$ is real on $C_u(1)$ and smaller there than $2\sqrt{x(t-x)}.$

\subsubsection*{Step 3.}
We would like to remove the jump from the interval $z\in(-1, 1).$
To fulfil this objective, we define

\begin{align*}
&M^{(2)}(z; t, x) = M^{(1)}(t,x;k_0z)\begin{pmatrix}1 & 0 \\ \dfrac{-1}{r(k_0z)}e^{2i\theta(t,x;k_0z)} & 1\end{pmatrix}, \quad |z|<1, \Im z > 0,
\\
&M^{(2)}(z; t, x) = M^{(1)}(t, x;k_0z)\begin{pmatrix}1 & \dfrac{1}{r^*(k_0z)} e^{-2i\theta(t,x;k_0z)} \\ 0 & 1\end{pmatrix},  \quad |z|<1, \Im z < 0,
\\
&M^{(2)}(z; t, x) = M^{(1)}(t, x;k_0z), \quad |z|>1.
\end{align*}

This transformation removes the poles caused by the zeros of the function $b(.)$ in the domain $|z|<1.$ This can be seen either directly by applying the corresponding transformation $M^{(1)} \to M^{(2)}$ to the pole conditions of $M^{(1)}$ at the zeros of $b(.)$, or can also be seen in a more simple fashion, by tracking how $M^{(2)}$ depends on $Y, Z$ (which are defined from $M$ by formula \eqref{M}). We have 
$M^{(2)}(z;t,x)=\begin{pmatrix}
Z_{[1]}(t,x;k_0z)e^{i\theta(t,x;k_0z)}, & Z_{[2]}(t,x;k_0z)e^{-i\theta(t,x;k_0z)}
\end{pmatrix},\ |z|<1, \Im z>0,$ and the statement follows from the properties of $Y,Z,$ pp. \pageref{sect_Jost} - \pageref{sect_ab}.
Thus, the function $M^{(2)}$ has no poles neither at the zeros of the function $a(.)$ nor at the zeros of the function $b(.),$ and satisfies the following RH problem (see also Figure \ref{Figure_Omega_up_low}, left):

\subsubsection*{RH problem for $M^{(2)}.$}
\begin{enumerate}
\item $M^{(2)}(z; t, x)$ is analytic in $z\in\mathbb{C}\setminus\Sigma,$ where $\Sigma^{(2)} = \{z:\ |z| = 1\}\cup(\mathbb{R}\setminus(-1, 1)),$
\item $M^{(2)}(z;t,x) \to I$ as $z\to\infty,$
\item $M^{(2)}_-(z; t, x) = M_+^{(2)}(z; t, x) J^{(2)}(z; t, x),$
where
\begin{equation}\label{J_M_2}
\begin{split}
&J^{(2)}(z; t, x) = \begin{pmatrix}1 & -r(k_0z) e^{-2i\theta} \\ 0 & 1\end{pmatrix},\quad z \in C_{u}(1),
\\
&J^{(2)}(z; t, x) = \begin{pmatrix}1 & 0 \\  -r^*(k_0z) e^{2i\theta} & 1\end{pmatrix},\quad z \in C_{d}(1),
\\
&J^{(2)}(z; t, x) = \begin{pmatrix}1 + |r(k_0z)|^2 & -r(k_0z)e^{-2i\theta} \\ -r^*(k_0z)\,e^{2i\theta} & 1\end{pmatrix},\quad z \in (-\infty, -1)\cup(1, +\infty),
\end{split}
\end{equation}
\end{enumerate}
where $r^*(k_0z) = \ol{r(k_0\ol z)}$ and $\theta=\theta(t,x;k_0z).$

We see that as $k_0\to\infty,$ the jump matrix $J^{(2)}$ is close to the identity matrix over $\mathbb{R}\setminus[-1,1],$ but its behaviour over $C_{u}(1)$ and $C_{d}(1)$ depends significantly on the parameter $\sqrt{x(t-x)}.$
In the case when $x(t-x)$ is bounded or grows moderately, the jumps over $C_{u}(1),$ $C_{d}(1)$ remain small and we can use a small-nom theory.
On the other hand, when $x(t-x)$ is large, the jumps over $C_{u}(1),$ $C_{d}(1)$ are growing, and to tackle this issue, the idea is to construct a parametrix that satisfies exactly the jump condition over $C_{u}(1) \cup C_{d}(1)$; this can be done using Hermite polynomials, just as Laguerre polynomials were used in \cite{BM19}, \cite{KM19} (note however a significant difference of our parametrix from the ones from \cite{BM19}, \cite{KM19}: while there the corresponding parametrices were defined only in relatively small domains around the critical points, in our case the parametrices will be defined in rather big domains $\Omega_u, \Omega_d$, see Figure \ref{Figure_Omega_up_low}, right).

In the subsequent Sections \ref{sect_proof_thm_1_part_I} and \ref{sect_proof_thm_1_part_II} we treat the cases of bounded or moderately growing parameter $x(t-x)$, and the case of large parameter $x(t-x)$ is considered in Section \ref{sect_parametrices}.

\subsection{Proof of Theorem \ref{thm_really_close_light_cone}, part I}\label{sect_proof_thm_1_part_I}

It follows from the form of the jump matrices for $M^{(2)}$ \eqref{J_M_2}, formula \eqref{theta_alpha_beta} and Assumption \ref{assumption_1+} that uniformly for $0<\sqrt{x(t-x)}\leq1$ the jump matrix $J^{(2)}$ admits the estimate
\begin{equation}\label{estimates_J}\|J^{(2)}(z;t,x)-I\|_{L^1(\Sigma^{(2)})\cap L^2(\Sigma^{(2)}) \cap L^{\infty}(\Sigma^{(2)})} = \mathcal{O}(k_0^{-m}),\quad k_0\to\infty.\end{equation}
It then follows from the standard small-norm theory that $M^{(2)}(z;t,x)-I = \ord(k_0^{-m})$ uniformly for $z\in\mathbb{C},$ and 
\begin{equation}\label{M2_integral}
M^{(2)} = I + \mathcal{C}[M^{(2)}_+(.)(I-J^{(2)})],
\mbox{ or }
M^{(2)}(z;t,x) = I + \frac{1}{2\pi i}\int_{\Sigma^{(2)}}\frac{[M^{(2)}_+(s;t,x)(I-J^{(2)}(s;t,x))\,ds}{s-z},
\end{equation}
where $M^{(2)}_+ = M^{(2)}_+(.;t,x)$ is the solution of the singular integral equation
\[
M^{(2)}_+ = I + \mathcal{C}_+[M^{(2)}_+(.)(I-J^{(2)})],
\]
and where we denoted
\[
\mathcal{C}f(k) = \frac{1}{2\pi i}\int_{\Sigma^{(2)}}\frac{f(s)\,ds}{s-k},
\qquad
\mathcal{C}_+f(k) = \frac{1}{2\pi i}\int_{\Sigma^{(2)}}\frac{f(s)\,ds}{(s-k)_+}.
\]
Functions $\mathcal{E}, \rho, \mathcal{N}$ can be obtained from $M_+^{(2)}$ by the following Lemma.
\begin{lemma}\label{lem_reconstruction_E_N_rho}
\begin{multline}\label{E_thm_1_proof}
\mathcal{E}(t,x) = 
\frac{-4i k_0}{2\pi i}\int_{\Sigma^{(2)}}[(J^{(2)}(z;t,x) - I)]_{12}dz
+
\frac{-4i k_0}{2\pi i}\int_{\Sigma^{(2)}}[(M^{(2)}_+(z;t,x) - I)(J^{(2)}(z;t,x) - I)]_{12}dz,
\end{multline}
and 
\begin{equation}\label{rho_N_from_M}
\rho(t,x) = -2M^{(2)}_{11}(0;t,x)M^{(2)}_{12}(0;t,x),
\qquad
\mathcal{N}(t,x) = 1 - 2|M^{(2)}_{12}(0;t,x)|^2,
\end{equation}
where $M_{11}(0;t,x)$, $M_{12}(0;t,x)$ are the corresponding entries of the matrix $M^{(2)}(0;t,x),$
\begin{equation}\label{M_2_0}
M^{(2)}(0;t,x)
=
I - \int_{\Sigma^{(2)}}\frac{(J^{(2)}(z;t,x)-I)\,dz}{2\pi i\,z}
- \int_{\Sigma^{(2)}}\frac{(M_+^{(2)}(z;t,x)-I)(J^{(2)}(z;t,x)-I)\,dz}{2\pi i\,z}.
\end{equation}
\end{lemma}
\begin{proof}
Indeed, using formulae \eqref{E1}, \eqref{F1} and tracking back the chain of transformations that led from $M$ to $M^{(2)}$, we easily obtain that the
functions $\mathcal{E},$ $\rho,$ $\mathcal{N}$ can be expressed using $M^{(2)}$ by the same formulae \eqref{E1}, \eqref{F1}, if we change there $M$ to $M^{(2)}.$ From the representation \eqref{M2_integral} we then obtain
\[
\mathcal{E}(t,x) = \frac{-4i k_0}{2\pi i}\int_{\Sigma^{(2)}}[M^{(2)}_+(z;t,x)\(J^{(2)}(z;t,x) - I\)]_{12}dz,
\]
and substituting here $M^{(2)}_+(z;t,x) = I + (M^{(2)}_+(z;t,x) - I)$ we obtain formula \eqref{E_thm_1_proof}.
Similarly we obtain formula \eqref{M_2_0}. Finally, to get formulae \eqref{rho_N_from_M}, 
note that $M^{(2)}(0;t,x)$ has the structure
$$M^{(2)}(0;t,x) = \begin{pmatrix}A & B \\ -\ol{B} & \ol{A}\end{pmatrix},$$
and $|A|^2 + |B|^2 = 1.$
Multiplying matrices in \eqref{F1} (where we substituted $M$ with $M^{(2)}$), we find that 
$\rho(t,x) = -2AB,$ $\mathcal{N} = |A|^2 - |B|^2 = 1 - 2|B|^2,$ which finishes the proof of the Lemma.
\end{proof}

We now use Lemma \ref{lem_reconstruction_E_N_rho} to obtain asymptotic formulae \eqref{E_thm_1}, \eqref{N_thm_1}, \eqref{rho_thm_1} for the functions $\mathcal{E}, \mathcal{N}, \rho.$
\subsection*{Function $\mathcal{E}(t,x).$}
We now obtain the asymptotics for $\mathcal{E}(t,x).$
The second term in formula \eqref{E_thm_1_proof} is of the order $k_0^{-2m+1},$ since both $M^{(2)} - I$ and $J^{(2)} - I$ are of the order $k_0^{-m}.$
Furthermore, the first term in formula \eqref{E_thm_1_proof} can be written more explicitly as 
\begin{multline}\label{first_term_for_E}
\frac{-4ik_0}{2\pi i}\int_{\Sigma^{(2)}}(J^{(2)}(z;t,x) - I)_{12}dz
=
\frac{4i k_0}{2\pi i}\int_{\Sigma_u}r(k_0z)e^{-i\sqrt{x(t-x)}(z-\frac{1}{z})}dz
\\
=
\frac{4\,i^{m+1}k_0\, r(ik_0)}{2\pi i}
\int_{\Sigma_u}z^{-m}e^{-i\sqrt{x(t-x)}(z-\frac{1}{z})} dz
+
\frac{4i k_0}{2\pi i}
\int_{\Sigma_u}\ord(k_0^{-m-1}z^{-m})e^{-i\sqrt{x(t-x)}(z-\frac{1}{z})} dz,
\end{multline}
where we integrate over the contour $\Sigma_u,$
$$\Sigma_u = (-\infty, -1)\cup C_u \cup(1, +\infty),$$
and where we used that
$
r(k_0z) = i^m\,r(i k_0)z^{-m} + \ord(k_0^{-m-1}z^{-m})\quad \mbox{ uniformly for }z\in \Sigma_u,
$
which is a direct consequence of the representation for $r(.)$ from Assumption \ref{assumption_1+}.
The second term in \eqref{first_term_for_E} admits the estimate $\ord(k_0^{-m}),$ and the first term is equal to 
$4i^{m-1}k_0 r(ik_0)J_{m-1}(-2i\sqrt{x(t-x)}),$ where $J_{m-1}$ is the Bessel function of the first kind of the order $m-1.$
Indeed, recall the integral representation for the Bessel function \cite[formula (7), p.640]{lavrentevshabat} (cf. \cite[formula (10.9.19)]{DLMF}),
$$J_{\nu}(\xi) = \frac{1}{2\pi i}\int_{-\infty}^{(0_+)}z^{-\nu-1}e^{\frac{\xi}{2}(z-z^{-1})}dz,$$
where the integral path goes from $-\infty$ along the lower bank of $\mathbb{R}_-$, then circumvents the origin in the positive (counter-clockwise) direction, and then goes back to $-\infty$ along the upper bank of $\mathbb{R}_+.$
Note that by Jordan's lemma, the integral of $z^{-m}e^{-i\sqrt{x(t-x)}(z-\frac{1}{z})}$ over $\Sigma_u$ is equal to minus integral of it over $(-\infty, (0_+)).$
Combining together the above estimates and expressing the Bessel function in terms of the modified Bessel function (\cite[formulae (9.6.3), (9.1.35)]{abramowitz}),
$$I_{\nu}(\xi) = i^{\nu}J_{\nu}(-i\xi), \quad \xi>0,$$
we obtain formula \eqref{E_thm_1}.

\subsection*{Functions $\mathcal{N}(t,x),$ $\rho(t,x).$}
Similarly as in the previous paragraph, we note that the last term in formula \eqref{M_2_0} is of the order $\ord(k_0^{-2m}),$ since both factors in the integrand are of the order $\ord(k_0^{-m}).$
Then
\begin{equation}\label{M_11_estimate}
M_{11}^{(2)}(0;t,x) = 1 - \frac{1}{2\pi i}\int_{\mathbb{R}\setminus[-1,1]}|r(k_0z)|^2\frac{dz}{z} + \ord(k_0^{-2m}) = 1 + \ord(k_0^{-2m}),
\end{equation}
\begin{multline}\label{M_12_estimate}
M_{12}^{(2)}(0;t,x) = \frac{1}{2\pi i}\int_{\Sigma_u}\(\frac{i^m r(i k_0)}{z^m} + \ord(\frac{1}{k_0^{m+1}z^m})\)e^{-i\sqrt{x(t-x)}(z-\frac{1}{z})}\frac{dz}{z} + \ord(k_0^{-2m})
\\
=i^{m+2}r(ik_0)J_{m}(-2i\sqrt{x(t-x)}) + \ord(k_0^{-m-1}),
\end{multline}
where we manipulate the integrals similarly to the previous section
(devoted to the function $\mathcal{E}(t,x)$). 
Using formulae \eqref{rho_N_from_M} we obtain formulae \eqref{N_thm_1}, \eqref{rho_thm_1}, and thus the statement of Theorem \ref{thm_really_close_light_cone}, part I.
\begin{rem}
Note that the formula $\mathcal{N}(t,x) = |M_{11}^{(2)}(0;t,x)|^2 - |M_{12}^{(2)}(0;t,x)|^2,$
which is equivalent to the second of formulae \eqref{rho_N_from_M},
does not allow to obtain formula \eqref{N_thm_1}, since it gives the error term $\ord(k_0^{-2m})$, which is comparable with the main term in \eqref{N_thm_1}.
We thus need to use an equivalent formula \eqref{rho_N_from_M}, i.e.
$\mathcal{N}(t,x) = 1 - 2|M_{12}^{(2)}(0;t,x)|^2$.
\end{rem}

\subsection{Proof of Theorem \ref{thm_really_close_light_cone}, part II}\label{sect_proof_thm_1_part_II}
The proof goes in the same way, as in part I, Section \ref{sect_proof_thm_1_part_I}, but the corresponding estimates are changed.

In the domain \eqref{domain_part_II} we have
\[
2\leq2\sqrt{x(t-x)} \leq m \ln x - (m + \varepsilon_1)\ln\ln x,
\]
and thus the quantity $p_1=p_1(t,x)$ defined in \eqref{p} satisfies in the domain \eqref{domain_part_II} the inequalities
\[
e^{-p_1}\leq \frac{(x(t-x))^{m/2}}{(\ln x)^{m+\varepsilon_1}} = \ord\(\frac{(\ln x)^m}{(\ln x)^{m+\varepsilon_1}}\)
=\ord((\ln x)^{-\varepsilon_1})
\]
as $x\to\infty.$ In other words, $p_1(t,x)$ grows at least as $\varepsilon_1\ln\ln x,$ i.e.
$p_1(t,x) \geq \varepsilon_1 \ln \ln x + \ord(1).$
Note also that 
$$e^{-p_1(t,x)} = (2k_0)^{-m}e^{2\sqrt{x(t-x)}}.$$

The estimate of the error matrix \eqref{estimates_J} is changed to $\ord(e^{-p_1(t,x)}) = \ord((\ln x)^{-\varepsilon_1}),$ since $J_{err}-I$ is still of the order $\ord(k_0^{-m})$ on $\mathbb{R}\setminus[-1, 1],$ but is of the order $\ord(k_0^{-m}e^{2\sqrt{x(t-x)}})=\ord(e^{-p_1(t,x)})$ on $C_u\cup C_d.$
Hence, $M^{(2)}(z;t,x)-I = \ord(e^{-p_1(t,x)})$ uniformly for $z\in\mathbb{C}.$

\subsection*{Function $\mathcal{E}(t,x).$}
Formula \eqref{E_thm_1_proof} implies, similarly as in formula \eqref{first_term_for_E},
\[
\mathcal{E}(t,x)
=
4k_0\,r(ik_0)\,I_{m-1}(2\sqrt{x(t-x)})
+
\frac{4i k_0}{2\pi i}
\int_{\Sigma_u}\ord(k_0^{-m-1}z^{-m})e^{-i\sqrt{x(t-x)}(z-\frac{1}{z})} dz
+
\ord(k_0e^{-2p_1(t,x)}),
\]
and since the middle integral is of the order $\ord(k_0^{-m}e^{2\sqrt{x(t-x)}}) = \ord(e^{-p_1(t,x)}),$
we get
$$\mathcal{E}(t,x)=
4k_0\,r(ik_0)\,I_{m-1}(2\sqrt{x(t-x)})
+
\ord(e^{-p_1(t,x)} + k_0e^{-2p_1(t,x)}).
$$

\subsection*{Functions $\mathcal{N}(t,x), \rho(t,x).$}
Formulae \eqref{M_11_estimate}, \eqref{M_12_estimate} are now transformed to 
\[
M_{11}^{(2)}(0;t,x) = 1 + \ord(e^{-2p_1(t,x)}),
\ 
M_{12}^{(2)}(0;t,x)
=
-r(ik_0)I_m(2\sqrt{x(t-x)}) + \ord(k_0^{-1}e^{-p_1(t,x)} + e^{-2p_1(t,x)}).
\]
Formulae \eqref{rho_N_from_M} then finish the proof of Theorem \ref{thm_really_close_light_cone}, part II.

\subsection{Parametrices construction}\label{sect_parametrices}
Here we deal with hte parts III, IV of Theorem \ref{thm_really_close_light_cone}.
We start with constructing parametrices around the points $z = \pm i.$

\subsubsection{Hermite polynomials.}

Denote by $\pi_n(.)$ the monic Hermite polynomial of degree $n,$ $\pi_n(\zeta) = \zeta^n + \mathcal{O}(\zeta^{n-1})$ as $\zeta\to\infty,$ orthogonal on $\mathbb{R}$ with the weight $e^{-\zeta^2}:$
\[
\int_{-\infty}^{+\infty}\pi_n(\zeta)\pi_l(\zeta)e^{-\zeta^2}d\zeta = \sqrt{\pi}\, n!\, 2^{-n}\, \delta_{nl},\qquad n,l\geq0.
\]
For $n\geq 0,$ introduce the matrix-valued function
$$
L_n(\zeta) = 
\begin{pmatrix}
\pi_n(\zeta) & \dfrac{1}{2\pi i}\displaystyle\int_{-\infty}^{+\infty}\dfrac{\pi_n(s)\,e^{-s^2}\, ds}{s-\zeta}
\\
\gamma_n\,\pi_{n-1}(\zeta) & \dfrac{\gamma_n}{2\pi i}\displaystyle\int_{-\infty}^{+\infty}\dfrac{\pi_{n-1}(s)\,e^{-s^2}\,ds}{s-\zeta}
\end{pmatrix},
n\geq 1,
\quad
L_0(\zeta)=
\begin{pmatrix}
1 & \dfrac{1}{2\pi i}\displaystyle\int_{-\infty}^{+\infty}\dfrac{e^{-s^2}\, ds}{s-\zeta}
\\
0 & 1
\end{pmatrix},
$$
where $\gamma_n = \dfrac{-i\,\sqrt{\pi}\,2^n}{(n-1)!}, n\geq 1.$
The function $L_n(\zeta)$ satisfies the following jump condition:
\[
L_{n, -}(\zeta) = L_{n, +}(\zeta)
\begin{pmatrix}
1 & -e^{-\zeta^2} \\ 0 & 1
\end{pmatrix},\quad \zeta\in\mathbb{R},
\]
and has the following large $\zeta$ asymptotics:
\[
L_n(\zeta) = (I+\mathcal{O}(\zeta^{-1}))\zeta^{n\sigma_3},\quad \zeta\to\infty.
\]
Note that the rate of vanishing of the off-diagonal terms in the $\mathcal{O}$ term can be improved in the following way:
\begin{align}\label{improved_asymp_2}
\begin{pmatrix}
1 & \dfrac{-i\,n!}{\sqrt{\pi}\,2^{n+1}\,\zeta} \\ 0 & 1
\end{pmatrix}L_n(\zeta)
=
\begin{pmatrix}
1 + \mathcal{O}(\zeta^{-1}) & \mathcal{O}(\zeta^{-2}) 
\\
\mathcal{O}(\zeta^{-1}) & 1 + \mathcal{O}(\zeta^{-1})
\end{pmatrix}
\zeta^{n\sigma_3},\quad \zeta\to\infty,
\\\label{improved_asymp_2}
\begin{pmatrix}
1 & 0 \\ \dfrac{i\,\sqrt{\pi}\,n\,2^n}{n!\,\zeta} & 1
\end{pmatrix}L_n(\zeta)
=
\begin{pmatrix}
1 + \mathcal{O}(\zeta^{-1}) & \mathcal{O}(\zeta^{-1}) 
\\
\mathcal{O}(\zeta^{-2}) & 1 + \mathcal{O}(\zeta^{-1})
\end{pmatrix}
\zeta^{n\sigma_3},\quad \zeta\to\infty.
\end{align}
Note also that $\frac{n}{n!}$ equals 0 for $n=0$ and hence the above formulae make sense for all $n\geq 0.$

\subsubsection{Approximation: the first attempt.}\label{sect_first_attempt}
Introduce the following conformal changes of variables, valid in some neighbourhoods $\Omega_{u}, \Omega_{d}$ of the half-circles $C_{u}, C_{d},$ respectively (see Figure \ref{Figure_Omega_up_low}, right):
\[
\zeta_u = \zeta_u(z;t,x) = 2\sqrt[4]{x(t-x)}\,\sin\frac{\alpha_u}{2}, 
\qquad
\zeta_d =\zeta_d(z,t;x) = 2\sqrt[4]{x(t-x)}\,\sin\frac{\alpha_d}{2},
\]
where $\alpha_u=\alpha_u(z),$ $\alpha_d=\alpha_d(z)$ are defined in \eqref{alpha_ud_z}.
In view of \eqref{theta_alpha_beta} we then have that
\[
-2i\theta(t, x; k_0z) = 2\sqrt{x(t-x)} - \zeta_u^2,\ z\in \Omega_{u},
\]\[
2i\theta(t, x; k_0z) = 2\sqrt{x(t-x)} - \zeta_d^2,\ z\in \Omega_{d}.
\]
Introduce now the following function $M_{appr}$, which satisfies approximately the jump conditions of $M^{(2)}:$
\begin{align*}
& M_{appr}(z; t, x) = \(\frac{z-i}{z+i}\)^{n\sigma_3}, \quad z\in\mathbb{C}\setminus(\Omega_{u}\cup\Omega_{d}),
\\
& M_{appr}(z; t, x) = B_u(z;t,x) L_n(\zeta_u)r(k_0z)^{-\sigma_3/2}\cdot e^{-\sqrt{x(t-x)}\sigma_3}, \quad z\in\Omega_{u},
\\
& M_{appr}(z; t, x) = B_d(z;t,x) \sigma L_n(\zeta_d)\sigma \cdot \(r^*(k_0z)\)^{\sigma_3/2}\cdot e^{\sqrt{x(t-x)}\sigma_3}, \quad z\in\Omega_{d},
\end{align*}
where $\sigma=\begin{pmatrix}0 & 1 \\ 1 & 0\end{pmatrix}.$
Note that since the jump matrix for $M^{(2)}$ admits the following factorization on $C_{u}(1), C_{d}(1),$

\begin{align*}
&
\begin{pmatrix}1 & -r(k_0z) e^{-2i\theta} \\ 0 & 1\end{pmatrix}
=
\(r(k_0z)\)^{\sigma_3/2}e^{\sqrt{x(t-x)}\,\sigma_3}
\begin{pmatrix}
1 & -e^{-\zeta_u^2} \\ 0 & 1
\end{pmatrix}
e^{-\sqrt{x(t-x)}\,\sigma_3}
\(r(k_0z)\)^{-\sigma_3/2},\quad z\in C_{u}(1),
\\
&\begin{pmatrix}1 & 0 \\ -r^*(k_0z)\, e^{2i\theta} & 1\end{pmatrix}
=
\(r^*(k_0z)\)^{-\sigma_3/2} e^{-\sqrt{x(t-x)}\sigma_3}
\begin{pmatrix}
1 & 0 \\ -e^{-\zeta_d^2} & 1
\end{pmatrix}
e^{\sqrt{x(t-x)}\sigma_3}
\cdot\(r^*(k_0z)\)^{\sigma_3/2}, \ z\in C_{d}(1),
\end{align*}
thus the function $M_{appr}$ satisfies exactly the same jump conditions on $C_{u}(1), C_{d}(1)$ as $M^{(2)}$ does.

Next, the $B_u, B_d$ in the definition of $M_{appr}$ are some yet unknown functions, analytic in $z\in\Omega_{u}, \Omega_{d}$, respectively, which are introduced in order to minimize the incoherence of $M_{appr}$ on the borders of $\Omega_{u}, \Omega_{d}.$
We thus define
\begin{equation}\label{B_ud}
\begin{split}
&B_u(z;t,x)
=
\(\frac{z-i}{(z+i)\cdot\zeta_u}\)^{n\sigma_3}\cdot r(k_0z)^{\frac{\sigma_3}{2}}\cdot e^{\sqrt{x(t-x)}\,\sigma_3},\quad z \in \Omega_{u},
\\
&B_d(z;t,x)
=
\(\frac{(z-i)\cdot\zeta_d}{z+i}\)^{n\sigma_3}\cdot r^*(k_0z)^{\,\frac{-\sigma_3}{2}}\cdot e^{-\sqrt{x(t-x)}\,\sigma_3},\quad z \in \Omega_{d}.
\end{split}
\end{equation}
The matching of the $M_{appr}(z)=M_{appr}(z;t,x)$ on $\partial\Omega_{u}$ thus becomes
\begin{multline*}
M_{appr, +}(z)M_{appr, -}(z)^{-1}
=
\(\frac{z-i}{(z+i)\zeta_u}\)^{n\sigma_3}\cdot r(k_0z)^{\frac{\sigma_3}{2}}\cdot e^{\sqrt{x(t-x)}\,\sigma_3}
\(
I+\mathcal{O}(\zeta_u^{-1})\)\cdot
\\
\cdot
e^{-\sqrt{x(t-x)}\,\sigma_3}
\cdot
r(k_0z)^{-\frac{\sigma_3}{2}}
\cdot
\(\frac{z-i}{(z+i)\zeta_u}\)^{-n\sigma_3}=
\\
=
\begin{pmatrix}
1 + \mathcal{O}(\frac{1}{\sqrt[4]{x(t-x)}})
&
\mathcal{O}\(\frac{1}{\sqrt[4]{x(t-x)}}\) \frac{r(k_0 z)\,e^{2\sqrt{x(t-x)}}}{\(x(t-x)\)^{\frac{n}{2}}}
\\
\mathcal{O}\(\frac{1}{\sqrt[4]{x(t-x)}}\) \cdot\frac{\(x(t-x)\)^{\frac{n}{2}}\,e^{-2\sqrt{x(t-x)}}}{r(k_0 z)}
&
1 + \mathcal{O}(\frac{1}{\sqrt[4]{x(t-x)}})
\end{pmatrix} = 
\end{multline*}
\begin{align*}
=&
\begin{pmatrix}
1 + \mathcal{O}\(\frac1{\(x(t-x)\)^{\frac14}}\)
&
\mathcal{O}( e^{2\sqrt{x(t-x)}+\(\frac{m}{2}-\frac{n}{2}-\frac14\)\ln(x(t-x)) - m\ln x} )
\\
\mathcal{O}(e^{-2\sqrt{x(t-x)} + \(\frac{n}{2}-\frac{m}{2}-\frac14\)\ln(x(t-x)) + m\ln x})
&
1 + \mathcal{O}\(\frac{1}{(x(t-x))^{\frac14}}\)
\end{pmatrix}.
\end{align*}
Here we used our Assumption \ref{assumption_1} that $r(k)\asymp k^{-m}$ ($m\geq1$) as $k\to\infty$, and hence 
$$r(k_0z)\asymp \(\frac{x}{t-x}\)^{-m/2} = x^{-m}\cdot\(x(t-x)\)^{m/2} \mbox{ for }|z|=1 \mbox{ as } k_0\to\infty.$$
The arguments of the exponents in the 12 and 21 elements of the above estimates suggest considering the following curves in the $x,t$ plane:
\begin{equation}\label{parametrization}
m\ln x = 2\sqrt{x(t-x)} - \beta\ln\sqrt{x(t-x)},
\end{equation}
thus the matching becomes
\begin{equation}\label{J_err_not_refined}
M_{appr, +}(z)M_{appr, -}(z)^{-1}
=
\begin{pmatrix}
1 + \mathcal{O}\(\frac{1}{\sqrt[4]{x(t-x)}}\)
&
\mathcal{O}\((x(t-x))^{-\frac14 + \frac{\beta}{2} - \frac{n}{2} + \frac{m}{2}}\)
\\
\mathcal{O}\((x(t-x))^{-\frac14 - \frac{\beta}{2} + \frac{n}{2} - \frac{m}{2}}\)
&
1 + \mathcal{O}\(\frac{1}{\sqrt[4]{x(t-x)}}\)
\end{pmatrix}.
\end{equation}
The goal now is to make both the $12$ and $21$ entries small, i.e. 
to make both quantities 
$n-m-\beta-\frac12$ and $-n+m+\beta-\frac12$ well below zero.
This is not possible when $m+\beta +\frac12$ is close to an integer. Hence, instead of $L_n(\zeta)$ we need to use another parametrix with an ``improved'' asymptotic behaviour for large $\zeta.$ This however introduces poles at the points $z=\pm i$, and thus we need to multiply the whole $M_{appr}$ by a matrix that accounts for these poles. We do all this in the
forthcoming Sections \ref{sect_refined_analysis_1}, \ref{sect_refined_analysis_2}.

Note that in \eqref{parametrization}, the parameter $\beta$ can take both positive and negative values. 

\begin{rem}
Note also that inverting \eqref{parametrization} one obtains
\begin{align*}
2\sqrt{x(t-x)} &= m\ln x +\beta\ln\ln x + \beta\ln\frac{m}{2} + \mathcal{O}\(\frac{\ln\ln x}{\ln x}\),\quad x\to+\infty,
\\
& = m\ln t +\beta\ln\ln t + \beta\ln\frac{m}{2} + \mathcal{O}\(\frac{\ln\ln t}{\ln t}\),\quad t\to+\infty.
\end{align*}
This is a consequence of the following Lemma \ref{lem_inversion} (cf.\cite{Knuth}), which we use with $z = \frac{m}{2}\ln x,$ $y = \sqrt{x(t-x)},$ $\gamma = \dfrac{\beta}{2}$ (and which we prove in \ref{sect_appendix}):
\end{rem}
\begin{lemma}\label{lem_inversion}[cf. \cite{Knuth}]
Let $\gamma\in\mathbb{R}$ and 
\begin{equation}\label{ziny}
y -\gamma \ln y = z
\end{equation}
for large positive $y.$ Then as $z\to+\infty,$ y can be expressed in terms of $z$ as follows:
\[
y = z + \gamma \ln z + \frac{\gamma^2 \ln z}{z} + \frac{\gamma^3 \(-\ln^2 z + 2\ln z\)}{2z^2} + \mathcal{O}\(\frac{\ln^3z}{z^3}\),\qquad z\to+\infty.
\]
\end{lemma}


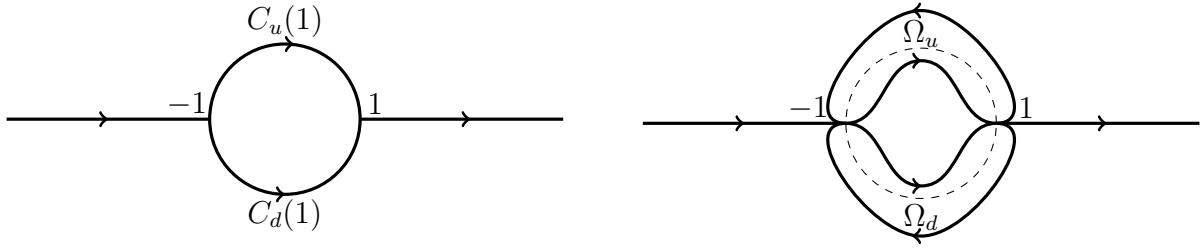
\begin{figure}\begin{center}
\begin{tikzpicture}
\def\r{1}
\def\L{3.7}
\def\h{1.25}
\draw [very thick, decoration={markings, mark=at position 0.25 with {\arrow{<}}}, decoration={markings, mark=at position 0.75 with {\arrow{>}}}, postaction=decorate] (0,0) circle (\r);

\draw[very thick, decoration={markings, mark=at position 0.5 with {\arrow{>}}}, postaction=decorate](-\L, 0) to [out=0, in = 180](-\r, 0);
\draw[very thick, decoration={markings, mark=at position 0.5 with {\arrow{<}}}, postaction=decorate](\L, 0) to [out=180, in = 0](\r, 0);

\node at (0, 1.3 * \r){$C_{u}(1)$};
\node at (0, -1.17 * \r - 0.1){$C_{d}(1)$};
\node at (\r+0.2, 0.2){$1$};
\node at (-\r-0.3, 0.2){$-1$};

\end{tikzpicture}
\qquad
\begin{tikzpicture}
\def\r{1}
\def\L{3.7}
\def\h{1.25}
\draw [dashed] (0,0) circle (\r);

\draw[very thick, decoration={markings, mark=at position 0.5 with {\arrow{>}}}, postaction=decorate](-\L, 0) to [out=0, in = 180](-\r, 0);
\draw[very thick, decoration={markings, mark=at position 0.5 with {\arrow{<}}}, postaction=decorate](\L, 0) to [out=180, in = 0](\r, 0);
\node at (0, 1.19 * \r){$\Omega_{u}$};
\node at (0, -1.15 * \r - 0.1){$\Omega_{d}$};
\node at (\r+0.4, 0.2){$1$};
\node at (-\r-0.5, 0.2){$-1$};

\draw[very thick, decoration={markings, mark=at position 0.5 with {\arrow{<}}}, postaction=decorate] (-\r, 0) [out=180, in=180] to (0, 1.5*\r) [out=0, in=0] to (\r, 0);
\draw[very thick, decoration={markings, mark=at position 0.5 with {\arrow{>}}}, postaction=decorate] (-\r, 0) [out=0, in=180] to (0, \r / 1.2) [out=0, in=180] to (\r, 0);

\draw[very thick, decoration={markings, mark=at position 0.5 with {\arrow{<}}}, postaction=decorate] (-\r, 0) [out=180, in=180] to (0, -1.5*\r) [out=0, in=0] to (\r, 0);
\draw[very thick, decoration={markings, mark=at position 0.5 with {\arrow{>}}}, postaction=decorate] (-\r, 0) [out=0, in=180] to (0, -\r / 1.2) [out=0, in=180] to (\r, 0);

\end{tikzpicture}
\end{center}
\caption{On the left: jump contour for $M^{(2)}$. On the right: jump contour for $M_{err}.$}
\label{Figure_Omega_up_low}
\end{figure}

\subsubsection{The case $n=0$ and proof of Theorem \ref{thm_really_close_light_cone}, part III.}
There is one particular case when the matching condition \eqref{J_err_not_refined} from previous Section \ref{sect_first_attempt} is sufficiently close to the identity matrix. This is the case when $n=0,$ and thus the matrix in the right-hand side of formula \ref{J_err_not_refined} on the half-circles $C_u$, $C_d$ are upper or lower triangular, and hence there is no competition between $12$ and $21$ terms.

Indeed, define
the error matrix $M_{err}(z;t,x) = M^{(2)}(z;t,x)M_{appr}(z;t,x)^{-1},$ then it satisfies the jump condition $M_{err,-}(z;t,x) = M_{err, +}(z;t,x)J_{err}(z;t,x),$ $z\in\Sigma_{err} = (-\infty, -1)\cup(1, +\infty)\cup\partial\Omega_u\cup\partial\Omega_d,$
where
\[
J_{err}(z;t,x)
=
\begin{cases}
\begin{pmatrix}
1 & r(k_0z) \cdot e^{2\sqrt{x(t-x)}} \cdot \frac{1}{2\pi i}\int_{-\infty}^{\infty}\frac{e^{-s^2}ds}{s-\zeta_u}
\\ 0 & 1
\end{pmatrix}, z\in\partial\Omega_u,
\\
\begin{pmatrix}
1 & 0 \\
\ol{r(k_0\ol z)} \cdot e^{2\sqrt{x(t-x)}} \cdot \frac{1}{2\pi i}\int_{-\infty}^{\infty}\frac{e^{-s^2}ds}{s-\zeta_d}
 & 1
\end{pmatrix}, z\in\partial\Omega_d,
\\
\begin{pmatrix}
1 + |r(k_0z)|^2 & -r(k_0z)e^{-2i\theta(t,x;k_0z)}
\\
-r^*(k_0z)e^{2i\theta(t,x;k_0z)} & 1
\end{pmatrix}, z\in(-\infty, -1)\cup(1, +\infty),
\end{cases}
\]
and functions $\mathcal{E}(t,x),$ $\mathcal{N}(t,x),$ $\rho(t,x)$ are reconstructed from the matrix $M_{err}(z;t,x)$ by formulae \eqref{E1}, \eqref{F1}, where we change $M$ with $M_{err}.$

\subsubsection*{Quantity $p_2.$}
Note that in the domain \eqref{domain_part_III}
we have
\[
m\ln x - K \ln\ln x
\leq
2\sqrt{x(t-x)} \leq m\ln x -(m+\varepsilon_2-\frac12)\ln\ln x
\]
and thus the quantity $p_2=p_2(t,x)$ defined in \eqref{p12}
admits the estimate
\[
e^{-p_2} = \frac{e^{2\sqrt{x(t-x)}}\,\(x(t-x)\)^{\frac{m}{2}-\frac14}}{x^m}
\leq
\frac{x^m\,\(x(t-x)\)^{\frac{m}{2}-\frac14}}{x^m\,(\ln x)^{m+\varepsilon_2-\frac12}}
\asymp
\frac{(\ln x)^{m-\frac12}}{(\ln x)^{m+\varepsilon_2-\frac12}}
=
(\ln x)^{-\varepsilon_2},
\]
and hence $p_2$ is a large quantity (but not too large: similarly we obtain $e^{p_2}=\ord((\ln x)^{K-m+\frac12})$).
Note also that $\sqrt{x(t-x)}\asymp \ln x$ as $x\to\infty$ and that 
\begin{equation}\label{p1_elaborate}
e^{-p_2} = \frac{e^{2\sqrt{x(t-x)}}}{(2k_0)^m\ \sqrt[4]{x(t-x)}}.
\end{equation}

\subsubsection*{Estimates of the error matrix $J_{err}.$}
Note that 
$\frac{1}{2\pi i}\int_{\infty}^{+\infty}\frac{e^{-s^2}ds}{s-\zeta_u} = \ord(\zeta_u^{-1})$ for $z\in\partial\Omega_u,$
and hence the jump error $J_{err}$ is close to the identity matrix on $\partial\Omega_u,$
$$
\|J_{err}-I\|_{L^2(C_u)\cap L^1(C_u)\cap L^{\infty}(C_u)} = \ord\(\frac{1}{k_0^m}e^{2\sqrt{x(t-x)}}\frac{1}{\sqrt[4]{x(t-x)}}\)
=\ord(e^{-p_2}).
$$
A similar estimate holds for $C_d,$ and a better estimate holds for the part of the contour on the real axis,
$$
\|J_{err}-I\|_{L^2(\mathbb{R}\setminus[-1,1])\cap L^1(\mathbb{R}\setminus[-1,1])\cap L^{\infty} (\mathbb{R}\setminus[-1,1])} = \ord\(k_0^{-m}\) = \ord(e^{-p_2}).
$$
Hence, considerations of Sections \ref{sect_proof_thm_1_part_I}, \ref{sect_proof_thm_1_part_II} can be repeated; in particular, formulae \eqref{E_thm_1_proof}, \eqref{rho_N_from_M} are valid, if we change there $M^{(2)}$ with $M_{err}$, $J^{(2)}$ with $J_{err}$ and $\Sigma^{(2)}$ with $\Sigma_{err},$ respectively.

\subsection*{Function $\mathcal{E}(t,x).$}
An analogue of formula \eqref{E_thm_1_proof} can be written as 
\begin{multline*}
\mathcal{E}(t,x)
=
\frac{-4ik_0}{2\pi i}\int_{\partial\Omega_u}r(k_0z)\,e^{2\sqrt{x(t-x)}}\,\frac{1}{2\pi i}
\int_{-\infty}^{+\infty}\frac{e^{-s^2}\,ds}{s-\zeta_u}dz
+
\frac{4i k_0}{2\pi i}\int_{\mathbb{R}\setminus[-1, 1]}r(k_0z)e^{-2i\theta(t,x;k_0z)}dz
\\
-
\frac{4ik_0}{2\pi i}\int_{\Sigma_{err}}[(M_{err, +}(z;t,x)-I)(J_{err}(z;t,x)-I)]_{12}dz
=: A_1 + A_2 + A_3.
\end{multline*}
Integrals $A_2, A_3$ admit the estimates
\begin{equation}\label{II_III_estimates}
A_3 = \ord\(\frac{e^{4\sqrt{x(t-x)}}}{k_0^{2m-1}\sqrt{x(t-x)}}\)
=
\ord(k_0e^{-2p_2}),
\qquad
A_2 = \ord\(k_0^{1-m}\) = \ord(k_0e^{-2p_2}),
\end{equation}
and in the term $A_1$ we substitute
\begin{equation}\label{integral_estimate}\frac{1}{2\pi i}\int_{\infty}^{+\infty}\frac{e^{-s^2}ds}{s-\zeta_u} = \frac{i}{2\sqrt{\pi}}\zeta_u^{-1} + \ord(\zeta_u^{-3}),
\end{equation}
then $A_1$ splits accordingly in two terms: the first one can be computed explicitly by computing the first order residue at the point $z=i$ (we use the relation $\zeta_u = \sqrt[4]{x(t-x)}e^{i\alpha_u/2}(z-i)$), and the second term admits an appropriate estimate:
\begin{multline}\label{I_estimate}
A_1 = 
\frac{-4ik_0}{2\pi i}\int_{\partial\Omega_u}r(k_0z)\,e^{2\sqrt{x(t-x)}}\,
\(
\frac{i}{2\sqrt{\pi}}\frac{e^{-i\alpha_u(z;t,x)/2}}{\sqrt[4]{x(t-x)}(z-i)}
+
\ord\(\frac{1}{(x(t-x))^{3/4}}\)
\)dz
\\
=
\frac{2k_0\,r(ik_0)\,e^{2\sqrt{x(t-x)}}}{\sqrt{\pi}\sqrt[4]{x(t-x)}} + \ord\(\frac{e^{2\sqrt{x(t-x)}}}{k_0^{m-1}\,\(x(t-x)\)^{3/4}}\)
=
\frac{2k_0\,r(ik_0)\,e^{2\sqrt{x(t-x)}}}{\sqrt{\pi}\sqrt[4]{x(t-x)}} + \ord\(\frac{k_0\,e^{-p_2}}{\ln x}\),
\end{multline}
where we used \eqref{p1_elaborate}.
Combining \eqref{II_III_estimates} and \eqref{I_estimate}, we obtain
\begin{equation}\label{E_part_III_interm}
\mathcal{E}(t,x)
=
\frac{2k_0r(ik_0)e^{2\sqrt{x(t-x)}}}{\sqrt{\pi}\sqrt[4]{x(t-x)}}
+
\ord\(k_0e^{-2p_2} + k_0 e^{-p_2}(\ln x)^{-1}
\).
\end{equation}
Note that $k_0 = \frac{x}{2\sqrt{x(t-x)}}\asymp\frac{x}{\ln x}$ and thus the $\ord$ term in \eqref{E_part_III_interm} is of the order $\ord\(\frac{x}{e^{2p_2}\,\ln x} + \frac{x}{e^{p_2}\,(\ln x)^2}\),$ i.e. growing. We thus prefer to rewrite \eqref{E_part_III_interm} by factoring out the main term,
\begin{equation*}
\mathcal{E}(t,x)
=
\frac{2k_0r(ik_0)e^{2\sqrt{x(t-x)}}}{\sqrt{\pi}\sqrt[4]{x(t-x)}}
\(1+
\ord\((\ln x)^{-1} + e^{-p_2}\)
\).
\end{equation*}
\subsection*{Functions $\mathcal{N}(t,x)$ and $\rho(t,x).$}
An analogue of formula \eqref{M_2_0} gives us
\begin{multline*}
M_{err}(0;t,x)
=
I - \frac{1}{2\pi i}\int_{\Sigma_{err}}\!\!\!\(J_{err}(z;t,x)-I\)\frac{dz}{z}
\\
- \frac{1}{2\pi i}\int_{\Sigma_{err}}\!\!\!\(M_{err,+}(z;t,x)-I\)\(J_{err}(z;t,x)-I\)\frac{dz}{z}
\\
=
I - \frac{1}{2\pi i}\int_{\Sigma_{err}}\(J_{err}(z;t,x)-I\)\frac{dz}{z}
+
\ord\(\frac{e^{4\sqrt{x(t-x)}}}{k_0^{2m}\,\sqrt{x(t-x)}}\)
\\
=
I - \frac{1}{2\pi i}\int_{\Sigma_{err}}\(J_{err}(z;t,x)-I\)\frac{dz}{z}
+
\ord\(e^{-2p_2}\),
\end{multline*}
and hence for the elements $M_{err,11},$ $M_{err,12}$ of the first row of the matrix $M_{err}$ we obtain
\begin{multline}\label{M_err_11}
M_{err, 11}(0;t,x)
=
1 - \frac{1}{2\pi i}\int_{\mathbb{R}\setminus[-1,1]}|r(k_0z)|^2\frac{dz}{z}
+
\ord\(e^{-2p_2}\)
\\
=
1 + \ord\(k_0^{-2m} + e^{-2p_2}\) 
=
1 + \ord\(e^{-2p_2}\),
\end{multline}
\begin{multline*}
M_{err, 12}(0;t,x)
=
\frac{-1}{2\pi i}\int_{\partial\Omega_u}r(k_0z)e^{2\sqrt{x(t-x)}}\frac{1}{2\pi i}\int_{-\infty}^{+\infty}\frac{e^{-s^2}ds}{s-\zeta_u}\frac{dz}{z}
+
\frac{1}{2\pi i}\int_{\mathbb{R}\setminus[-1,1]}\!\!r(k_0z)e^{-2i\theta(t,x;k_0z)}\frac{dz}{z}
\\
+ \ord\(e^{-2p_2}\)
=: B_1 + B_2 + \ord\(e^{-2p_2}\).
\end{multline*}
Here $B_2$ admits the estimate $B_2 = \ord(k_0^{-m}) = \ord(x^{-m}\,\ln^m x),$ and we elaborate on the term $B_1$ similarly as we did in \eqref{I_estimate}, by substituting formula \eqref{integral_estimate} and computing the residue at $z=i.$ Thus
\begin{multline}\label{M_err_12}
M_{err, 12}(0;t,x) 
= \frac{-r(ik_0)e^{2\sqrt{x(t-x)}}}{2\sqrt{\pi}\sqrt[4]{x(t-x)}}
+
\ord\(
\frac{k_0^{-m}e^{2\sqrt{x(t-x)}}}{\(x(t-x)\)^{3/4}}
+ \frac{1}{k_0^m} + e^{-2p_2}\)
\\
= \frac{-r(ik_0)e^{2\sqrt{x(t-x)}}}{2\sqrt{\pi}\sqrt[4]{x(t-x)}}
+
\ord\(
\frac{e^{-p_2}}{\ln x} 
+ \frac{\ln^mx}{x^m} + e^{-2p_2}\)
\\
= \frac{-r(ik_0)e^{2\sqrt{x(t-x)}}}{2\sqrt{\pi}\,\sqrt[4]{x(t-x)}}
+
\ord\(
\frac{e^{-p_2}}{\ln x} 
+ e^{-2p_2}\).
\end{multline}
Note that the main term in \eqref{M_err_12} is of the order $e^{-p_2}.$
Now we use an analogue of formulae \eqref{rho_N_from_M},
$\rho(t,x) = -2M_{err,11}(0;t,x)M_{err,12}(0;t,x),$
$\mathcal{N}(t,x) = 1 - 2|M_{err, 12}(0;t,x)|^2,$
where we substitute formulae \eqref{M_err_11}, \eqref{M_err_12}.
We thus obtain
\[
\rho(t,x) = \frac{r(ik_0)e^{2\sqrt{x(t-x)}}}{\sqrt{\pi}\sqrt[4]{x(t-x)}}
+
\ord\(\frac{e^{-p_2}}{\ln x} + e^{-2p_2}
\),
\]
\[
\mathcal{N}(t,x)
=
1 - \frac{|r(ik_0)|^2e^{4\sqrt{x(t-x)}}}{2\pi\sqrt{x(t-x)}}
+
\mathcal{O}\(
\frac{e^{-2p_2}}{\ln x} + e^{-3p_2}
\).
\]
This completes the proof of Theorem \ref{thm_really_close_light_cone}, part III.

\subsubsection{Proof of Theorem \ref{thm_really_close_light_cone}, part IV, first half.}\label{sect_refined_analysis_1}

In this section we consider one half of the domain \eqref{domain_part_IV}, namely (below, $n=0, 1, 2, 3\ldots$)
\begin{equation}\label{domain_part_IV_1}
x + \frac{1}{4x}\(m\ln x + (n-m)\ln\ln x\)^2 \leq t \leq x + \frac{1}{4x}\(m\ln x + (n-m+\frac12)\ln\ln x\)^2.
\end{equation}

\begin{lemma}\label{lem_beta_estimate}
For $(t,x)$ in the domain \eqref{domain_part_IV_1} one has
\[
n-m + \ord((\ln\ln x)^{-1}) \leq \beta(t,x) \leq n-m+\frac12 + \ord((\ln\ln x)^{-1}),\quad \mbox{ as }x\to\infty,
\]
where $\beta=\beta(t,x)$ is defined by formula \eqref{parametrization}.
\end{lemma}
\begin{proof}
Indeed, \eqref{domain_part_IV_1} implies
\begin{equation}\label{domain_consequence}
m \ln x + (n-m)\ln\ln x
\leq 2\sqrt{x(t-x)} \leq m \ln x + (n-m+\frac12)\ln\ln x,
\end{equation}
and hence 
\begin{equation}\label{domain_consequence_2}
(n-m+\frac12)\ln \ln x \leq \beta(t,x) \ln \sqrt{x(t-x)} \equiv 2\sqrt{x(t-x)} - m\ln x \leq (n-m+\frac12)\ln \ln x.
\end{equation}
Estimating $\ln\sqrt{x(t-x)}$ from formula \eqref{domain_consequence} and substituting it into \eqref{domain_consequence_2}, we obtain the statement of the Lemma.
\end{proof}

Define
\[
M_{appr}^{(1)}(z;t,x)
=
\begin{cases}
N(z;t,x)\cdot\(\dfrac{z-i}{z+i}\)^{n\sigma_3}, \quad z\in\mathbb{C}\setminus\(\Omega_{u}\cup\Omega_{d}\),
\\
N(z;t,x)\cdot B_u(z;t,x)\Delta^{(1)}(\zeta_u)L_{n}(\zeta_u)r(k_0z)^{-\frac{\sigma_3}{2}}e^{-\sqrt{x(t-x)}\sigma_3},
\quad z\in\Omega_{u},
\\
N(z;t,x)\cdot B_{d}(z;t,x) \cdot\sigma\Delta^{(1)}(\zeta_d)L_n(\zeta_d)\sigma \cdot \ol{r(k_0z)}^ {\,\frac{\sigma_3}{2}}e^{\sqrt{x(t-x)}\sigma_3},\quad z \in \Omega_{d}.
\end{cases}
\]
Here $\Delta^{(1)}(\zeta)=\begin{pmatrix}1 & \dfrac{\sqrt{\pi}\,n!}{2\pi i\cdot 2^n\,\zeta} \\ 0 & 1\end{pmatrix}$ and hence 
the factor $\Delta^{(1)}(\zeta) L_n(\zeta)$ has the ``improved'' asymptotics \eqref{improved_asymp_2} for large $\zeta$ compared to the asymptotics of $L_n(\zeta).$
However, $\Delta^{(1)}(\zeta)$
has a simple pole at $\zeta=0,$ which corresponds to a simple pole of $M^{(1)}_{appr}(z;t,x)$ at the points $z = \pm i, $ and to cancel the latter
we introduce the hitherto unknown function $N(z) = N(z;t,x),$ for which we look in the form
\[
N(z) = \begin{pmatrix}
1 + \dfrac{i a}{z+ i} & \dfrac{i b}{z-i}
\\
\dfrac{i\,\ol{b}}{z+i} & 1 - \dfrac{i\,\ol{a}}{z-i}
\end{pmatrix},
\]
where $a, b$ are some yet unknown complex coefficients (which might depend on the parameters $t,x$).
The condition that $N(z) B_u(z) \Delta^{(1)}(\zeta_u)$ is regular at the point $z = i$ is equivalent
to the following system of linear equations for $a, b:$
\[
\(1 + \dfrac{a}{2}\)\cdot \dfrac{\phi^2\sqrt{\pi}\,n!}{2\pi\cdot 2^n\cdot\sqrt[4]{x(t-x)}} - b = 0,
\qquad
a + \dfrac{b}{2}\cdot \dfrac{\sqrt{\pi}\,n!\,\ol{\phi}^2}{2\pi\cdot 2^n\cdot\sqrt[4]{x(t-x)}} = 0,
\]
where
\begin{equation}\label{phi}
\phi = \phi(t,x) = e^{\frac{\pi i n}{2}}\cdot 2^{-n}\cdot\(x(t-x)\)^{\frac{-n}{4}}
\cdot
\sqrt{r(i k_0)}\cdot e^{\sqrt{x(t-x)}}.
\end{equation}
Introducing the short-hand notation
\begin{equation}\label{psi}\psi = \psi(t,x) :=\frac{\phi^2\cdot\sqrt{\pi}\,n!}{4\pi\cdot 2^n\cdot\sqrt[4]{x(t-x)}},\end{equation}
the above system can be rewritten in the form
\[
\begin{cases}
a + \ol{\psi}b = 0,
\\
\psi a - b = -2\psi,
\end{cases}
\mbox{ and hence }
\begin{cases}
a = \dfrac{-2|\psi|^2}{1 + |\psi|^2},
\\
b = \dfrac{2\psi}{1 + |\psi|^2}.
\end{cases}
\]
Note that $a$ is real and $(a+1)^2 + |b|^2 = 1,$ and hence $\det N(z)\equiv 1.$ Besides, $a, b$ are uniformly bounded in $t,x$, and hence $N(z)$ and $N(z)^{-1}$ are bounded away from $z=\pm i.$

Let us examine the jumps of the matrix error function 
$$M^{(1)}_{err}(z;t,x) = M^{(2)}(z;t,x) M_{appr}^{(1)}(z;t,x)^{-1},$$
with the jump matrix $J_{err}^{(1)}(z;t,x)=J_{err}^{(1)}(z)$ such that $M^{(1)}_{err, -}(z;t,x) = M^{(1)}_{err, +}(z;t,x)J^{(1)}_{err}(z;t,x).$ For 
$z\in\partial\Omega_{u}$ we have

\begin{align*}
J^{(1)}_{err}(z) 
&= 
N(z)\cdot\(\frac{z-i}{(z+i)\zeta_u}\)^{n\sigma_3}
r(k_0z)^{\frac{\sigma_3}{2}}
e^{\sqrt{x(t-x)}}
\Delta^{(1)}(\zeta_u)L_n(\zeta_u)
r(k_0z)^{\frac{-\sigma_3}{2}}
e^{-\sqrt{x(t-x)}\sigma_3}\(\frac{z+i}{z-i}\)^{n\sigma_3}
\\
&=
\begin{pmatrix}
1 + \mathcal{O}\(\frac{1}{\sqrt[4]{x(t-x)}}\)
&
\mathcal{O}(1)\cdot \(\sqrt{x(t-x)}\)^{\beta(t,x)+m-n-1}
\\
\mathcal{O}(1)\cdot \(\sqrt{x(t-x)}\)^{-\beta(t,x)-m+n-\frac12}
&
1 + \mathcal{O}\(\frac{1}{\sqrt[4]{x(t-x)}}\)
\end{pmatrix}.
\end{align*}
 
We see that the $12$ entry vanishes more rapidly compared to \eqref{J_err_not_refined}.
Since $\beta(t,x)$ satisfies the estimates from Lemma \ref{lem_beta_estimate} and $\ln\sqrt{x(t-x)}\sim \ln \ln x$ as $x\to\infty,$ it follows that 
$
(\sqrt{x(t-x)})^{\ord((\ln\ln x)^{-1})} = \ord(1) 
$
and hence
$$
J^{(1)}_{err}(z) = I + \mathcal{O}\(\frac{1}{\sqrt[4]{x(t-x)}}\)
=
I + \mathcal{O}\(\frac{1}{\sqrt{\ln x}}\)
=
I + \mathcal{O}\(\frac{1}{\sqrt{\ln t}}\),\qquad z\in\partial\Omega_{u}.
$$
Similarly, 
$
J^{(1)}_{err}(z) = I + \mathcal{O}\(\dfrac{1}{\sqrt{\ln t}}\),\ z\in\partial\Omega_{d}.
$
Furthermore,
$$
J_{err}^{(1)}(z) = I + \mathcal{O}\(\frac{1}{k_0^m}\)
=
I + \mathcal{O}\(\frac{1}{k_0^m}\)
=
I + \mathcal{O}\(\frac{\ln^mx}{x^m}\),
\qquad
z \in (-\infty, -1)\cup(1, +\infty).
$$
Overall, 
\begin{equation}\label{estimate_J_err_1}
J^{(1)}_{err}(z) = 
I + \mathcal{O}\(\dfrac{1}{\sqrt{\ln x}}\)
=
I + \mathcal{O}\(\dfrac{1}{\sqrt{\ln t}}\),\ z\in\Sigma_{err}:=\partial\Omega_u\cup\partial\Omega_d\cup(-\infty, -1)\cup(1, +\infty),\end{equation}
and by standard arguments we conclude that 
$M_{err}^{(1)}(z; t, x) = I + \mathcal{O}\(\dfrac{1}{\sqrt{\ln t}}\),$ uniformly in $z\in\mathbb{C}.$

\subsubsection*{Reconstruction of $\mathcal{E}, \mathcal{N} ,\rho.$}
Recall \eqref{E1}, \eqref{F1} that 
\[
\mathcal{E}(t, x) = -4\,i\,\lim\limits_{k\to\infty} k(M(k; t, x) - I)_{12},\]\[ 
\begin{pmatrix}
\mathcal{N}(t, x) & \rho(t, x)
\\
\ol{\rho(t, x)} & -\mathcal{N}(t, x)
\end{pmatrix}
=
M(+i 0; t, x)\sigma_3 M(+i 0; t, x)^{-1}.
\]
\noindent Tracing back the connection between $M^{(2)}$ and $M$, it follows that
\begin{align*}
&
\mathcal{E}(t, x) = -2\,i\,\sqrt{\frac{x}{t-x}}\lim\limits_{z\to\infty} z(M^{(2)}(z; t, x) - I)_{12},
\\ 
&\begin{pmatrix}
\mathcal{N}(t, x) & \rho(t, x)
\\
\ol{\rho(t, x)} & -\mathcal{N}(t, x)
\end{pmatrix}
=
M^{(2)}(z=+i 0; t, x)\sigma_3 M^{(2)}(z=+i 0; t, x)^{-1}.
\end{align*}
Substituting here $M^{(2)}(z;t,x) = M_{err}(z;t,x)M_{appr}^{(1)}(z;t,x),$ we get
\[
\mathcal{E}(t, x) = -2i\sqrt{\frac{x}{t-x}}\cdot i b(t, x) + \mathcal{E}_{err}^{(1)}(t, x)
=
2\sqrt{\frac{x}{t-x}}\cdot \frac{2\psi}{1 + |\psi|^2} + \mathcal{E}_{err}^{(1)}(t, x),
\]
where
\[
\mathcal{E}_{err}^{(1)}(t, x) = -2\,i\,\sqrt{\frac{x}{t-x}}\cdot \lim\limits_{z\to\infty} z(M_{err}^{(1)}(z; t, x) - I)_{12}.
\]
From the estimates \eqref{estimate_J_err_1} on $J_{err}^{(1)}$ it follows that 
\[
\mathcal{E}^{(1)}_{err}(t, x) = \mathcal{O}\(\frac{x}{\sqrt{x(t-x)}}\)\cdot \frac{1}{\sqrt[4]{x(t-x)}}
=
\mathcal{O}\(\frac{x}{(x(t-x))^{3/4}}\)
=
\mathcal{O}\(\frac{x}{(\ln x)^{3/2}}\).
\]
and elaborating on the expression \eqref{psi} for $\psi,$
\begin{equation}\label{psi_full}
\psi = \frac{(-1)^n\,n!\,r(i k_0)e^{2\sqrt{x(t-x)}}}{\sqrt{\pi}\,2^{3n+2}\,(x(t-x))^{\frac{n}{2}+\frac14}}
\end{equation}
we hence obtain 
\begin{align}\label{E_expr_1}
\mathcal{E}(t, x)
=
2\sqrt{\frac{x}{t-x}}
\(
\frac{ (-1)^n \cdot e^{i\arg r(i k_0)} }
{
\cosh\(2\sqrt{x(t-x)} - (n+\frac12)\ln\sqrt{x(t-x)} + \chi_n(k_0) \)}
+
\mathcal{O}\(\frac{1}{\sqrt[4]{x(t-x)}}\)
\),
\end{align}
where 
$
\chi_n(k_0) =\ln\left( \dfrac{n!\cdot|r(i k_0)|} {\sqrt{\pi}\cdot2^{3n+2}}\right).
$
Expression \eqref{E_expr_1} coincides with formula \eqref{E_res_as_sol}, and thus the part IV of Theorem \ref{thm_really_close_light_cone} is proved for the function $\mathcal{E}(t,x)$ in domain \eqref{domain_part_IV_1}.

As for the $\mathcal{N}, \rho,$
we have
\begin{align*}
\begin{pmatrix}
\mathcal{N}(t, x) & \rho(t, x) \\ 
\ol{\rho(t, x)} & -\mathcal{N}(t, x)
\end{pmatrix}
&=
M_{err}^{(1)}(+i 0; t, x)
M_{appr}^{(1)}(+i 0; t, x)
\sigma_3
\(M_{appr}^{(1)}(+i 0; t, x)\)^{-1}
M_{err}^{(1)}(+i 0; t, x)^{-1}
\\
&=
M_{err}^{(1)}(+i 0; t, x)
\begin{pmatrix}1 + a & -b \\ \ol{b} & 1 + a\end{pmatrix}
\sigma_3
\begin{pmatrix}1 + a & b \\ -\ol{b} & 1 + a\end{pmatrix}
M_{err}^{(1)}(+i 0; t, x)^{-1}
\\
&=
M_{err}^{(1)}(+i 0; t, x)
\begin{pmatrix}(1 + a)^2 - |b|^2 & 2b(1+a) \\ 2\ol{b}(1+a) & -(1 + a)^2 + |b|^2 \end{pmatrix}
M_{err}^{(1)}(+i 0; t, x)^{-1}
\\
&=M_{err}^{(1)}(+i 0; t, x)
\begin{pmatrix}1 - 2|b|^2 & 2b(1+a) \\ 2\ol{b}(1+a) & -1 + 2 |b|^2 \end{pmatrix}
M_{err}^{(1)}(+i 0; t, x)^{-1},
\end{align*}
and thus $\mathcal{N}(t, x) = 1 - 2|b|^2 + \mathcal{O}\(\frac{1}{\sqrt[4]{x(t-x)}}\),
$
$\rho(t, x) = 2b(1+a) + \mathcal{O}(\frac{1}{\sqrt{\ln x}}).$
Thus, part IV of Theorem \ref{thm_really_close_light_cone} is proved for the domain \eqref{domain_part_IV_1}.
Note also that $b$ is close to $0$ away from the peaks of the solitons and is close to $1$ near the peaks. Hence, $\mathcal{N}$ is close to $1$ away from the peaks and is close to $-1$ near the peaks.

\subsubsection{Proof of Theorem \ref{thm_really_close_light_cone}, part IV, second half.}\label{sect_refined_analysis_2}
In this section we consider the other half of the domain \eqref{domain_part_IV} (but, for convenience, with the index $n\mapsto n-1$), namely (below, $n = 1, 2, 3, \ldots$)
\begin{equation}\label{domain_part_IV_2}
x + \frac{1}{4x}\(m\ln x + (n-m-\frac12)\ln\ln x\)^2 \leq t \leq x + \frac{1}{4x}\(m\ln x + (n-m)\ln\ln x\)^2.
\end{equation}
\begin{lemma}\label{lem_beta_estimate_2}
For $(t,x)$ in the domain \eqref{domain_part_IV_2} one has
\[
n-m-\frac12 + \ord((\ln\ln x)^{-1}) \leq \beta(t,x) \leq n-m + \ord((\ln\ln x)^{-1}),\quad \mbox{ as }x\to\infty,
\]
where $\beta=\beta(t,x)$ is defined by formula \eqref{parametrization}.
\end{lemma}
\begin{proof}
The proof is exactly the same as in Lemma \ref{lem_beta_estimate}.
\end{proof}

Now we want to ``improve'' the $21$ term in $L_n(\zeta)$ as in formula \eqref{improved_asymp_2}. Define
\[
M_{appr}^{(2)}(z;t,x)
=
\begin{cases}
\widehat N(z;t,x)\cdot\(\frac{z-i}{z+i}\)^{n\sigma_3}, \quad z\in\mathbb{C}\setminus\(\Omega_{u}\cup\Omega_{d}\),
\\
\widehat N(z;t,x)\cdot B_u(z;t,x)\cdot\widehat\Delta(\zeta_u)L_{n}(\zeta_u)\cdot r(k_0z)^{-\frac{\sigma_3} {2}}e^{-\sqrt{x(t-x)}\sigma_3},
\quad z\in\Omega_{u},
\\
\widehat N(z;t,x)\cdot B_{d}(z;t,x)\cdot\sigma\widehat\Delta(\zeta_d)L_n(\zeta_d)\sigma\cdot r^*(k_0z)^{\frac{\sigma_3}{2}}e^{\sqrt{x(t-x)}\sigma_3},\quad z \in \Omega_{d}.
\end{cases}
\]
Here $\widehat\Delta(\zeta)=\begin{pmatrix}1 & 0 \\ \dfrac{2\pi i\cdot 2^{n-1}\cdot n}{\sqrt{\pi}\,n!\,\zeta} & 1\end{pmatrix}$ has a simple pole at $\zeta=0,$ which corresponds to a simple pole of $M_{appr}^{(2)}$ at the points $z = \pm i, $ and to cancel it we introduce the yet unknown function $\widehat N(z;t,x),$ for which we look in the form
\[
\widehat N(z;t,x) = \begin{pmatrix}
1 + \dfrac{i\, \widehat a}{z - i} & \dfrac{i\, \widehat b}{z+i}
\\
\dfrac{i\,\ol{\widehat b}}{z-i} & 1 - \dfrac{i\,\ol{\widehat a}}{z+i}
\end{pmatrix}.
\]
The condition that $\widehat N(z;t,x) B_u(z;t,x) \widehat \Delta(\zeta_u)$ is regular at the point $z = i$ is equivalent
to the following system of linear equations for the unknown coefficients $\widehat a = \widehat a(t,x),\ \widehat b = \widehat b(t,x)\!:$
$$
\begin{cases}
\(1 - \dfrac{\widehat a}{2}\)\cdot \dfrac{2\pi\cdot 2^{n-1}\cdot n}{\ol{\phi}^2\sqrt{\pi}\,n!\cdot\sqrt[4]{x(t-x)}} + \widehat b = 0,
\\\\
\widehat a + \dfrac{\widehat b}{2}\cdot \dfrac{2\pi\cdot 2^{n-1}n}{\sqrt{\pi}\,n!\,\phi^2\cdot\sqrt[4]{x(t-x)}} = 0,
\end{cases}
\quad
\mbox{ or }
\begin{cases}
\widehat a +\widehat\psi\widehat b = 0,
\\
\ol{\widehat \psi}\cdot \widehat a - \widehat b = 2 \ol{\widehat \psi}.
\end{cases}
\quad
\mbox{ and }
\begin{cases}
\widehat a  = \dfrac{2|\widehat \psi|^2}{1 + |\widehat \psi|^2},
\\\\
\widehat b = \dfrac{-2 \ol{\widehat \psi}}{1 + |\widehat \psi|^2},
\end{cases}
$$
where $\phi$ is given by formula \eqref{phi},
\[
\phi = e^{\frac{\pi i n}{2}}\cdot 2^{-n}\cdot\(x(t-x)\)^{\frac{-n}{4}}
\cdot
\sqrt{r(i k_0)}\cdot e^{\sqrt{x(t-x)}},
\]
and 
\begin{equation}\label{psi_hat}\widehat \psi=\widehat \psi(t,x):=\frac{\pi\cdot 2^{n-1}\cdot n}{\phi^2\cdot\sqrt{\pi}\,n!\cdot\sqrt[4]{x(t-x)}}.\end{equation}
As in the previous section, $\widehat a$ is real and $(\widehat a+1)^2 + |\widehat b|^2 = 1,$ and hence $\det \widehat N(z;t,x)\equiv 1.$ Besides, $\widehat a, \widehat b$ are uniformly bounded in $t,x$, and hence $\widehat N(z;t,x)$ and $\widehat N(z;t,x)^{-1}$ are bounded away from $z=\pm i.$

Let us examine the jumps of the matrix  error function 
$$M^{(2)}_{err}(z;t,x) = M(z;t,x) M_{appr}^{(2)}(z;t,x)^{-1},$$
whose jump matrix $J^{(2)}_{err}(z;t,x) = J^{(2)}_{err}(z)$ satisfies the relation $M^{(2)}_{err, -}(z;t,x) = M^{(2)}_{err, +}(z;t,x)J^{(2)}_{err}(z;t,x).$ Similarly as in the previous subsection \ref{sect_refined_analysis_1},
for $z\in\partial\Omega_{u}$ we have
\begin{align*}
J^{(2)}_{err}(z;t,x) 
&= 
\begin{pmatrix}
1 + \mathcal{O}\(\frac{1}{\sqrt[4]{x(t-x)}}\)
&
\mathcal{O}(1)\cdot \(\sqrt{x(t-x)}\)^{\beta(t,x)+m-n-\frac12}
\\
\mathcal{O}(1)\cdot \(\sqrt{x(t-x)}\)^{-\beta(t,x)-m+n-1}
&
1 + \mathcal{O}\(\frac{1}{\sqrt[4]{x(t-x)}}\)
\end{pmatrix}.
\end{align*}

In the above matrix, the $21$ term has an improved decay compared to \eqref{J_err_not_refined}. 
In view of Lemma \ref{lem_beta_estimate_2}, the error matrix admits the estimate
\begin{equation}\label{estimate_J_err_2}
J^{(2)}_{err}(z;t,x) = I + \mathcal{O}\(\frac{1}{\sqrt[4]{x(t-x)}}\)
=
I + \mathcal{O}\(\frac{1}{\sqrt{\ln x}}\)
=
I + \mathcal{O}\(\frac{1}{\sqrt{\ln t}}\),\quad z\in\Sigma_{err}.
\end{equation}
Since the jump contour $\Sigma_{err} = \partial\Omega_u\cup\partial\Omega_d\cup(-\infty, -1)\cup(1, +\infty)$ of the error function $M_{err}^{(2)}$ is fixed for all $t, x,$ from here we conclude by standard arguments that 
$M_{err}^{(2)}(z; t, x) = I + \mathcal{O}\(\frac{1}{\sqrt{\ln t}}\),$ uniformly in $z\in\mathbb{C}.$

\subsubsection*{Reconstruction of $\mathcal{E}, \mathcal{N} ,\rho.$}
Similarly as in the previous subsection \ref{sect_refined_analysis_1}, 
\begin{align*}
&\mathcal{E}(t, x) = -2\,i\,\sqrt{\frac{x}{t-x}}\lim\limits_{z\to\infty} z(M^{(2)}(z; t, x) - I)_{12},
\\
&
\begin{pmatrix}
\mathcal{N}(t, x) & \rho(t, x)
\\
\ol{\rho(t, x)} & -\mathcal{N}(t, x)
\end{pmatrix}
=
M^{(2)}(+i 0; t, x)\sigma_3 M^{(2)}(+i 0; t, x)^{-1}.
\end{align*}
Substituting here $M^{(2)}(z;t,x) = M_{err}(z;t,x)M_{appr}^{(2)}(z;t,x),$ we get
\[
\mathcal{E}(t, x) = -2i\sqrt{\frac{x}{t-x}}\cdot i\, \widehat b(t, x) + \mathcal{E}_{err}^{(2)}(t, x)
=	
2\sqrt{\frac{x}{t-x}}\cdot \frac{-2\ol{\widehat \psi}}{1 + |\widehat\psi|^2} + \mathcal{E}_{err}^{(2)}(t, x),
\]
where
\[
\mathcal{E}^{(2)}_{err}(t, x) = -2\,i\,\sqrt{\frac{x}{t-x}}\cdot \lim\limits_{z\to\infty} z(M_{err}^{(2)}(z; t, x) - I)_{12}.
\]
From the estimate \eqref{estimate_J_err_2} on $J_{err}^{(2)}$ it follows that 
\[
\mathcal{E}_{err}^{(2)}(t, x) = \mathcal{O}\(\frac{x}{\sqrt{x(t-x)}}\)\cdot \frac{1}{\sqrt[4]{x(t-x)}}
=
\mathcal{O}\(\frac{x}{(x(t-x))^{3/4}}\)
=
\mathcal{O}\(\frac{x}{(\ln x)^{3/2}}\).
\]
and elaborating on the expression \eqref{psi_hat} for $\widehat\psi$,
\begin{equation}\label{psi_hat_full}
\widehat\psi
=
\frac{(-1)^n\,2^{3n-1}\,n\,\sqrt{\pi}\,(x(t-x))^{\frac{n}{2}-\frac14}}{n!\,r(ik_0)\,e^{2\sqrt{x(t-x)}}},
\end{equation}
we thus obtain 
\begin{align}\label{E_expr_2}
\mathcal{E}(t, x)
=
2\sqrt{\frac{x}{t-x}}
\(
\frac{ (-1)^{n-1} \cdot e^{i\arg r(i k_0)} }
{
\cosh\(2\sqrt{x(t-x)} - (n-\frac12)\ln\sqrt{x(t-x)} + \widehat \chi_n(k_0) \)}
+
\mathcal{O}\(\frac{1}{\sqrt[4]{x(t-x)}}\)
\)
\end{align}
for $n\geq 1,$ where
$
\widehat\chi_n(k_0) =\ln\left( \dfrac{(n-1)!\cdot|r(i k_0)|} {\sqrt{\pi}\cdot2^{3n-1}}\right).
$
Expression \eqref{E_expr_2} coincides with formulae \eqref{E_res_as_sol}, \eqref{E_expr_1} (after changing back $n\mapsto n+1$), and thus expression \eqref{E_expr_1} is valid not only in the domain \eqref{domain_part_IV_1}, but also in \eqref{domain_part_IV_2}, i.e. in the full domain \eqref{domain_part_IV}.
This completes the proof of part IV of Theorem~\ref{thm_really_close_light_cone} for the function $\mathcal{E}(t,x).$

As for the functions $\mathcal{N}(t,x), \rho(t,x),$
similarly as in subsection \ref{sect_refined_analysis_1}, we have
\begin{align*}
\begin{pmatrix}
\mathcal{N}(t, x) & \rho(t, x) \\ 
\ol{\rho(t, x)} & -\mathcal{N}(t, x)
\end{pmatrix}
&=
M_{err}^{(2)}(+i 0; t, x)
M_{appr}^{(2)}(+i 0; t, x)
\sigma_3
\(M_{appr}^{(2)}(+i 0; t, x)\)^{-1}
M_{err}^{(2)}(+i 0; t, x)^{-1}
\\
&=
M_{err}^{(2)}(+i 0; t, x)
\begin{pmatrix}1 - \widehat a & \widehat b \\ -\ol{\widehat b} & 1 -\widehat a\end{pmatrix}
\sigma_3
\begin{pmatrix}1 - \widehat a & -\widehat b \\ \ol{\widehat b} & 1 - \widehat a\end{pmatrix}
M_{err}^{(2)}(+i 0; t, x)^{-1}
\\
&=M_{err}^{(1)}(+i 0; t, x)
\begin{pmatrix}1 - 2|\widehat b|^2 & -2\widehat b(1 -\widehat a) \\ -2\ol{\widehat b}(1-\widehat a) & -1 + 2 |\widehat b|^2 \end{pmatrix}
M_{err}^{(1)}(+i 0; t, x)^{-1},
\end{align*}
and thus $\mathcal{N}(t, x) = 1 - 2|\widehat b|^2 + \mathcal{O}\(\frac{1}{\sqrt[4]{x(t-x)}}\),
$
$\rho(t, x) = -2\widehat b(1-\widehat a) + \mathcal{O}(\frac{1}{\sqrt{\ln x}}).$ Since $\widehat\psi$ in \eqref{psi_hat_full} is equal to $-\frac{1}{\psi},$ where $\psi$ is as in \eqref{psi_full} if we change there $n\mapsto n-1$,
then the asymptotics for $\mathcal{N}(t,x), \rho(t,x)$ from subsection \ref{sect_refined_analysis_1} are valid not only in the region \eqref{domain_part_IV_1}, but also in the region \eqref{domain_part_IV_2}. 
This completes the proof of part IV of Theorem \ref{thm_really_close_light_cone}.

\section{Asymptotic analysis. Region $\sigma t\leq x\leq (1-\sigma)t,\ \sigma\in(0,\frac12).$ Proof of Theorem \ref{thm_tail}}\label{sect_thm_2}

Introduce some functions needed in the course of the asymptotic analysis.

\subsection*{Phase functions.}
We will need in total three different phase functions, namely
\begin{equation}\label{def:g}
\widetilde\theta(k;k_0,\tau) = \tau\!\! \( \!k - \frac{k_0^2}{k} \),
\ \
h(k;k_0,\tau) = \tau\!\! \(\! k + \frac{k_0^2}{k} \),
\ \
g(k;k_0,\tau) =
\begin{cases}
h(k;k_0,\tau), \ \ \ \, |k| > k_0,
\\
-h(k;k_0,\tau), \ |k| < k_0.
\end{cases}
\end{equation}
Here $\tau = t-x,
\quad k_0 = \sqrt{\frac{x}{4(t-x)}}$ (see \eqref{k_0}).
\\Note that $\widetilde\theta(k;k_0,\tau) = \theta(t,x;k),$ where $\theta$ is defined in formula \eqref{theta}.

\subsection{Functions $\delta, F, G, H$}
{\bf Function $\delta$.} The function $\delta$ is defined as a solution to the scalar conjugation problem
\[
\frac{\delta_+(k; k_0)}{\delta_-(k; k_0)} = 1 + |r(k)|^2,\quad k\in (-\infty, -k_0)\cup(k_0, +\infty),
\]
with the normalisation $\delta(k; k_0) \to 1\mbox{ as } k\to\infty,$
and is given explicitly by the formula
\begin{equation}\label{def:delta}
\delta(k;k_0) = \exp\left\{
\frac{1}{2\pi i}
\(\int_{-\infty}^{-k_0}+\int_{k_0}^{+\infty}\)
\frac{\ln(1+|r(s)|^2)\ ds}{s - k}
\right\}.
\end{equation}

\medskip\medskip

\noindent
{\bf Function $G$.} The function $G$ is defined as a solution to the scalar conjugation problem
\[
\frac{G_+(k;k_0)}{G_-(k;k_0)} = 1 + \frac{1}{|r(k)|^2},\quad k\in (-k_0, k_0),\qquad G(k) \to 1\mbox{ as } k\to\infty,
\]
and is given explicitly by the formula
\begin{equation}\label{def:G}
G(k;k_0) = \exp\left\{
\frac{1}{2\pi i}
\int_{-k_0}^{k_0}
\frac{\ln\(1+|r(s)|^{-2}\)\ ds}{s - k}
\right\}.
\end{equation}

\medskip\medskip

\noindent
{\bf Function $H$.} The function $H$ is defined using the function $G$ as follows:
\begin{equation}\label{def:H}
H(k;k_0)
=
\begin{cases}
G(k;k_0) b(k),\quad \Im k > 0,
\\
\dfrac{G(k;k_0)}{b^*(k)},\quad \Im k < 0.
\end{cases}
\end{equation}
It satisfies the following jump condition:
\[H_+(k; k_0) = H_-(k; k_0) \frac{1}{1 + |r(k)|^{-2}},\quad k\in(-\infty, -k_0)\cup(k_0, +\infty).\]

\medskip\medskip
\noindent
{\bf Function $F$.}
\noindent Function $F$ is defined using $G$ as follows:
\begin{equation}\label{def:F}
F(k;k_0) = \begin{cases}
a(k)\delta(k; k_0)G(k;k_0),\quad |k|>k_0, \ \Im k>0,
\\
(a^*(k))^{-1}\delta(k;k_0)G(k;k_0), \quad |k|>k_0, \ \Im k<0,
\\
\dfrac{\delta(k;k_0)}{b(k) G(k;k_0)},\quad |k|<k_0,\  \Im k>0,
\\
\dfrac{b^*(k)\delta(k;k_0)}{G(k;k_0)},\quad |k|<k_0,\ \Im k<0.
\end{cases}
\end{equation}

\noindent It satisfies the conjugation conditions
\[
F_+(k;k_0) F_-(k;k_0) = \frac{\delta^2(k;k_0)}{r(k)}, \quad k\in C_{u},
\qquad
F_+(k;k_0) F_-(k;k_0) = r^*(k) \delta^2(k;k_0), \quad k\in C_{d},
\]
and has the asymptotics $F(k;k_0) \to 1$ as $k\to\infty.$
As before \eqref{Cud}, $C_{u}=C_{u}(k_0)$  and $C_{d}=C_{d}(k_0)$ are the upper and lower parts of the circle $|k|=k_0,$ oriented from the point $-k_0$ to the point $k_0,$ respectively.

\medskip
Now we proceed to the asymptotic analysis. We will make 6 transformations 
\begin{equation}\label{chain_transformations}
M \mapsto M^{(1)} \mapsto M^{(2)} \mapsto M^{(3)} \mapsto M^{(4)} \mapsto M^{(5)} \mapsto M^{(6)}
\end{equation}
to arrive at a RH problem for $M^{(6)},$ and then do the error analysis.

When performing the steps of transformations \eqref{chain_transformations}, we will not track how the appropriate pole conditions change. 
Instead, we 
do all the contour deformations first, 
and then track back how $M^{(6)}$ depends on $M$ (and thus on the Jost solutions), which allows us to state the precise pole conditions for $M^{(6)}.$

\subsection{Steps of transformations of the original RH problem \ref{mainRHP}}\label{sect_steps_transformations}
\subsubsection*{Step 1: Change the size of the circle from $\frac12$ to $k_0.$}
Define 
\[M^{(1)}(k;k_0,\tau) = M(t,x;k)\begin{cases}
\begin{pmatrix}
0 & -r(k) e^{-2i\theta} \\
\dfrac{1}{r(k)} e^{2i\theta} & 1
\end{pmatrix}^{\mathrm{sgn}(|k_0|-\frac12)},\quad |k| \in (\frac12, k_0), \ \Im k>0,
\\
\begin{pmatrix}
1 & \dfrac{-1}{r^*(k)}e^{-2i\theta} \\ r^*(k)e^{2i\theta} & 0
\end{pmatrix}^{\mathrm{sgn}(|k_0|-\frac12)},\quad |k|\in(\frac12, k_0), \ \Im k <0.
\end{cases}\]
where $\theta = \theta(t,x;k)=\widetilde\theta(k;k_0,\tau).$
This transformation effectively substitutes the jump over the circle $|k| = \frac12$ with the jump of the same form over the circle $|k| = k_0,$ and changes appropriately the jumps over the intervals $\pm(\frac12, k_0)$ (note that the transformation depends on whether $k_0$ is bigger than $\frac12$ or smaller.)

\subsubsection*{Step 2: changing the phase function $\theta$ to $g.$}
In the sequel, it will turn out that instead of the phase function $\theta$ it will be convenient for us to use a function, that has the same distribution of signs of its imaginary part on the complex plane as $\theta$ does, possess similar asymptotic behaviour as $\theta$ as $k\to0$ and $k\to\infty$, and at the same time satisfies the condition $g_+(k) + g_-(k) = 0$ for $k\in C_{u}\cup C_{d}.$ Such a function $g$ is defined in formula \eqref{def:g}.

We define
$M^{(2)}(k;k_0,\tau) = M^{(1)}(k;k_0,\tau) e^{i(g(k;k_0,\tau)-\theta(x,t;k))\sigma_3}.$
This turns the exponentially growing jump matrices for $M^{(1)}$ on the $C_{u}\cup C_{d}$ into oscillating ones for $M^{(2)}.$

\subsubsection*{Step 3: preparation for lens opening over $(-\infty, -k_0)\cup (k_0, +\infty)\cup C_{u}\cup C_{d}.$}

In this step, we would like to transform the oscillating jump matrix over $\mathbb{R}$ and over $C_{u}\cup C_{d}$ into jumps close to the identity matrix over some contours off the real line and off $C_{u}\cup C_{d}$, respectively.
This is done in two steps: the $\delta F^{-1}$-transformation and the lens opening.

For the purpose, define $M^{(3)}(k;k_0,\tau) = M^{(2)}(k;k_0,\tau)\delta(k;k_0)^{\sigma_3}F(k;k_0)^{-\sigma_3},$ where $\delta$ is defined in \eqref{def:delta} and $F$ is defined in \eqref{def:F}.
Then the jump for $M^{(3)}$ over $(-\infty, -k_0)\cup(k_0, +\infty)$ takes the form
(below, $g=g(k;k_0,\tau)$)
$$
M^{(3)}_-(k;k_0,\tau) = 
M^{(3)}_+(k;k_0,\tau)
\begin{pmatrix}
1 & 0 \\ 
\dfrac{-a(k)b^*(k)\ \delta_+^2(k;k_0)\ e^{2ig}}{F^2(k;k_0)} & 1
\end{pmatrix}
\begin{pmatrix}
1 & \dfrac{-a^*(k)b(k)F^2(k;k_0)e^{-2ig}}{\delta_-^2(k;k_0)}
\\ 0 & 1
\end{pmatrix}.
$$
Here we used $\frac{r^*(k)}{1 + r(k)r^*(k)} = a(k)b^*(k)$, cf. \eqref{aabb}.
Furthermore, over the $(-k_0, k_0)$  the jump is 
\[M_-^{(3)}(k;k_0,\tau) = M_+^{(3)}(k;k_0,\tau)
\begin{pmatrix}
1 & 0 \\ \dfrac{-\delta^2(k;k_0)\ e^{2ig}}{r(k)\,F^2(k;k_0)} & 1
\end{pmatrix}
\begin{pmatrix}
1 & \dfrac{-F^2(k;k_0)e^{-2ig}}{r^*(k)\delta^2(k;k_0)}
\\ 0 & 1
\end{pmatrix},\quad k\in(-k_0, k_0),\]
and over $C_{u},$ $C_{d}$ it takes the form
\[
M^{(3)}_-(k;k_0,\tau) = M_+^{(3)}(k;k_0,\tau)
\begin{pmatrix}
1 & 0 \\ \dfrac{-\delta^2(k)\,e^{2ig_+}}{r(k)\,F_+^2(k)} & 1
\end{pmatrix}
\begin{pmatrix}0 & -1 \\ 1 & 0\end{pmatrix},\quad k\in C_{u},
\]
and 
\[
M_-^{(3)}(k;k_0,\tau) = M_{+}^{(3)}(k;k_0,\tau)
\begin{pmatrix}
0 & 1 \\ -1 & 0
\end{pmatrix}
\begin{pmatrix}
1 & \dfrac{-F_-^2(k;k_0)e^{-2ig_-}}{r^*(k)\,\delta^2(k;k_0)}
\\ 0 & 1
\end{pmatrix},\quad k\in C_{d}.
\]
Here $g_{\pm}$ are the limiting values of the function $g$ on the contours $C_u, C_d,$ from the positive/negative side of the contour, respectively.

\subsubsection*{Step 4: lens opening.}
Now we are ready to open the lenses over $\mathbb{R}\cup C_{u} \cup C_{d}.$ It is convenient at the same time to flip the columns inside the circle $|k| = k_0$; the latter erases the jump across the circle and makes jumps expressed naturally in terms of the phase function $h$ \eqref{def:g}, which is our final phase function. 

Define appropriate contours $L_1, L_2, L_3, L_4, L_5, L_6$ needed for the lens opening as follows.
Let $L_5$ be an oriented composed curve that consists of two pieces: one starts at $-\infty + idk_0$ for some $d>0$ and ends at $-k_0$, another one starts at $k_0$ and ends at $+\infty + id k_0$ (see Figure \ref{Fig:M4}, left).
Next, $L_6$ is the curve symmetric to $L_5$ with respect to the real axis.

Contour $L_1$ starts at the point $k = k_0,$ ends at the point $k=-k_0,$ and is located above the circle $|k| = k_0.$
Contours $L_3$ starts at the point $k = -k_0,$ ends at the point $k = k_0,$ and is situated between the upper part of the circle $|k| = k_0$ and the real axis.
Contours $L_2, L_4$ are symmetric to $L_1, L_3$ with respect to the real axis, respectively, and are situated in the lower half-plane $\Im k<0.$

Introduce the following notation. If $\gamma_1, \ldots, \gamma_k$ are some contours that circumscribe a domain $\Omega,$ then denote $\Omega =: Int(\gamma_1, \ldots, \gamma_k).$

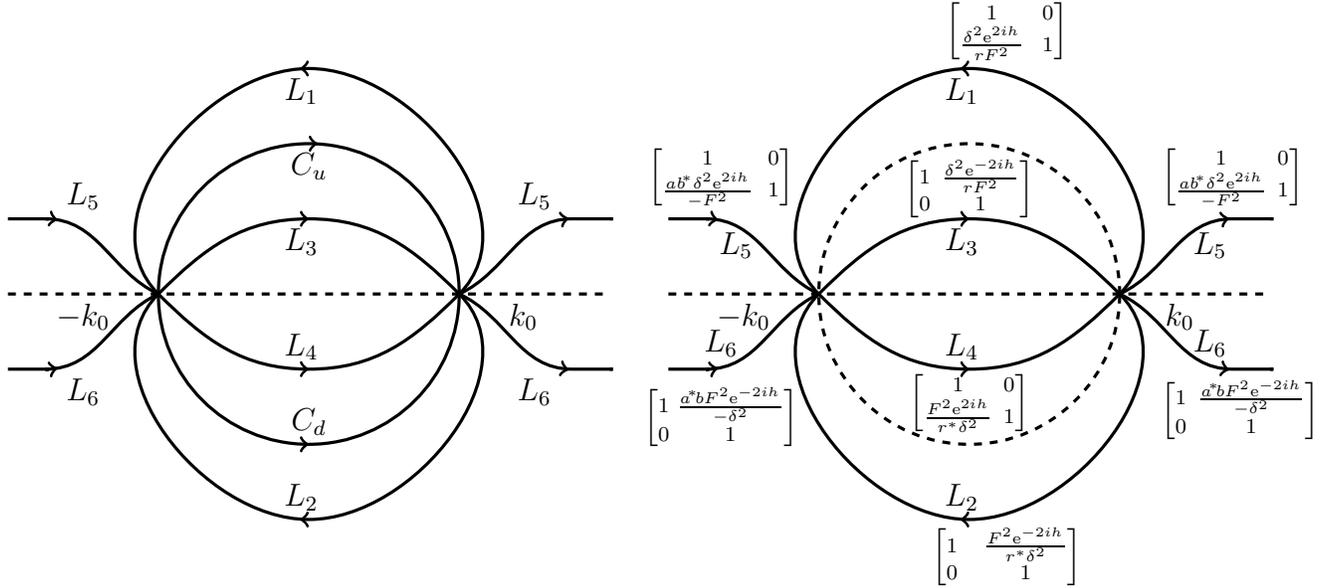
\begin{figure}
\begin{tikzpicture}
\draw[very thick, dashed] (-4,0) to (4,0);  
\draw[very thick, decoration={markings, mark=at position 0.25 with {\arrow{<}}}, decoration={markings, mark=at position 0.75 with {\arrow{>}}}, postaction={decorate}] (0,0) circle(2);

\draw[very thick, decoration={markings, mark=at position 0.5 with {\arrow{<}}}, postaction={decorate}] (-2,0) to [out=135, in=180] (0,3) [out =0, in=45] to (2,0);
\node at (-0.1,2.7){$L_1$};
\draw[very thick, decoration={markings, mark=at position 0.5 with {\arrow{<}}}, postaction={decorate}] (-2,0) to [out=-135, in=180] (0,-3) [out =0, in=-45] to (2,0);
\node at (-0.1,-2.7){$L_2$};
\draw[very thick, decoration={markings, mark=at position 0.5 with {\arrow{>}}}, postaction={decorate}] (-2,0) to [ in=180] (0,1) [out =0] to (2,0);
\node at (-0.1,0.7){$L_3$};
\draw[very thick, decoration={markings, mark=at position 0.5 with {\arrow{>}}}, postaction={decorate}] (-2,0) to [out=-45, in=-180] (0,-1) [out =0, in=-135] to (2,0);
\node at (-0.1, -0.7){$L_4$};

\draw[very thick, decoration={markings, mark=at position 0.3 with {\arrow{>}}}, postaction={decorate}] (-4,1) to [out=0, in=0](-3.5,1) [out=0, in =155] to  (-2,0); 
\node at (-3,1.3) {$L_5$};
\draw[very thick, decoration={markings, mark=at position 0.3 with {\arrow{<}}}, postaction={decorate}] (4,1) to [out=0, in=0](3.5,1) [out=180, in =25] to  (2,0); 
\node at (3,1.3) {$L_5$};

\draw[very thick, decoration={markings, mark=at position 0.3 with {\arrow{>}}}, postaction={decorate}] (-4,-1) to [out=0, in=0](-3.5,-1) [out=0, in =-155] to  (-2,0); 
\node at (-3, -1.3) {$L_6$};
\draw[very thick, decoration={markings, mark=at position 0.3 with {\arrow{<}}}, postaction={decorate}] (4, -1) to [out=0, in=0](3.5, -1) [out=180, in =-25] to  (2,0); 
\node at (3, -1.3) {$L_6$};

\node at (0, 1.7){$C_{u}$};
\node at (0, -1.7){$C_{d}$};
\node at (2.85, -0.3){$k_0$};
\node at (-3., -0.3){$-k_0$};

\node at (0.5, -3.5){\color{white}$\scriptsize\begin{bmatrix}1 & \frac{F^2 \e^{-2\ii h}}{r^*\delta^2} \\ 0 & 1\end{bmatrix}$};
\end{tikzpicture}
\begin{tikzpicture}
\draw[very thick, dashed] (-4,0) to (4,0);  
\draw[very thick, dashed] (0,0) circle(2);

\draw[very thick, decoration={markings, mark=at position 0.5 with {\arrow{<}}}, postaction={decorate}] (-2,0) to [out=135, in=180] (0,3) [out =0, in=45] to (2,0);
\node at (-0.1,2.7){$L_1$};
\draw[very thick, decoration={markings, mark=at position 0.5 with {\arrow{<}}}, postaction={decorate}] (-2,0) to [out=-135, in=180] (0,-3) [out =0, in=-45] to (2,0);
\node at (-0.1,-2.7){$L_2$};
\draw[very thick, decoration={markings, mark=at position 0.5 with {\arrow{>}}}, postaction={decorate}] (-2,0) to [ in=180] (0,1) [out =0] to (2,0);
\node at (-0.1,0.7){$L_3$};
\draw[very thick, decoration={markings, mark=at position 0.5 with {\arrow{>}}}, postaction={decorate}] (-2,0) to [out=-45, in=-180] (0,-1) [out =0, in=-135] to (2,0);
\node at (-0.1, -0.7){$L_4$};

\draw[very thick, decoration={markings, mark=at position 0.3 with {\arrow{>}}}, postaction={decorate}] (-4,1) to [out=0, in=0](-3.5,1) [out=0, in =155] to  (-2,0); 
\node at (-3.1, 0.65) {$L_5$};
\draw[very thick, decoration={markings, mark=at position 0.3 with {\arrow{<}}}, postaction={decorate}] (4,1) to [out=0, in=0](3.5,1) [out=180, in =25] to  (2,0); 
\node at (3.2, 0.65) {$L_5$};

\draw[very thick, decoration={markings, mark=at position 0.3 with {\arrow{>}}}, postaction={decorate}] (-4,-1) to [out=0, in=0](-3.5,-1) [out=0, in =-155] to  (-2,0); 
\node at (-3.3, -0.65){$L_6$};
\draw[very thick, decoration={markings, mark=at position 0.3 with {\arrow{<}}}, postaction={decorate}] (4, -1) to [out=0, in=0](3.5, -1) [out=180, in =-25] to  (2,0); 
\node at (3.2, -0.65) {$L_6$};

\node at (2.8, -0.3){$k_0$};
\node at (-3., -0.3){$-k_0$};

\node at (-3.3,1.55){$\scriptsize\begin{bmatrix}1\! & \!\! 0\\ \frac{ab\!^*\delta^2\e^{2\ii h}}{-F^2} \! & \!\! 1\end{bmatrix}$};
\node at (3.5, 1.55){$\scriptsize\begin{bmatrix}1 \!\! & \!\! 0\\\frac{ab\!^*\delta^2\e^{2\ii h}}{-F^2} \!\! & \!\! 1\end{bmatrix}$};
\node at (-3.3, -1.65){$\scriptsize\begin{bmatrix}1 \!\! & \!\!\! \frac{a\!^*\!bF^2\e^{-2\ii h}}{-\delta^2} \\ 0 \!\! & \!\!\! 1\end{bmatrix}$};
\node at (3.6, -1.55){$\scriptsize\begin{bmatrix}1 \!\! & \!\! \frac{a\!^*\!bF^2\e^{-2\ii h}}{-\delta^2} \\ 0 \!\! & \!\! 1\end{bmatrix}$};
%
\node at (0.5,3.5){$\scriptsize\begin{bmatrix}1 \! & \! 0\\\frac{\delta^2\e^{2\ii h}}{rF^2} \! & \! 1\end{bmatrix}$};
\node at (0.0,1.4){$\scriptsize\begin{bmatrix}1 \!\! & \!\! \frac{\delta^2\e^{-2\ii h}}{rF^2} \\ 0 \!\! & \!\! 1\end{bmatrix}$};
\node at (0.5, -3.5){$\scriptsize\begin{bmatrix}1 & \frac{F^2 \e^{-2\ii h}}{r^*\delta^2} \\ 0 & 1\end{bmatrix}$};
\node at (0.0, -1.45){$\scriptsize\begin{bmatrix}1 \!\! & \!\! 0 \\ \frac{F^2 \e^{2\ii h}}{r^*\delta^2} \!\! & \!\! 1\end{bmatrix}$};

\end{tikzpicture}
\caption{On the left: Contours $L_j,$ $j=1, \ldots ,6.$ On the right: Jump contour and jumps for $M^{(4)}.$}
\label{Fig:M4}
\end{figure}

Using this notation, define
\[
M^{(4)}(k;k_0,\tau) = 
\begin{cases}
M^{(3)}(k;k_0,\tau)
\begin{pmatrix}
1 & 0 \\\dfrac{-a(k)b^*(k)\,\delta^2(k;k_0) e^{2i g(k;k_0,\tau)}}{F^2(k;k_0)} & 1
\end{pmatrix}, \qquad k\in Int(L_5, \mathbb{R}),
\\
M^{(3)}(k;k_0,\tau)
\begin{pmatrix}
1 & \dfrac{-a^*(k)b(k) F^2(k;k_0) e^{-2i g(k;k_0,\tau)}}{\delta^2(k;k_0)} \\ 0 & 1
\end{pmatrix}^{-1},\qquad k\in Int(L_6, \mathbb{R}),
\end{cases}\]
\[
M^{(4)}(k;k_0,\tau) = \begin{cases}
M^{(3)}(k;k_0,\tau)\begin{pmatrix}1 & 0 \\ \dfrac{-\delta^2(k;k_0)e^{2ig(k;k_0,\tau)}}{r(k)F^2(k;k_0)} & 1
\end{pmatrix},\qquad k\in Int(L_1, C_{u}),
\\
M^{(3)}(k;k_0,\tau)
\begin{pmatrix}
1 & \dfrac{-F^2(k;k_0)e^{-2ig(k;k_0,\tau)}}{r^*(k)\delta^2(k;k_0)}
\\ 0 & 1
\end{pmatrix}^{{-1}},\qquad k\in Int(L_2, C_{d}),
\end{cases}
\]
\[
M^{(4)}(k;k_0,\tau) = \begin{cases}
M^{(3)}(k;k_0,\tau)\begin{pmatrix}0 & 1 \\ -1 & 0\end{pmatrix}, \qquad k\in Int(C_{u}, L_3),
\\
M^{(3)}(k;k_0,\tau)\begin{pmatrix}0 & 1 \\ -1 & 0\end{pmatrix}, \qquad k\in Int(C_{d}, L_4),
\end{cases}\]
\[
M^{(4)}(k;k_0,\tau) = \begin{cases}
M^{(3)}(k;k_0,\tau)
\begin{pmatrix}1 & 0 \\ \dfrac{-\delta^2(k;k_0)e^{2ig(k;k_0,\tau)}}{r(k) F^2(k;k_0)} & 1\end{pmatrix}
\begin{pmatrix}0 & 1 \\ -1 & 0\end{pmatrix},\qquad k\in Int(\mathbb{R}, L_3),
\\
M^{(3)}(k;k_0,\tau)
\begin{pmatrix}
1 & \dfrac{-F^2(k;k_0)e^{-2i g(k;k_0,\tau)}}{\delta^2(k;k_0) \, r^*(k)} \\ 0 & 1
\end{pmatrix}^{\!\!\!-1}\begin{pmatrix}0 & 1 \\ -1 & 0\end{pmatrix},\qquad k\in Int(\mathbb{R}, L_4),
\end{cases}
\]
The jump for $M^{(4)}$ is shown in Figure \ref{Fig:M4}, right.

\subsubsection*{Step 5: moving the contours $L_5, L_6$ to infinity.}\label{sect_step_5}
The next step is to get rid of contours $L_5, L_6$ by moving them to infinity. 
This is the step where we use the compactness of the input pulse, by employing \eqref{b_estimate} that
$b^*(k) e^{2ih(k; k_0, \tau)} = \ord(k^{-1}e^{(T-2\tau)\Im k}) = o(1), \ k\to\infty,\ \Im k\geq0,$ for $\tau\geq T/2.$
Here $T$ is the constant that determines the support of the input pulse, i.e. $\mathrm{supp}\,\mathcal{E}_1 \subset[0, T].$
Define 
\[
\begin{cases}M^{(5)}(k;k_0,\tau) = M^{(4)}(k;k_0,\tau)
\begin{pmatrix}
1 & 0 \\\dfrac{-a(k)b^*(k)\,\delta^2(k;k_0) e^{2i h(k;k_0,\tau)}}{F^2(k;k_0)} & 1
\end{pmatrix},\qquad k\in Int(L_5, L_1),
\\
M^{(5)}(k;k_0,\tau) = M^{(4)}(k;k_0,\tau)
\begin{pmatrix}
1 & \dfrac{-a^*(k)b(k) F^2(k;k_0) e^{-2i h(k;k_0,\tau)}}{\delta^2(k;k_0)} \\ 0 & 1
\end{pmatrix}^{\hskip-2mm-1},\quad k\in\Omega_6 = Int(L_6, \mathbb{R}),
\\
M^{(5)}(k;k_0,\tau) = M^{(4)}(k;k_0,\tau),\qquad \mbox{elsewhere.}
\end{cases}
\]
This transformation removes the jumps over $L_5, L_6$, and changes the ones over $L_1, L_2$ as follows:
\[
M_-^{(5)}(k;k_0,\tau) = M_+^{(5)}(k;k_0,\tau)
\begin{pmatrix}1 & 0 \\ \dfrac{\delta^2(k;k_0) e^{2ih(k;k_0,\tau)}}{r(k)(1+r(k)r^*(k))\,F^2(k;k_0)} & 1\end{pmatrix},\qquad k \in L_1,
\]
\[
M_-^{(5)}(k;k_0,\tau) = M_+^{(5)}(k;k_0,\tau)
\begin{pmatrix}1 & \dfrac{F^2(k;k_0) e^{-2ih(k;k_0,\tau)}}{r^*(k)(1+r(k)r^*(k))\,\delta^2(k;k_0)} \\ 0 & 1\end{pmatrix}, \qquad k\in L_2.
\]

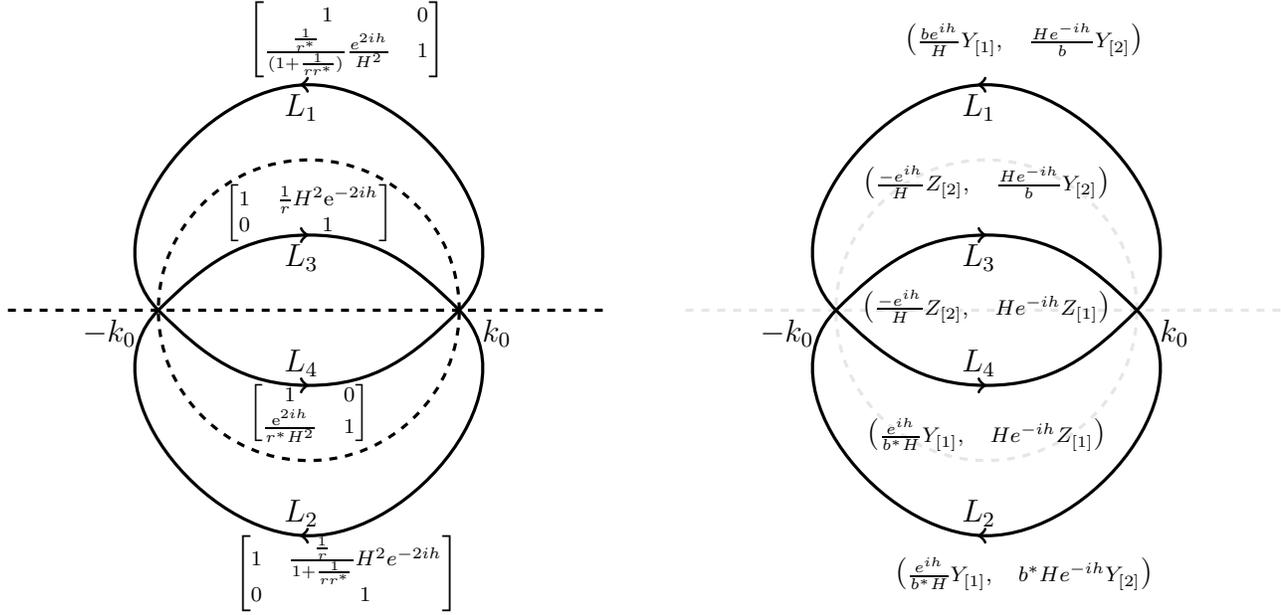
\begin{figure}
\begin{tikzpicture}
\draw[very thick, dashed] (-4,0) to (4,0);  
\draw[very thick, dashed] (0,0) circle(2);

\draw[very thick, decoration={markings, mark=at position 0.5 with {\arrow{<}}}, postaction={decorate}] (-2,0) to [out=135, in=180] (0,3) [out =0, in=45] to (2,0);
\node at (-0.1,2.7){$L_1$};
\draw[very thick, decoration={markings, mark=at position 0.5 with {\arrow{<}}}, postaction={decorate}] (-2,0) to [out=-135, in=180] (0,-3) [out =0, in=-45] to (2,0);
\node at (-0.1,-2.7){$L_2$};
\draw[very thick, decoration={markings, mark=at position 0.5 with {\arrow{>}}}, postaction={decorate}] (-2,0) to [ in=180] (0,1) [out =0] to (2,0);
\node at (-0.1,0.7){$L_3$};
\draw[very thick, decoration={markings, mark=at position 0.5 with {\arrow{>}}}, postaction={decorate}] (-2,0) to [out=-45, in=-180] (0,-1) [out =0, in=-135] to (2,0);
\node at (-0.1, -0.7){$L_4$};

\node at (2.5, -0.3){$k_0$};
\node at (-2.65, -0.3){$-k_0$};

%
\node at (0.5,3.6){$\scriptsize\begin{bmatrix}1 & 0\\\frac{\frac1{r^*}}{(1 + \frac{1}{rr^*})}\frac{e^{2\ii h}}{H^2} & 1\end{bmatrix}$};
\node at (0.0, 1.3){$\scriptsize\begin{bmatrix}1 & \frac1{r}H^2\e^{-2\ii h} \\ 0 & 1\end{bmatrix}$};
\node at (0.5, -3.5){$\scriptsize\begin{bmatrix}1 & \frac{\frac1{r}}{1 + \frac{1}{rr^*}}H^2e^{-2\ii h} \\ 0 & 1\end{bmatrix}$};
\node at (0.0, -1.37){$\scriptsize\begin{bmatrix}1 & 0 \\ \frac{\e^{2\ii h}}{r^*H^2}  & 1\end{bmatrix}$};
\end{tikzpicture}
\quad
\quad
\begin{tikzpicture}
\draw[very thick, dashed, black!10] (-4,0) to (4,0);  
\draw[very thick, dashed, black!10] (0,0) circle(2);

\draw[very thick, decoration={markings, mark=at position 0.5 with {\arrow{<}}}, postaction={decorate}] (-2,0) to [out=135, in=180] (0,3) [out =0, in=45] to (2,0);
\node at (-0.1,2.7){$L_1$};
\draw[very thick, decoration={markings, mark=at position 0.5 with {\arrow{<}}}, postaction={decorate}] (-2,0) to [out=-135, in=180] (0,-3) [out =0, in=-45] to (2,0);
\node at (-0.1,-2.7){$L_2$};
\draw[very thick, decoration={markings, mark=at position 0.5 with {\arrow{>}}}, postaction={decorate}] (-2,0) to [ in=180] (0,1) [out =0] to (2,0);
\node at (-0.1,0.7){$L_3$};
\draw[very thick, decoration={markings, mark=at position 0.5 with {\arrow{>}}}, postaction={decorate}] (-2,0) to [out=-45, in=-180] (0,-1) [out =0, in=-135] to (2,0);
\node at (-0.1, -0.7){$L_4$};

\node at (2.5, -0.3){$k_0$};
\node at (-2.65, -0.3){$-k_0$};

\node at (0.5, -3.5){\color{white}$\scriptsize\begin{bmatrix}1 & \frac{\frac1{r}}{1 + \frac{1}{rr^*}}H^2e^{-2\ii h} \\ 0 & 1\end{bmatrix}$};
%
\node at (0.5,3.6){$\scriptsize\begin{pmatrix}\frac{be^{ih}}{H}Y_{[1]}, & \frac{H e^{-i h}}{b}Y_{[2]}\end{pmatrix}$};
\node at (0.0, 1.7){$\scriptsize\begin{pmatrix}\frac{-e^{ih}}{H}Z_{[2]}, & \frac{H e^{-i h}}{b}Y_{[2]}\end{pmatrix}$};
\node at (0.5, -3.5){$\scriptsize\begin{pmatrix}\frac{e^{i h}}{b^* H}Y_{[1]}, & b^*He^{-ih}Y_{[2]}\end{pmatrix}$};
\node at (0.0, -1.65){$\scriptsize\begin{pmatrix}\frac{e^{ih}}{b^*H}Y_{[1]}, & He^{-ih}Z_{[1]}\end{pmatrix}$};
\node at (0.0, 0.05){$\scriptsize\begin{pmatrix}\frac{-e^{ih}}{H}Z_{[2]}, & He^{-ih}Z_{[1]}\end{pmatrix}$};
\end{tikzpicture}
\caption{ On the left: jumps for  $M^{(5)}.$ On the right: structure of $M^{(5)}$ in terms of $Y, Z.$}\label{Fig:M5}
\end{figure}

The jumps for $M^{(5)}$ are naturally expressed in terms of the function $H$ defined in \eqref{def:H} (see also Figure~\ref{Fig:M5}, left).
We have 
$M^{(5)}_-(k;k_0,\tau)=M^{(5)}_+(k;k_0,\tau)J^{(5)}(k;k_0,\tau),$ where

\begin{align*}
J^{(5)}(k) & =\begin{pmatrix}1 & 0 \\ \dfrac{\frac1{r^*(k)}}{1+\frac{1}{r(k)r^*(k)}}\cdot\dfrac{e^{2ih(k)}}{H^2(k)}
 & 1\end{pmatrix},\quad k\in L_1,
\qquad\qquad\,
=\begin{pmatrix}1 & \dfrac{1}{r(k)}H^2(k)e^{-2ih(k)} \\ 0 & 1\end{pmatrix},\quad k\in L_3,
\\
&
=\begin{pmatrix}1 & \dfrac{\frac1{r(k)}}{1+\frac{1}{r(k)r^*(k)}}\cdot H^2(k)e^{-2ih(k)} \\ 0 & 1\end{pmatrix},\quad k\in L_2,
\qquad
=\begin{pmatrix}1 & 0 \\ \dfrac{e^{2ih(k)}}{r^*(k)H^2(k)} & 1\end{pmatrix},\quad k\in L_4,
\end{align*}
and where
$J^{(5)}(k_0) = J^{(5)}(k_0;k_0,\tau),\ $ $h(k) = h(k;k_0,\tau),\ $ $H(k)=H(k;k_0).$

\subsection*{Expressing $M^{(5)}$ in terms of the original Jost solutions, and figuring out the pole conditions.}
Tracing back the chain of transformations that led from $M$ to $M^{(5)}$, we find that $M^{(5)}$ has the following structure (see also Figure \ref{Fig:M5}, right):
\begin{align*}
M^{(5)}(k;k_0,\tau) = \begin{cases}
\begin{pmatrix}
\frac{b(k)}{H(k)}e^{i h(k)}Y_{[1]}(t, x; k), & \frac{H(k)}{b(k)}e^{-i h(k)}Y_{[2]}(t, x; k)
\end{pmatrix},\quad k\in Int(L_1, \mathbb{R}),
\\
\begin{pmatrix}
\frac{e^{i h(k)}Y_{[1]}(t, x; k)}{b^*(k)\,H(k)}, & b^*(k) H(k) e^{-i h(k)}Y_{[2]}(t, x; k)
\end{pmatrix},\quad k\in Int(L_2, \mathbb{R}),
\\
\begin{pmatrix}
\frac{-Z_{[2]}(k) e^{i h(k)}}{H(k)}, & \frac{H(k)}{b(k)}e^{-i h(k)}Y_{[2]}(t, x; k)
\end{pmatrix},\quad k\in Int(L_1, L_3),
\\
\begin{pmatrix}
\frac{e^{i h(k)}\,Y_{[1]}(k)}{b^*(k)\,H(k)}, & H(k)e^{-i h(k)}Z_{[1]}(t, x; k)
\end{pmatrix},\quad k\in Int(L_2, L_4),
\\
\begin{pmatrix}
\frac{-Z_{[2]}(k) e^{i h(k)}}{H(k)}, & H(k) e^{-i h(k)}Z_{[1]}(t, x; k)
\end{pmatrix},\quad k\in Int(L_3, L_4).
\end{cases}
\end{align*}

Recalling the definitions of the function $H$ \eqref{def:H}, we see that the function $a(.)$ does not contribute at all to the poles of $M^{(5)};$ the only poles are caused by zeros of the function $b(.).$ Furthermore,
$M^{(5)}$ has no poles in the domains $Int(L_1, \mathbb{R})$ and $Int(L_2, \mathbb{R})$, and the only poles are concentrated in the domains $Int(L_1, L_3), $ $Int(L_2, L_4)$ and $Int(L_3, L_4).$
The next transformation eliminates the poles in the domain $Int(L_3, L_4).$

\subsubsection*{Step 6: eliminating poles in $Int(L_3, L_4).$}
First, we ``modify'' the functions $H, G$ \eqref{def:G}, \eqref{def:H}, by eliminating or introducing zeros and poles of $H$ inside $Int(L_3, L_4),$ as follows
\begin{equation}\label{def:Htilde}
\widetilde{H}(k;k_0) = H(k;k_0)\prod\limits_{k_j \in Z_b\cap Int(L_3, (-k_0, k_0))}\frac{k - \ol{k_j}}{k - k_j}\,,
\qquad
\widetilde{G}(k;k_0) = G(k;k_0)\prod\limits_{k_j \in Z_b\cap Int(L_3, (-k_0, k_0))}\frac{k - \ol{k_j}}{k - k_j}\,.
\end{equation}
In doing so we assume that there are no zeros of $b(.)$ on the line $L_3,$ which can always be achieved by slightly deforming $L_3.$
Then the function
\[
M^{(6)}(k;k_0,\tau) = M^{(5)}(k;k_0,\tau)\prod\limits_{k_j \in Z_b\cap Int(L_3, (-k_0, k_0))}\(\frac{k - k_j}{k - \ol{k_j}}\)^{\sigma_3}
\]
satisfies the jump conditions, which are obtained from the jumps for the function $M^{(5)}$ if we substitute there $H$ with $\widetilde H,$ i.e.
$M^{(6)}_-(k;k_0,\tau)=M^{(6)}_+(k;k_0,\tau)J^{(6)}(k;k_0,\tau),$ where
\begin{align*}
J^{(6)}(k) & =\begin{pmatrix}1 & 0 \\ \dfrac{\frac1{r^*(k)}}{1+\frac{1}{r(k)r^*(k)}}\cdot\dfrac{e^{2ih(k)}}{\widetilde H^2(k)}
 & 1\end{pmatrix},\quad k\in L_1,
\quad\qquad\,
=\begin{pmatrix}1 & \dfrac{1}{r(k)}\widetilde H^2(k)e^{-2ih(k)} \\ 0 & 1\end{pmatrix},\quad k\in L_3,
\\
&
=\begin{pmatrix}1 & \dfrac{\frac1{r(k)}}{1+\frac{1}{r(k)r^*(k)}}\cdot \widetilde H^2(k)e^{-2ih(k)} \\ 0 & 1\end{pmatrix},\quad k\in L_2,
\qquad
=\begin{pmatrix}1 & 0 \\ \dfrac{e^{2ih(k)}}{r^*(k)\widetilde H^2(k)} & 1\end{pmatrix},\quad k\in L_4,
\end{align*}
and where $J^{(6)}(k) = J^{(6)}(k;k_0,\tau),\ $ $\widetilde H(k) = \widetilde H(k;k_0),\ $ $h(k)=h(k;k_0,\tau).$
Moreover, using properties \eqref{ZYrelations}, we can specify the pole conditions of the function $M^{(6)}$ in the domains $Int(L_1, L_3)$ and $Int(L_2, L_4)$ (recall that we assume (Assumption \ref{assumption_2}) that zeros of $b(.)$ are of the first order):
\begin{align*}
M^{(6)}(k;k_0,\tau)\begin{pmatrix}1 & 0 \\ 
\frac{\gamma_j\, e^{2ih(k;k_0,\tau)}}{(k - k_j)\widetilde G^2(k_j;k_0)} & 1\end{pmatrix}\mbox{is regular at the point $k=k_j$},
\\
M^{(6)}(k;k_0,\tau)\begin{pmatrix}1 & \frac{-\ol{\gamma_j}\,\widetilde G^2(\ol{k_j};k_0)\, e^{-2ih(k;k_0,\tau)}}{(k - \ol{k_j})} \\ 0 & 1\end{pmatrix}\mbox{is regular at the point $k=\ol{k_j}$},
\end{align*}
where 
$k_j \in Z_b\cap Int(L_1, L_3),$ 
\[
\gamma_j^{-1} = a(k_j)\dot{b}(k_j),
\]
and the dot $\dot{}$ denotes the derivative in $k.$

\subsection{Parametrix analysis}
We see that the jump matrices for $M^{(6)}$ are exponentially close to the identity matrix everywhere except for the neighbourhoods of the points $\pm k_0.$ We will construct functions that satisfy approximately the jump conditions near these points.

\subsubsection{Global parametrix.}\label{sect_global_parametrix}
Recall that by our Assumption \ref{assumption_2} all the absolute values $|k_j|$ for $k_j\in Z_b$ are mutually different, hence we can always achieve, by slightly moving contours $L_1, L_3,$ if necessary, that there is at most one zero of the function $b(.)$ in the domain $Int(L_1, L_3.)$

We define the global parametrix $M_{glob}(k)=M_{glob}(k;k_0,\tau)$ to be equal to the identity matrix, $M_{glob}(k) = I,$ in the case if $Z_b\cap Int(L_1, L_3) = \emptyset.$

In the case $Z_b\cap Int(L_1, L_3) = \left\{k_j\right\},$ we define the global parametrix to be equal to 
\begin{equation}\label{M_glob}
M_{glob}(k) = M_{glob}(k;k_0,\tau) = \begin{pmatrix}
1 + \frac{i\,A}{k - k_j} & \frac{i\, B}{k - \ol{k_j}}
\\
\frac{i\, \ol{B}}{k - k_j} & 1 - \frac{i\,\ol{A}}{k - \ol{k_j}}
\end{pmatrix}
\end{equation}
where $A, B$ and their complex conjugates $\ol{A}, \ol{B}$ are some complex parameters (which might depend on $k_0, \tau,$ or, equivalently, on $t,x$), which are determined from the condition
\[
M_{glob}(k)\begin{pmatrix}
1 & 0 \\ \frac{i\,\gamma_j\,e^{2ih(k;k_0,\tau)}}{\widetilde G^2(k;k_0)\,(k-k_j)} & 1
\end{pmatrix}
\quad \mbox{ is regular at } k\to k_j.
\]
This leads to the linear system of equations for $A, B,$
\begin{equation}\label{wj}
\begin{cases}
A + \w_j B = 0,
\\
\ol{\w_j} A - B = 2\Im k_j\cdot \ol{\w_j},
\end{cases}
\qquad 
\mbox { where }
\quad 
\w_j = \frac{\gamma_j\cdot e^{2ih(k_j;k_0,\tau)}}{2\Im k_j\cdot \widetilde G^2(k_j;k_0)},
\end{equation}
from where 
\begin{equation}\label{AjBj}
A = A_j = \frac{|\w_j|^2\cdot 2\Im k_j}{1 + |\w_j|^2},
\qquad 
B = B_j = \frac{-\ol{\w_j}\cdot 2\Im k_j}{1 + |\w_j|^2}.
\end{equation}
Note that $A_j, B_j$ are uniformly bounded in $t,x$, that $A_j>0$, that $A_j^2 + |B_j|^2 = 2\Im k_j\cdot A_j$ (and hence automatically $\det M_{glob}(k)\equiv 1$), and that $M_{glob}(k)$ satisfies the following pole condition at the complex conjugate point $\ol{k_j}:$
\[
M_{glob}(k)\begin{pmatrix}1 & \frac{i\,\ol{\zeta_j}\,\widetilde G^2(k)\,e^{-2ih(k)}}{k - \ol{k_j}} \\ 0 & 1\end{pmatrix}
\quad \mbox{ is regular at }\ \ol{k_j}.\]

\subsubsection{Local parametrices.}

\subsubsection*{Local behaviour of the function $\widetilde H$ near $k = -k_0$ and $k = k_0.$}
We have that 
\begin{align*}
\widetilde H(k;k_0) = \(\frac{k + k_0}{k_0}\)^{i\nu_l} \cdot \chi_l(k; k_0),
\qquad
\widetilde H(k;k_0) = \(\frac{k_0 - k}{k_0}\)^{-i\nu_r} \cdot \chi_r(k; k_0),
\end{align*}
where
\[
\nu_l = \nu_l(k_0) = \frac{1}{2\pi}\ln\(1 + \frac{1}{|r(-k_0)|^2}\),
\qquad
\nu_r =\nu_r(k_0) = \frac{1}{2\pi}\ln\(1 + \frac{1}{|r(k_0)|^2}\),
\]
and
\begin{align*}
\chi_l(k; k_0) = 
\exp\left[
\int\limits_{-k_0}^{k_0}
\frac{\ln\(\frac{1 + |r(s)|^{-2}}{1 + |r(-k_0)|^{-2}}\)\, ds}{(s - k)\cdot 2\pi i}
\right]
\(\frac{k_0-k}{k_0}\)^{-i\nu_l}
\cdot
\prod\limits_{k_j\in Int(L_3, \mathbb{R})}\frac{k - \ol{k_j}}{k - k_j}
\cdot
\begin{cases}
b(k)e^{\pi\nu_l},\ \Im k>0,
\\
\frac{1}{b^*(k)}e^{-\pi\nu_l},\ \Im k<0,
\end{cases}
\end{align*}
and
\begin{align*}
\chi_r(k; k_0) = \exp\left[\int\limits_{-k_0}^{k_0}
\frac{\ln\(\frac{1 + |r(s)|^{-2}}{1 + |r(k_0)|^{-2}}\)\, ds}{(s-k) \cdot 2\pi i}
\right]
\(\frac{k+k_0}{k_0}\)^{i\nu_r}
\cdot\prod\limits_{k_j\in Int(L_3, \mathbb{R})}\frac{k - \ol{k_j}}{k - k_j}
\cdot
\begin{cases}
b(k)e^{\pi\nu_r},\ \Im k>0,
\\
\frac{1}{b^*(k)}e^{-\pi\nu_r},\ \Im k<0.
\end{cases}
\end{align*}
Note that the functions $\chi_l(k; k_0), \chi_r(k; k_0)$ have (nonzero) limits as $k\to\mp k_0,$ respectively, namely
\begin{equation}\label{chi_l}
\chi_l(-k_0; k_0)
=
\exp\left[
\int\limits_{-k_0}^{k_0}
\frac{\ln\(\frac{1 + |r(s)|^{-2}}{1 + |r(-k_0)|^2}\)\, ds}{(s+k_0)\cdot 2\pi i}
\right]
\cdot 2^{-i\nu_l(k_0)}
e^{i\arg b(-k_0)}
\cdot\prod\limits_{k_j\in Int(L_3, \mathbb{R})}\frac{k_0 + \ol{k_j}}{k_0 + k_j}
\end{equation}
and
\begin{equation}\label{chi_r}
\chi_r(k_0; k_0)
=
\exp\left[\int\limits_{-k_0}^{k_0}
\frac{\ln\(\frac{1 + |r(s)|^{-2}}{1 + |r(k_0)|^{-2}}\)\, ds}{(s-k_0)\cdot 2\pi i}
\right]
\cdot 2^{i\nu_r(k_0)}
e^{i\arg b(k_0)}
\cdot\prod\limits_{k_j\in Int(L_3, \mathbb{R})}\frac{k_0 - \ol{k_j}}{k_0 - k_j},
\end{equation}
and $|\chi_l(-k_0; k_0)| = |\chi_r(k_0; k_0)| = 1.$

\subsubsection*{Local behaviour of the function $h$ near $k = -k_0$ and near $k = k_0.$}
Define conformal mappings $z_l = z_l(k; k_0), z_r=z_r(k; k_0)$ as follows:
\[
z_l = \frac{k+k_0}{k_0}\cdot\(\sum\limits_{j=0}^{\infty}\(\frac{k+k_0}{k_0}\)^j\)^{1/2},
\qquad
z_r = \frac{k_0-k}{k_0}\cdot\(\sum\limits_{j=0}^{\infty}\(\frac{k_0-k}{k_0}\)^j\)^{1/2}
\]
and define the scaled variables $\lambda_l = \lambda_l(k;k_0,\tau),$ $\lambda_r = \lambda_r(k;k_0,\tau)$ as follows:
\[\lambda_l = z_l \sqrt{\tau k_0},
\qquad
\lambda_r = z_r \sqrt{\tau k_0},\]
so that
\[
h(k;k_0,\tau) = \tau \(k + \frac{k_0^2}{k} \) = -2\tau k_0 - \tau k_0\, z_l^2 = -2\tau k_0 - \lambda_l^2,
\]\[
h(k;k_0,\tau) = 2\tau k_0 + \tau k_0\, z_r^2 = 2\tau k_0 + \lambda_r^2.
\]

\subsubsection*{Parabolic cylinder parametrix.}

For $a\in\mathbb{C},$ the parabolic cylinder function $D_{a}(z)$ is an entire function that satisfies the differential equation
\[\partial_{zz}D_{a}(z)+\(a+\frac12-\frac{z^2}{4}\)D_{a}(z)=0,\]
and has the asymptotics as $z\to\infty$
$$
D_{a}(z) = z^{a}\e^{-z^2/4}\(1-\frac{a(a-1)}{2z^2}+
\frac{a(a-1)(a-2)(a-3)}{8z^4}+\mathcal{O}(z^{-6})\),\quad \arg z\in\(\frac{-3\pi}{4},\frac{3\pi}{4}\).
$$
Furthermore, $D_a(z)$ satisfies the following relations:
\begin{equation}\label{ParaCyl_prop}\begin{split}&
D_{a}(z)=\ee^{-\pi a\ii}D_{a} (-z)+\frac{\sqrt{2\pi}}{\Gamma(-a)}\e^{-\pi(a+1)\ii/2}D_{-a-1}(\ii z),
\\
&
D_{a}(z) = \e^{\pi a\ii}D_{a} (-z)+\frac{\sqrt{2\pi}}{\Gamma(-a)}\e^{\pi(a+1)\ii/2}D_{-a-1}(-\ii z),
\\
&
D_{a}(z)=\frac{\Gamma(a+1)}{\sqrt{2\pi}}\(\e^{\pi\ii a/2}D_{-a-1}(\ii z)+\e^{-\pi\ii a/2}D_{-a-1}(-\ii z)\),
\\&
D_{a+1}(z)-zD_{a}(z)+aD_{a-1}(z)=0,
\qquad 
D_{a}'(z)=-\frac{z}{2}D_{a}(z)+aD_{a-1}(z)=0.
\end{split}\end{equation}

\noindent
Let $r_0\in\mathbb{C}$ be a non-zero complex parameter. Define 
$$\nu = \frac{1}{2\pi}\ln\(1 + \frac1{|r_0|^2}\)>0$$
and consider the following piece-wise analytic function:
$$
\Psi(\lambda; r_0)=\begin{bmatrix}
v^+D_{\ii\nu}(2\e^{-\pi\ii/4}\,\lambda)
 & 
-\beta_2u^+D_{-\ii\nu-1}(2\e^{{-3\pi\ii}/{4}}\,\lambda)
\\
-\beta_1v^+D_{\ii\nu-1}(2\e^{{-\pi\ii}/{4}}\,\lambda)  
&
u^+D_{-\ii\nu}(2\e^{{-3\pi\ii}/{4}}\,\lambda)
\end{bmatrix},\quad \Im\lambda>0,
$$
$$
\Psi(\lambda; r_0)=\begin{bmatrix}
v^-D_{\ii\nu}(2\e^{3\pi\ii/4}\,\lambda) 
& 
\beta_2u^-D_{-\ii\nu-1}(2\e^{{\pi\ii}/{4}}\,\lambda)
\\
\beta_1 v^-D_{\ii\nu-1}(2\e^{{3\pi\ii}/{4}}\,\lambda)
&
u^-D_{-\ii\nu}(2\e^{{\pi\ii}/{4}}\,\lambda)
\end{bmatrix},\quad \Im\lambda<0,
$$
where 
\begin{align*}&u^+=2^{\ii\nu}\e^{3\pi\nu/4},\quad
u^-=2^{\ii\nu}\e^{-\pi\nu/4},\quad
v^+=2^{-\ii\nu}\e^{-\pi\nu/4},\quad
v^-=2^{-\ii\nu}\e^{3\pi\nu/4},
\\
&\beta_1 = \dfrac{-\ii\sqrt{2\pi}\,r_0\,2^{2\ii\nu}\e^{\pi\nu/2}}{\Gamma(\ii\nu)} = \dfrac{\Gamma(-\ii\nu+1)2^{2\ii\nu}}{\sqrt{2\pi}\cdot\ol{r_0}\cdot\e^{\pi\nu/2}},
\quad
\beta_2 = \dfrac{\Gamma(1+\ii\nu)}{\sqrt{2\pi}\,r_0\,2^{2\ii\nu}\e^{\pi\nu/2}}=\dfrac{\ii\sqrt{2\pi}\cdot\ol{r_0}\cdot\e^{\pi\nu/2}}{2^{2\ii\nu}\Gamma(-\ii\nu)}\,.
\end{align*}
Note that $\beta_1\beta_2 = \nu,\ \ol{\beta_1} = \beta_2.$

Properties \eqref{ParaCyl_prop} allow to verify that $\Psi(\lambda; r_0)$ has the following jump across the real line:
\[
\Psi(\lambda-\ii 0; r_0) = \Psi(\lambda+\ii 0; r_0)\begin{pmatrix}1+\frac1{|r_0|^2} & \frac1{r_0} \\ \frac1{\ol{r_0}} & 1\end{pmatrix},\quad \lambda\in\mathbb{R}.
\]
The function
\begin{equation}
\begin{split}
	P(\lambda; r_0)&=
	\Psi(\lambda; r_0)\cdot\lambda^{-\ii\nu\sigma_3}\e^{-\ii\lambda^2\sigma_3}, \qquad  \qquad 
	\arg\lambda\in\(\dfrac{\pi}{4},\frac{3\pi}{4}\)\cup
	\(\frac{-3\pi}{4},\dfrac{-\pi}{4}\),
	\\&=
	\Psi(\lambda; r_0)\cdot\lambda^{-\ii\nu\sigma_3}\e^{-\ii\lambda^2\sigma_3}
	\begin{pmatrix}1 & \frac1{r_0} \lambda^{2\ii\nu}\e^{2\ii\lambda^2} \\ 0 & 1\end{pmatrix}, \qquad \arg\lambda\in\(0,\dfrac{\pi}{4}\),
	\\
	&=
	\Psi(\lambda; r_0)\cdot\lambda^{-\ii\nu\sigma_3}\e^{-\ii\lambda^2\sigma_3}
	\begin{pmatrix}1 &  0 \\ \frac{-1}{\ol{r_0}}\lambda^{-2\ii\nu}\e^{-2\ii\lambda^2} & 1\end{pmatrix}, \qquad \arg\lambda
	\in\(\dfrac{-\pi}{4}, 0\),
	\\
	&=
	\Psi(\lambda; r_0)\cdot\lambda^{-\ii\nu\sigma_3}\e^{-\ii\lambda^2\sigma_3}
	\begin{pmatrix}1 & 0 \\ \frac{\frac{1}{\ol{r_0}}}{1+\frac{1}{|r_0|^2}}\lambda^{-2\ii\nu}\e^{-2\ii\lambda^2} & 1\end{pmatrix}, \qquad \arg\lambda
	\in\(\frac{3\pi}{4}, \pi\),
	\\
	&=
	\Psi(\lambda; r_0)\cdot\lambda^{-\ii\nu\sigma_3}\e^{-\ii\lambda^2\sigma_3}
	\begin{pmatrix}1 & \frac{\frac{-1}{r_0}}{1+\frac1{|r_0|^2}}\lambda^{2\ii\nu}\e^{2\ii\lambda^2} \\  0 & 1\end{pmatrix}, \qquad \arg\lambda
	\in\(-\pi, \frac{-3\pi}{4}\)
\end{split}
\end{equation}
satisfies the jump relation  $P_{-}(\lambda; r_0) = P_{+}(\lambda; r_0)
J(\lambda; r_0)$ on the contour $$
\lambda\in \Sigma=(\infty\e^{3\pi\ii/4},0)\cup (\infty\e^{-3\pi\ii/4},0) \cup (0,\infty\e^{\pi\ii/4})\cup(0,\infty\e^{-\pi\ii/4}),
$$
where 
$J(\lambda; r_0)$ is equal to 
\begin{align*}
J(\lambda; r_0)&=\begin{bmatrix}1 & 0 \\ \frac{\frac1{\ol{r_0}}}{1+\frac{1}{|r_0|^2}}\lambda^{-2\ii\nu}\e^{-2\ii\lambda^2} & 1\end{bmatrix},
	\lambda\in (\infty\e^{3\pi\ii/4},0),
	\ \ \,
	=\begin{bmatrix}1 & \frac1{r_0}\lambda^{2\ii\nu}\e^{2\ii\lambda^2} \\ 0 & 1\end{bmatrix},
	\lambda\in (0,\infty\e^{\pi\ii/4}),
	\\
	&=\begin{bmatrix}1 & \frac{\frac{1}{r_0}}{1+\frac{1}{|r_0|^2}}\lambda^{2\ii\nu}\e^{2\ii\lambda^2} \\ 0 & 1 \end{bmatrix}, \lambda\in (\infty\e^{-3\pi\ii/4},0),
	\ \quad
	=\begin{bmatrix}1 & 0 \\ \frac{1}{\ol{r_0}}\lambda^{-2\ii\nu}\e^{-2\ii\lambda^2} & 1\end{bmatrix},
	\lambda\in (0,\infty\e^{-\pi\ii/4}),
\end{align*}
and has the following asymptotics as $\lambda\to\infty,$ which is uniform in $\arg\lambda\in[-\pi, \pi]:$
\begin{equation}\label{PC_asymp}
	P(\lambda; r_0) = 
	\begin{bmatrix}
		1 - \frac{\nu(1-\i\nu)}{8\lambda^2} + \frac{\ii\nu(1-\ii\nu)(2-\ii\nu)(3-\ii\nu)}{128\lambda^4}+\mathcal{O}(\lambda^{-6})
		& 
		\frac{\e^{-\pi\ii/4}\,\beta_2}{2\lambda}
		+\frac{\e^{\pi\ii/4}\,\beta_2(1+\ii\nu)(2+\ii\nu)}{16\lambda^3}+\mathcal{O}(\lambda^{-6})
		\\\\
		\frac{\e^{-3\pi\ii/4}\,\beta_1}{2\lambda}+\frac{\e^{3\pi\ii/4}\,\beta_1(1-\ii\nu)(2-\ii\nu)}{16\lambda^3}+\mathcal{O}(\lambda^{-6})
		&
		1 - \frac{\nu(1+\ii\nu)}{8\lambda^2} - \frac{\ii\nu(1+\ii\nu)(2+\ii\nu)(3+\ii\nu)}{128\lambda^4}+\mathcal{O}(\lambda^{-6})
	\end{bmatrix}.\end{equation}
Note also that the functions $\Psi(\lambda; r_0)$ and $P(\lambda; r_0)\lambda^{-i\nu\sigma_3}$ are continuous at the origin.

\subsubsection*{Approximate solution $M_{appr}$.}

Now we define a function $M_{appr},$ which satisfies approximately the jump conditions near the points $-k_0, k_0.$
Let $\delta\in(0, 1)$ be a fixed number and denote by $C_l, C_r$ the circles of radius $\delta\cdot k_0$ centred at $-k_0, k_0, $respectively. 
Denote the corresponding disks by
$$Int(C_l) = \left\{k:\ |k+k_0| < k_0\cdot \delta\right\},
\qquad
Int(C_r) := \left\{k:\ |k - k_0| < k_0\cdot \delta\right\}.$$
Define
$$
M_{appr}(k;k_0,\tau)
=
\begin{cases}
M_{glob}(k;k_0,\tau),\qquad |k\mp k_0|>\delta\cdot |k_0|,
\\
Q_l(k;k_0,\tau)P\(\lambda_l; r_0 = r(-k_0)\)\cdot \phi_l(k;k_0,\tau)^{-\sigma_3},\qquad k\in Int(C_l),
\\
Q_r(k;k_0,\tau)
{\scriptsize \begin{pmatrix}0 & 1 \\ -1 & 0 \end{pmatrix} }
P\(\lambda_r; r_0 = \ol{r(k_0)}\)
{\scriptsize \begin{pmatrix}0 & -1 \\ 1 & 0 \end{pmatrix}}
\cdot
\phi_r(k;k_0,\tau)^{-\sigma_3},\quad k\in Int(C_r),
\end{cases}
$$
where
\begin{equation}\label{phi_lr}
\begin{split}
&\phi_l(k;k_0,\tau) = 
\(\frac{k+k_0}{k_0 \cdot z_l\cdot \sqrt{\tau k_0}}\)^{\ii\nu_l(k_0)}
\cdot
\chi_l(-k_0; k_0)
\cdot
e^{2\ii\tau k_0},
\\
&\phi_r(k;k_0,\tau) = 
\(\frac{k_0 \cdot z_r \cdot \sqrt{\tau k_0}}{k_0 - k}\)^{\ii\nu_r(k_0)}
\cdot
\chi_r(k_0; k_0)
\cdot
e^{-2\ii\tau k_0}.
\end{split}
\end{equation}

\noindent Here $Q_l, Q_r$ are $2\times2$ matrix-valued functions analytic in $k\in Int(C_l), Int(C_r),$ respectively, which are needed to make 
the jump on the circles $C_l, C_r$ for the error matrix
\[
M_{err}(k;k_0,\tau) = M^{(6)}(k;k_0,\tau)M_{appr}(k;k_0,\tau)^{-1}
\]
as close to the identity matrix as possible.
We thus define 
\[
Q_l(k;k_0,\tau) := M_{glob}(k;k_0,\tau)\phi_l(k;k_0,\tau)^{\sigma_3},
\qquad
Q_r(k;k_0,\tau) := M_{glob}(k;k_0,\tau)\phi_r(k;k_0,\tau)^{\sigma_3}.
\]
Note that $M_{appr}$ has the following jumps: 
$M_{appr, -}(k;k_0,\tau)
=
M_{appr, +}(k;k_0,\tau)
J_{appr}(k;k_0,\tau),$ where 
the jump matrix $J_{appr}$ is as follows: inside $Int(C_l)$ it has the form
\begin{align*}
J_{appr}(k; k_0, \tau)&=\begin{bmatrix}1 & 0 \\ \frac{\frac1{\ol{r(-k_0)}}}{1+\frac{1}{|r(-k_0)|^2}}
\cdot
\frac{e^{2ih(k;k_0,\tau)}}{\chi^2_l(-k_0; k_0)}
\cdot
\(\frac{k+k_0}{k_0}\)^{-2i\nu_l(k_0)} & 1\end{bmatrix},
	k\in L_1\cap Int(C_l),
	\\
	&=\begin{bmatrix}1 & \frac1{r(-k_0)}\chi_l^2(-k_0; k_0)e^{-2ih(k;k_0,\tau)}\(\frac{k+k_0}{k_0}\)^{2i\nu_l(k_0)} \\ 0 & 1\end{bmatrix},
	\lambda\in L_3\cap Int(C_l),
	\\
	&=\begin{bmatrix}1 & \frac{\frac{1}{r(-k_0)}}{1+\frac{1}{|r(-k_0)|^2}}\cdot\chi_l^2(-k_0; k_0)e^{-2ih(k;k_0,\tau)}\(\frac{k+k_0}{k_0}\)^{2i\nu_l(k_0)} \\ 0 & 1 \end{bmatrix}, \lambda\in L_2\cap Int(C_l),
	\\
	&=\begin{bmatrix}1 & 0 \\ \frac{1}{\ol{r(-k_0)}}\cdot\frac{e^{2ih(k;k_0,\tau)}}{\chi_l^2(-k_0; k_0)}\(\frac{k+k_0}{k_0}\)^{-2i\nu_l(k_0)} & 1\end{bmatrix},
	\lambda\in L_4\cap Int(C_l),
\end{align*}
and inside $Int(C_r)$ it has the following form:
\begin{align*}
J_{appr}(k;k_0,\tau)&=
	\begin{bmatrix}1 & \frac1{r(k_0)}\chi_r^2(k_0; k_0)e^{-2ih(k)}\(\frac{k_0-k}{k_0}\)^{-2i\nu_r(k_0)} \\ 0 & 1\end{bmatrix},
	\lambda\in L_3\cap Int(C_r),
	\\&=
	\begin{bmatrix}1 & 0 \\ \frac{\frac1{\ol{r(k_0)}}}{1+\frac{1}{|r(k_0)|^2}}
\cdot
\frac{e^{2ih(k)}}{\chi^2_r(k_0; k_0)}
\cdot
\(\frac{k_0 - k}{k_0}\)^{2i\nu_r(k_0)} & 1\end{bmatrix},
	k\in L_1\cap Int(C_r),
	\\
	&=\begin{bmatrix}1 & 0 \\ \frac{1}{\ol{r(k_0)}}\cdot\frac{e^{2ih(k)}}{\chi_r^2(k_0; k_0)}\(\frac{k_0 - k}{k_0}\)^{2i\nu_r(k_0)} & 1\end{bmatrix},
	\lambda\in L_4\cap Int(C_r),
	\\
	&=\begin{bmatrix}1 & \frac{\frac{1}{r(k_0)}}{1+\frac{1}{|r(k_0)|^2}}\cdot\chi_r^2(k_0; k_0)e^{-2ih(k;k_0,\tau)}\(\frac{k_0-k}{k_0}\)^{-2i\nu_r(k_0)} \\ 0 & 1 \end{bmatrix}, \lambda\in L_2\cap Int(C_r),
\end{align*}

\subsection{Reconstruction of $\mathcal{E}, \mathcal{N}, \rho$}

\begin{lemma}\label{lem_reconstruction}
In terms of the function $M_{err}$, functions $\mathcal{E}(t, x), \mathcal{N}(t, x), \rho(t, x)$ can be expressed as follows:
\begin{itemize}
\item[(a)]
In the case $Z_b\cap Int(L_1, L_3)=\emptyset$ we have
\begin{align*}
\mathcal{E}(t, x) = -\lim\limits_{k\to\infty}4 \ii k M_{err, 12}(k),
\qquad
\begin{pmatrix}
\mathcal{N}(t, x) & \rho(t, x) \\ \ol{\rho(t, x)} & -\mathcal{N}(t, x)
\end{pmatrix}
=
-M_{err}(+\ii 0)\sigma_3 M_{err}(+\ii 0)^{-1},
\end{align*}
where $M_{err}(k)=M_{err}(k;k_0,\tau),$ $\tau=t-x,$ $k_0=\frac12\sqrt{x/(t-x)}.$

\item[(b)]
In the case $Z_b\cap Int(L_1, L_3) = \left\{k_j\right\}$, we have
$
\mathcal{E}(t, x) = 4B_j + \mathcal{E}_{err}(t, x)
$ and 
\begin{equation}\label{second_formula}
\begin{pmatrix}
\mathcal{N}(t, x) & \rho(t, x)
\\
\ol{\rho(t, x)} & -\mathcal{N}(t, x)
\end{pmatrix}
=
M_{err}(+i 0)
\begin{pmatrix}
-1 + \frac{2|B_j|^2}{|k_j|^2} & \frac{-2 i B_j}{\ol{k_j}}\(1 - \frac{i A_j}{k_j}\)
\\
\frac{2 i \, \ol{B_j}}{k_j}\(1 + \frac{i A_j}{\ol{k_j}}\) & 1 - \frac{2|B_j|^2}{|k_j|^2}
\end{pmatrix}
M_{err}(+i 0)^{-1},
\end{equation}
where 
\begin{align*}
\mathcal{E}_{err}(t, x) = -\lim\limits_{k\to\infty}4 \ii k M_{err, 12}(k),
\qquad
A_j = \frac{2\Im k_j\cdot |\w_j|^2}{1 + |\w_j|^2},\quad B_j = \frac{-2\Im k_j \cdot \ol{\w_j}}{1 + |\w_j|^2},
\end{align*}
and where $\w_j$ is defined in \eqref{wj}. More explicitly, $\w_j = |\w_j|e^{i\arg \w_j},$ where
\begin{multline*}
|\w_j| = \frac{1}
{2\Im k_j\cdot |a(k_j)\dot{b}(k_j)|}
\exp\left\{-2\Im k_j\(t-x-\frac{x}{4\left[(\Re k_j)^2 + (\Im k_j)^2\right]}\)\right\}
\cdot\\\cdot
\exp\left[
\frac{-\Im k_j}{\pi}\int_{-k_0}^{k_0}\frac{\ln\(1 + |r(s)|^{-2}\)\, ds}
{(s-\Re k_j)^2 + (\Im k_j)^2}\right]
\prod\limits_{p: |k_p|<|k_j|}\left|\frac{k_j - k_p}{k_j - \ol{k_p}}\right|^2\,,
\end{multline*}
and
\begin{multline*}
\arg \w_j = -\arg\(a(k_j)\dot{b}(k_j)\) + 2\Re k_j\cdot\(t-x+\frac{x}{4[(\Re k_j)^2 + (\Im k_j)^2]}\)
+
\\
+\frac{1}{\pi}\int_{-k_0}^{k_0}
\frac{(s - \Re k_j)\ln\(1 + |r(s)|^{-2}\)\, ds}{(s-\Re k_j)^2 + (\Im k_j)^2}
+
2\sum\limits_{p: |k_p|<|k_j|}\arg\(\frac{k_j - k_p}{k_j - \ol{k_p}}\).
\end{multline*}
\end{itemize}
\end{lemma}

\begin{proof}
\color{white}\end{proof}\color{black}
\subsubsection*{Function $\mathcal{E}(t, x).$}Tracking back the chain of transformations \eqref{chain_transformations}, $M \to \ldots \to M^{(6)},$ we see that
\begin{align*}
\mathcal{E}(t, x) = -4i\lim\limits_{k\to\infty} k M(t,x;k) = -4i\lim\limits_{k\to\infty} k M^{(j)}(k;k_0,\tau),\quad \mbox{ for } j = 1, 2, \ldots, 6,
\end{align*}
and hence 
\[
\mathcal{E}(t, x) = \mathcal{E}_{appr}(t, x) + \mathcal{E}_{err}(t, x),
\]
where
\[
\mathcal{E}_{appr}(t, x) = -4i\lim\limits_{k\to\infty} k M_{appr}(k;k_0,\tau),
\qquad
\mathcal{E}_{err}(t, x) = -4i\lim\limits_{k\to\infty} k M_{err}(k;k_0,\tau).
\]
The term $\mathcal{E}_{appr}(t, x)$ can be computed explicitly. Indeed, using expression \eqref{M_glob} for $M_{glob}$ from subsection \ref{sect_global_parametrix}, we see that
$\mathcal{E}_{appr}(t, x) = 0$ in the case $Z_b\cap Int(L_1, L_3) = \emptyset,$
and $\mathcal{E}_{appr}(t, x) = 4B_j$ in the case $k_j \in Int(L_1, L_3),$ where $B_j$ is defined in \eqref{AjBj},
\[
\mathcal{E}(t, x) = 4B_j + \mathcal{E}_{err}(t, x).
\]

\subsubsection*{Functions $\mathcal{N}(t, x)$ and $\rho(t, x).$}
Similarly, 
$$
\begin{pmatrix}
\mathcal{N}(t, x) & \rho(t, x)
\\
\ol{\rho(t, x)} & -\mathcal{N}(t, x)
\end{pmatrix}
=
M(+i 0)\sigma_3M(+i 0)^{-1}
=
M^{(j)}(+i 0)\sigma_3 M^{(j)}(+i 0)^{-1}\qquad \mbox{ for } j=1, 2, 3,
$$
where $M(k) = M(t,x;k),$ $M^{(j)}(k) = M^{(j)}(k;k_0,\tau),$ and 
\[
\begin{pmatrix}
\mathcal{N}(t, x) & \rho(t, x)
\\
\ol{\rho(t, x)} & -\mathcal{N}(t, x)
\end{pmatrix}
=
-M^{(j)}(+i 0)\sigma_3 M^{(j)}(+i 0)^{-1}\qquad \mbox{ for } j=4, 5, 6.
\]
Hence (below, we drop dependence of $M_{err}, M_{appr}$ on $k_0, \tau$),
\[
\begin{pmatrix}
\mathcal{N}(t, x) & \rho(t, x)
\\
\ol{\rho(t, x)} & -\mathcal{N}(t, x)
\end{pmatrix}
=
M_{err}(+i 0)\cdot\(-M_{appr}(+i 0)\sigma_3 M_{appr}(+i 0)^{-1}\)\cdot M_{err}(+i 0)^{-1}.
\]
The middle factor here can be computed explicitly, we have 
$$M_{appr}(+i 0) = \begin{pmatrix}1 - \frac{i A_j}{k_j} & \frac{-i B_j}{\ol{k_j}} \\ \frac{-i \, \ol{B_j}}{k_j} & 1 + \frac{i A_j}{\ol{k_j}}\end{pmatrix}$$
in the case $k_j \in Z_b\cap Int(L_1, L_3)$ and $M_{appr}(+i 0 ) = I$ in the case that there  are no zeros of $b(.)$ between the lines $L_1$ and $L_3.$ Thus
\[
-M_{appr}(+i 0)\sigma_3 M_{appr}(+i 0)^{-1} = 
\begin{pmatrix}
-1 + \frac{2|B_j|^2}{|k_j|^2} & \frac{-2 i B_j}{\ol{k_j}}\(1 - \frac{i A_j}{k_j}\)
\\
\frac{2 i \, \ol{B_j}}{k_j}\(1 + \frac{i A_j}{\ol{k_j}}\) & 1 - \frac{2|B_j|^2}{|k_j|^2}
\end{pmatrix}.
\]
and 
$$
\begin{pmatrix}
\mathcal{N}(t, x) & \rho(t, x)
\\
\ol{\rho(t, x)} & -\mathcal{N}(t, x)
\end{pmatrix}
=
M_{err}(+i 0)
\begin{pmatrix}
-1 + \frac{2|B_j|^2}{|k_j|^2} & \frac{-2 i B_j}{\ol{k_j}}\(1 - \frac{i A_j}{k_j}\)
\\
\frac{2 i \, \ol{B_j}}{k_j}\(1 + \frac{i A_j}{\ol{k_j}}\) & 1 - \frac{2|B_j|^2}{|k_j|^2}
\end{pmatrix}
M_{err}(+i 0)^{-1}.
$$

{
\color{white}\begin{proof}\color{black}
\hskip-13mm {\small \it End of proof of Lemma \ref{lem_reconstruction}.}
\end{proof}
}

\subsection{Riemann-Hilbert problem for $M_{err}$}
Function $M_{err}$ satisfies the following RH problem (below we again drop the dependence on $k_0, \tau$):
\begin{enumerate}
\item {\it Analyticity}: $M_{err}(k)$ is analytic in $\mathbb{C}\setminus\Sigma_{err},$ where $\Sigma_{err} = \bigcup\limits_{j=1}^{4}L_j\cup C_l\cup C_r.$
\item {\it Normalisation}: $M_{err}(k)\to I$ as $k\to\infty.$
\item {\it Jump}: $M_{err, -}(k) = M_{err, +}(k) J_{err}(k),$ 
where
\begin{equation}\label{J_err}
\begin{split}
& J_{err}(k) = 
M_{glob}(k)
\phi_l(k)^{\sigma_3}
P(\lambda_l; r_0 = r(-k_0))
\phi_l(k)^{-\sigma_3}
M_{glob}(k)^{-1},\quad k\in C_l,
\\
& J_{err}(k) = 
M_{glob}(k)
\phi_r(k)^{\sigma_3}
{\scriptsize \begin{pmatrix}0 & 1 \\ -1 & 0 \end{pmatrix}}
P(\lambda_r; r_0 = \ol{r(k_0)})
{\scriptsize \begin{pmatrix}0 & -1 \\ 1 & 0 \end{pmatrix}}
\phi_r(k)^{-\sigma_3}
M_{glob}(k)^{-1},\ k\in C_l,
\\
& J_{err}(k) = M_{glob}(k)J^{(6)}(k)M_{glob}(k)^{-1},\quad k\in \bigcup\limits_{j=1}^{4} L_j \setminus (C_l\cup C_r),
\\
& J_{err}(k) = M_{glob}(k)
J_{appr}(k)^{-1}J^{(6)}(k)
M_{glob}(k)^{-1},\quad k \in \bigcup\limits_{j=1}^{4}L_j\cap(C_l\cup C_r).
\end{split}
\end{equation}
\end{enumerate}

\begin{lemma}\label{lem_J_err}
In the regime $\tau\to+\infty,$ uniformly in $C^{-1} \leq k_0 \leq C,$ $C>1,$ 
the jump matrix $J_{err}$ admits the following estimates:
\begin{align*}
\|J_{err}(k;k_0,\tau) - I\|_{L_{\infty}(\Sigma_{err})\cap L_{1}(\Sigma_{err})\cap L_{2}(\Sigma_{err})}
=
\mathcal{O}\(\frac{1}{\tau^{1/2}}\).
\end{align*}
\end{lemma}
\begin{proof}
\medskip
\medskip
\medskip
\noindent
The jump matrix $J_{err}$ on the contours $L_j, j=1, 2, 3, 4,$ inside the circles $C_l, C_r$ can be written more explicitly as follows:
\begin{align*}
J_{err}(k) &= 
M_{glob}(k)
\begin{pmatrix}
1 & 0 
\\
\(
\frac{r^*(k)^{-1}\cdot\chi_l^{-2}(k, k_0)}{1 + (r(k)r^*(k))^{-1}}
-
\frac{\ol{r(-k_0)}^{\,-1}\cdot\chi_l^{-2}(-k_0, k_0)}{1 + |r(-k_0)|^{-2}}
\)
\cdot\(\frac{k+k_0}{k_0}\)^{-2i\nu_l}\cdot e^{2ih(k)} & 1
\end{pmatrix}
M_{glob}(k)^{-1}, 
\\
&\hskip128mm k\in L_1\cap Int(C_l),
\\
&=M_{glob}(k)
\begin{pmatrix}
1 & 
\(
\frac{\chi_l^2(k, k_0)}{r(k)}
-
\frac{\chi_l^2(-k_0; k_0)}{r(-k_0)}
\)
\hskip-1mm\cdot\hskip-1mm
\(\frac{k+k_0}{k_0}\)^{2i\nu_l} 
\hskip-1.5mm\cdot\hskip-0.5mm
e^{-2ih(k)} 
\\
0 & 1
\end{pmatrix}
M_{glob}(k)^{-1}, 
\ k\in L_3\cap Int(C_l),
\\
&=M_{glob}(k)
\begin{pmatrix}
1 & 
\(
\frac{r(k)^{-1}\cdot\chi_l^{2}(k, k_0)}{1 + (r(k)r^*(k))^{-1}}
-
\frac{r(-k_0)^{-1}\cdot\chi_l^{2}(-k_0, k_0)}{1 + |r(-k_0)|^{-2}}
\)
\cdot\(\frac{k+k_0}{k_0}\)^{2i\nu_l}\cdot e^{-2ih(k)} 
\\
0 & 1
\end{pmatrix}
M_{glob}(k)^{-1}, 
\\
&\hskip128mm k\in L_2\cap Int(C_l),
\\
&=M_{glob}(k)
\begin{pmatrix}
1 & 0
\\
\(
\frac{\chi_l^{-2}(k, k_0)}{r^*(k)}
-
\frac{\chi_l^{-2}(-k_0, k_0)}{\ol{r(-k_0)}}
\)
\hskip-1mm\cdot\hskip-1mm
\(\frac{k_0}{k+k_0}\)^{2i\nu_l}
\hskip-2mm\cdot\hskip-0.5mm
 e^{2ih(k)} 
 & 1
\end{pmatrix}
M_{glob}(k)^{-1}, 
\ k\in L_4\cap Int(C_l),
\end{align*}

\begin{align*}
J_{err}(k) 
&=M_{glob}(k)
\begin{pmatrix}
1 & 
\(
\frac{\chi_r^2(k, k_0)}{r(k)}
-
\frac{\chi_r^2(k_0; k_0)}{r(k_0)}
\)
\hskip-1mm\cdot\hskip-1mm
\(\frac{k_0-k}{k_0}\)^{-2i\nu_r} 
\hskip-1.5mm\cdot\hskip-0.5mm
e^{-2ih(k)} 
\\
0 & 1
\end{pmatrix}
M_{glob}(k)^{-1}, 
\ k\in L_3\cap Int(C_r),
\\
&= 
M_{glob}(k)
\begin{pmatrix}
1 & 0 
\\
\(
\frac{r^*(k)^{-1}\cdot\chi_r^{-2}(k, k_0)}{1 + (r(k)r^*(k))^{-1}}
-
\frac{\ol{r(-k_0)}^{\,-1}\cdot\chi_r^{-2}(k_0, k_0)}{1 + |r(k_0)|^{-2}}
\)
\cdot\(\frac{k_0-k}{k_0}\)^{2i\nu_r}\cdot e^{2ih(k)} & 1
\end{pmatrix}
M_{glob}(k)^{-1}, 
\\
&\hskip128mm k\in L_1\cap Int(C_r),
\\
&=M_{glob}(k)
\begin{pmatrix}
1 & 0
\\
\(
\frac{\chi_r^{-2}(k, k_0)}{r^*(k)}
-
\frac{\chi_r^{-2}(k_0, k_0)}{\ol{r(k_0)}}
\)
\hskip-1mm\cdot\hskip-1mm
\(\frac{k_0-k}{k_0}\)^{2i\nu_r}
\hskip-2mm\cdot\hskip-0.5mm
 e^{2ih(k)} 
 & 1
\end{pmatrix}
M_{glob}(k)^{-1}, 
\ k\in L_4\cap Int(C_r),
\\
&=M_{glob}(k)
\begin{pmatrix}
1 & 
\(
\frac{r(k)^{-1}\cdot\chi_r^{2}(k, k_0)}{1 + (r(k)r^*(k))^{-1}}
-
\frac{r(k_0)^{-1}\cdot\chi_r^{2}(k_0, k_0)}{1 + |r(k_0)|^{-2}}
\)
\cdot\(\frac{k_0-k}{k_0}\)^{-2i\nu_r}\cdot e^{-2ih(k)} 
\\
0 & 1
\end{pmatrix}
M_{glob}(k)^{-1}, 
\\
&\hskip128mm k\in L_2\cap Int(C_r),
\end{align*}
It follows that inside the circle $C_l, C_r,$ respectively, the jump matrix $J_{err}$ admits the estimate
\begin{equation}\label{estimates_LjClCr}
J_{err}(k) = I + \mathcal{O}(k\mp k_0)\cdot e^{-2\tau k_0\,\left|\frac{k\mp k_0}{k_0}\right|^2}.
\end{equation}
Furthermore, 
\begin{equation}\label{estimates_ClCr}
J_{err}(k) = I + \mathcal{O}(\tau^{-1/2} k_0^{-1/2}) \quad \mbox{ uniformly for } k\in C_l\cup C_r,
\end{equation}
as follows from the fact that $M_{glob}(k)$ and $\phi_l(k), \phi_r(k)$ are uniformly bounded from $0$ and $\infty$ on the circles $C_l, C_r,$ respectively, 
and from the estimates
$$
J_{err}(k) = M_{glob}(k)
\phi_l(k)^{\sigma_3}
P(\lambda_l; r_0 = r(-k_0))
\phi_l(k)^{-\sigma_3}
M_{glob}(k)^{-1}
=
I + \mathcal{O}(\lambda_l^{-1})
=
I + \mathcal{O}(\tau^{-1/2}k_0^{-1/2}),
$$
and similarly 
$J_{err}(k) = I + \mathcal{O}(\tau^{-1/2} k_0^{-1/2})$ for $k\in C_r.$

Using the estimates \eqref{estimates_ClCr}, \eqref{estimates_LjClCr} and using the fact that on the parts of the contours $L_j$ outside the disks $Int(C_l)$, $Int(C_r)$ the jump matrix is exponentially close to the identity matrix, we obtain the statement of the Lemma.
\end{proof}

\subsection{Asymptotics of $M_{err}$}

\begin{lemma}\label{lem_M_err}
Let $C>1$ be a fixed real number. In the regime $\tau\to\infty$, uniformly for $\frac{1}{C}\leq k_0 \leq C$, one has
\begin{align*}
\lim\limits_{k\to\infty}-4ik \(M_{err}(k)-I\) = \frac{-4i}{2\pi i}\int\limits_{C_l\cup C_r}(J_{err}(s)-I)ds + \mathcal{O}(\tau^{-1}),
\\
M_{err}(0) = I - \frac{1}{2\pi i}\int\limits_{C_l\cup C_r}\frac{1}{s}(J_{err}(s)-I)ds + \mathcal{O}(\tau^{-1}).
\end{align*}
\end{lemma}
\begin{proof}
The error matrix $M_{err}$ can be obtained as 
\[
M_{err} = I + \mathcal{C}[M_{err, +}(.)(I - J_{err}(.))],
\]
where $M_{err, +}$ is the solution of the singular integral equation
\[
M_{err, +} = I + \mathcal{C}_+[M_{err, +}(.)(I - J_{err}(.))],
\]
and
where we denote
\[
\mathcal{C}f(k) = \frac{1}{2\pi i}\int_{\Sigma_{err}}\frac{f(s)\,ds}{s-k},
\quad
\mathcal{C}_{\pm}f(k) = \frac{1}{2\pi i}\int_{\Sigma_{err}}\frac{f(s)\,ds}{(s-k)_{\pm}}\ .
\]
It follows that
$$
M_{err}(k) = I + \frac{1}{2\pi i}\int_{\Sigma_{err}}\frac{(I-J_{err}(s))\ ds}{s-k}
+ \frac{1}{2\pi i}\int_{\Sigma_{err}}\frac{(M_{err, +}(s)-I)(I-J_{err}(s))\ ds}{s-k}
$$
and 
\begin{align*}
&
\lim\limits_{k\to\infty}-4ik(M_{err}(k)-I) = \frac{-4i}{2\pi i}\int_{\Sigma_{err}}(J_{err}(s)-I)\ ds
+ \frac{-4i}{2\pi i}\int_{\Sigma_{err}}(M_{err, +}(s)-I)(J_{err}(s)-I)\ ds\,,
\\
&
M_{err}(0) = I - \frac{1}{2\pi i}\int_{\Sigma_{err}}\frac{1}{s}(J_{err}(s)-I)\ ds
- \frac{1}{2\pi i}\int_{\Sigma_{err}}\frac{1}{s}(M_{err, +}(s)-I)(J_{err}(s)-I)\ ds\,.
\end{align*}
Using estimates \eqref{estimates_LjClCr}, \eqref{estimates_ClCr}, 
one obtains
\[
\int_{C_{l,r}}(J_{err}(k)-I)ds = \mathcal{O}\(k_0^{1/2}\tau^{-1/2}\),
\quad
\int_{L_j\cap C_{l,r}}(J_{err}(k)-I)ds = \mathcal{O}\(k_0 \tau^{-1}\),
\]
$$
\int_{C_{l,r}}(M_{err,+}(s)-I)(J_{err}(s)-I)ds = \mathcal{O}(\tau^{-1}),
\quad
\int_{L_j\cap C_{l,r}}(M_{err,+}(s)-I)(J_{err}(s)-I)ds = \mathcal{O}(k_0^{3/2}\tau^{-3/2}),
$$
from where it follows that the main contribution to $\int_{\Sigma_{err}}(J_{err}(s)-I)ds$ comes from $C_{l}, C_r.$ Similar argument applies for the integral 
$\int_{\Sigma_{err}}s^{-1}(J_{err}(s)-I)ds$, from where we obtain the statement of the Lemma.
\end{proof}

\begin{lemma}\label{lem_integrals_Clr}
Under conditions of Lemma \ref{lem_M_err}, we have 
\begin{align*}
&\frac{1}{2\pi i}\int_{C_l}(J_{err}(k) - I)dk
=
M_{glob}(-k_0)
\begin{pmatrix}
0 & i\,e^{i\,\omega_l}
\\
i\,e^{-i\,\omega_l} & 0
\end{pmatrix}
M_{glob}(-k_0)^{-1}
\cdot
\frac{\sqrt{k_0\,\nu_l}}{2\sqrt{\tau}}
+
\mathcal{O}\(\frac{1}{\tau}\),
\\
&\frac{1}{2\pi i}\int_{C_r}(J_{err}(k) - I)dk
=
M_{glob}(k_0)
\begin{pmatrix}
0 & i\,e^{i\,\omega_r}
\\
i\,e^{-i\,\omega_r} & 0
\end{pmatrix}
M_{glob}(k_0)^{-1}
\cdot
\frac{\sqrt{k_0\,\nu_r}}{2\sqrt{\tau}}
+
\mathcal{O}\(\frac{1}{\tau}\),
\\
&\frac{1}{2\pi i}\int_{C_l}\frac{1}{k}(J_{err}(k) - I)dk
=
M_{glob}(-k_0)
\begin{pmatrix}
0 & -i\,e^{i\,\omega_l}
\\
-i\,e^{-i\,\omega_l} & 0
\end{pmatrix}
M_{glob}(-k_0)^{-1}
\cdot
\frac{\sqrt{\nu_l}}{2\sqrt{k_0\,\tau}}
+
\mathcal{O}\(\frac{1}{\tau}\),
\\
&\frac{1}{2\pi i}\int_{C_r}\frac{1}{k}(J_{err}(k) - I)dk
=
M_{glob}(k_0)
\begin{pmatrix}
0 & i\,e^{i\,\omega_r}
\\
i\,e^{-i\,\omega_r} & 0
\end{pmatrix}
M_{glob}(k_0)^{-1}
\cdot
\frac{\sqrt{\nu_r}}{2\sqrt{k_0\,\tau}}
+
\mathcal{O}\(\frac{1}{\tau}\),
\end{align*}
where $\omega_l = \omega_l(t,x), \omega_r = \omega_r(t,x)$ are defined as follows:
\begin{multline*}
\omega_l(t, x) = 4\tau k_0 - \nu_l \ln(16\tau k_0) - \frac{1}{\pi}\int_{-k_0}^{k_0}\frac{\ln\frac{1 + |r(s)|^{-2}}{1 + |r(-k_0)|^{-2}}\, ds}{s+k_0} + \arg\(a(-k_0)b(-k_0)\) + \arg\Gamma(i\nu_l)
\\
+\sum\limits_{k_j\in Z_b\cap Int(L_3, \mathbb{R})}2\arg\frac{k_0 + \ol{k_j}}{k_0 + k_j}
- \frac{\pi}{4},
	\end{multline*}
\begin{multline*}
\omega_r(t, x) = - 4 \tau k_0 + \nu_r \ln(16\tau k_0) - \frac{1}{\pi}\int_{-k_0}^{k_0}\frac{\ln\frac{1 + |r(s)|^{-2}}{1 + |r(k_0)|^{-2}}\, ds}{s-k_0} + \arg (a(k_0) b(k_0)) - \arg\Gamma(i\nu_r) 
\\
+\sum\limits_{k_j\in Z_b\cap Int(L_3, \mathbb{R})}2\arg\frac{k_0-\ol{k_j}}{k_0-k_j}
+ \frac{\pi}{4}.
\end{multline*}
\end{lemma}

\begin{proof}
The jump matrix $J_{err}$ on the circles $C_l, C_r$ takes the form (see the first two formulae of \eqref{J_err} and formula 
\eqref{PC_asymp})
$$
J_{err}(k) - I = M_{glob}(k)
\left[
\begin{pmatrix}
0 & \frac{\phi_l^2(k)e^{-\pi i/4}\beta_{2,l}}{2}
\\
\frac{e^{-3\pi i/4}\beta_{1,l}}{2\phi_l^2(k)}
\end{pmatrix}
\frac{1}{\sqrt{\tau\,k_0}\,z_l} + \mathcal{O}\(\frac{1}{\tau\,k_0\,z_l^2}\)
\right]
M_{glob}(k)^{-1},\ k\in C_l,
$$
$$
J_{err}(k) - I = M_{glob}(k)
\left[
\begin{pmatrix}
0 & \frac{\phi_r^2(k)e^{\pi i/4}\beta_{1,r}}{2}
\\
\frac{e^{3\pi i/4}\beta_{2,r}}{2\phi_r^2(k)}
\end{pmatrix}
\frac{1}{\sqrt{\tau\,k_0}\,z_r} + \mathcal{O}\(\frac{1}{\tau\,k_0\,z_r^2}\)
\right]
M_{glob}(k)^{-1},\ k\in C_r.
$$
Computing the first-order residues, we find the integrals over the circles $C_l, C_r$
and thus obtain the statement of the Lemma.
\end{proof}

\subsubsection*{Proof of Theorem \ref{thm_tail}.} {\it Functions $\mathcal{N}, \rho.$}
From the symmetry
$$\ol{M_{err}(\ol k)} = \begin{bmatrix}0 & 1 \\ -1 & 0\end{bmatrix}
M_{err}(k) \begin{bmatrix}0 & -1 \\ 1 & 0\end{bmatrix},$$
it follows that formula \eqref{second_formula} from Lemma \ref{lem_reconstruction} can be written in the form
\begin{equation}\label{Nrho_expr}
\begin{pmatrix}
\mathcal{N} & \rho \\ \ol{\rho} & -\mathcal{N}
\end{pmatrix}
=
\begin{pmatrix}
1 + X & Y \\ -\ol{Y} & 1 + \ol{X}
\end{pmatrix}
\begin{pmatrix}
-P & Q \\ \ol{Q} & P
\end{pmatrix}
\begin{pmatrix}
1 + \ol{X} & -Y \\ \ol{Y} & 1 + X
\end{pmatrix},
\end{equation}
where $\mathcal{N} = \mathcal{N}(t,x),$ $\rho = \rho(t,x)$ and
\[
1 + X = M_{err}(0)_{11},\ 
Y = M_{err}(0)_{12},
\quad
P = 1 - \frac{2|B_j|^2}{|k_j|^2},\
Q = \frac{-2i\,B_j}{\ol{k_j}}\(1 - \frac{i\,A_j}{k_j}\).
\]
Note that $\det M_{err}(0) = 1$ and hence $|1 + X|^2 + |Y|^2 = 1.$ Multiplying matrices in the right-hand side of \eqref{Nrho_expr}, we find that
\[
\mathcal{N}(t, x)
=
(-|1 + X|^2 + |Y|^2)P + (1 + \ol{X})\,Y\,\ol{Q} + (1+X)\,\ol{Y}\,Q,
\]
\[
\rho(t, x)
=
2(1+X)YP - Y^2\, \ol{Q} + (1+X)^2 Q.
\]
From Lemmas \ref{lem_integrals_Clr}, \ref{lem_M_err} it follows that $X = \mathcal{O}(\tau^{-1/2}),$ $Y = \mathcal{O}(\tau^{-1/2})$, thus
\[
\mathcal{N}(t, x)
=
-P + Y\,\ol{Q} + \ol{Y}\,Q + \mathcal{O}(\tau^{-1}),
\qquad
\rho(t, x)
=
2YP + \(1 + 2X\) Q + \mathcal{O}(\tau^{-1}).
\]
Multiplying matrices from Lemma \ref{lem_integrals_Clr}, $M_{err}(0)$ and thus $X$, $Y$ can be written more explicitly,
\begin{align*}
X &=  
\frac{\sqrt{\nu_l}}{2\sqrt{k_0\,\tau}}
\(
\(1 + \frac{i\,A_j}{k_0 + \ol{k_j}}\)\frac{B_j\,e^{-i\,\omega_l}}{k_0 + \ol{k_j}}
-
\(1 - \frac{i\,A_j}{k_0 + k_j}\)\frac{\ol{B_j}\,e^{i\,\omega_l}}{k_0 + k_j}
\)
\\
&+ 
\frac{\sqrt{\nu_r}}{2\sqrt{k_0\,\tau}}
\(
\(1 - \frac{i\,A_j}{k_0 - \ol{k_j}}\)\frac{B_j\,e^{-i\,\omega_r}}{k_0 - \ol{k_j}}
-
\(1 + \frac{i\,A_j}{k_0 - k_j}\)\frac{\ol{B_j}\,e^{i\,\omega_r}}{k_0-k_j}
\)
\end{align*}
and
\begin{align*}
Y = \frac{i\,\sqrt{\nu_l}}{2\sqrt{k_0\,\tau}}
\(
e^{i\,\omega_l}\(\!\!1 - \frac{i\,A_j}{k_0 + k_j}\)^2
+
\frac{B_j^2\,e^{-i\,\omega_l}}{(k_0 + \ol{k_j})^2}
\)
-\frac{i\,\sqrt{\nu_r}}{2\sqrt{k_0\,\tau}}
\(
e^{i\,\omega_r}\(\!\!1 + \frac{i\,A_j}{k_0 - k_j}\)^2
+
\frac{B_j^2\,e^{-i\,\omega_r}}{(k_0 - \ol{k_j})^2}
\)\!.
\end{align*}

\medskip
\noindent
{\it Function $\mathcal{E}.$}
It follows from Lemma \ref{lem_reconstruction} that
\begin{align*}
\mathcal{E}(t, x)
&=
4B_j - \lim\limits_{k\to\infty}4ik(M(k)-I)_{12}
\\
&=4B_j - \frac{4i}{2\pi i}\int_{C_l\cup C_r}(J_{err}(k)-I)_{12}dk + \mathcal{O}(\tau^{-1})
\\
&=
4B_j
+\frac{2\sqrt{k_0\,\nu_l}}{\sqrt{\tau}}
\left[M_{glob}(-k_0)
\begin{pmatrix}
0 & e^{i\omega_l} \\ e^{-i\omega_l} & 0
\end{pmatrix}
M_{glob}(-k_0)^{-1}
\right]_{12}
\\
&\quad+\frac{2\sqrt{k_0\,\nu_r}}{\sqrt{\tau}}
\left[M_{glob}(k_0)
\begin{pmatrix}
0 & e^{i\omega_r} \\ e^{-i\omega_r} & 0
\end{pmatrix}
M_{glob}(k_0)^{-1}
\right]_{12}
+ \mathcal{O}(\tau^{-1}).
\end{align*}
Multiplying the matrices, we complete the proof of Theorem \ref{thm_tail}.

\begin{rem}
Theorem \ref{thm_tail} can be proved under the weaker assumption of an exponential decay of the input pulse.
In such a case, Step 5 of Section \ref{sect_steps_transformations} is not performed, but instead the parabolic cylinder parametrices in the neighborhoods of the points $\pm k_0$ must be adjusted.
\end{rem}

\appendix
\addtocontents{toc}{\fixappendix}

\section{Proof of Lemma \ref{lem_inversion}}\label{sect_appendix}

\begin{proof}
Let us look for $y$ in the form $y = z + \gamma p,$ where $0< p \ll z;$ substituting this in \eqref{ziny}, we get
$p = \ln(z + \gamma p),$ and $z = e^p - \gamma p.$ Now we look for $p$ in the form $p = \ln z - q,$ where $q\ll p.$ Then
$$e^{-q} - \frac{\gamma \ln z}{z} + \frac{\gamma q}{z} - 1 = 0.$$
$$e^{-q} - 1 + \delta + \varepsilon q = 0,\quad
\mbox{ where we denoted }\quad \delta = \frac{-\gamma\ln z}{z},\quad \varepsilon = \frac{\gamma}{z}.$$
It follows that $q\to 0$ as $z\to+\infty.$ By Rouche's theorem, there exists a fixed $r>0$ such that for sufficiently small $\varepsilon, \delta$ the equation for $q$ has exactly one root in the disk $\left\{q:\, |q|<r\right\}$. Then 
\[
q = \frac{1}{2\pi i}\oint\limits_{|\zeta|=r}\frac{\zeta\(-e^{-\zeta} + \varepsilon\)\, d\zeta}{e^{-\zeta} - 1 + \delta + \varepsilon \zeta}
=
\frac{1}{2\pi i}\oint\limits_{|\zeta|=r}\frac{\zeta\(-e^{-\zeta} + \varepsilon\)}{e^{-\zeta} - 1}
\sum\limits_{j=0}^{\infty}\frac{(-1)^j(\delta+\varepsilon\zeta)^j}{(e^{-\zeta} - 1)^j}\,d\zeta,
\]
where we integrate in the counter-clockwise direction.
Thus $q$ can be expanded in a series $$q = \sum\limits_{j=1}^{\infty}\frac{1}{\zeta^j}\sum\limits_{k=0}^j c_{kj}\ln^k z $$
for some real $c_{kj}.$
Computing the corresponding residues, we obtain the statement of the lemma.
\end{proof}

\section{Uniqueness of the solution of the ibv problem}\label{sect_uniqueness}
The proof of uniqueness is very similar to the one for the case of MB equations with retarded time (cf. \cite[Appendix A]{LM2022}). It can be carried out for a more general notion of solution than the classical one.

{\it Definition.}
Let $\mathcal{E}_0, \rho_0:[0, +\infty)\to\mathbb{C},\quad $ $\mathcal{N}_0:[0, +\infty)\to\mathbb{R}, \quad \mathcal{E}_1:(0, +\infty)\to\mathbb{C}$ be given functions.  
We say that a triple of locally integrable functions $\mathcal{E}, \rho:[0, +\infty)\times[0, +\infty)\to\mathbb{C},$ $\mathcal{N}:[0, +\infty)\times[0, +\infty)\to\mathbb{R}$ satisfies the Maxwell-Bloch system \eqref{MB1a} in the {\it subclassical } sense,  
if for all $t\geq0, x\geq0, s\geq0$
\begin{equation}\label{MBintegralform}
\begin{split}
& \mathcal{E}(t+s, x+s) = \mathcal{E}(t, x) + \int_{0}^{s}\rho(t+y, x+y)dy,
\\
& \rho(t+s, x) = \rho(t, x) + \int_{0}^{s}\mathcal{N}(t+y, x)\mathcal{E}(t+y, x)dy,
\\
& \mathcal{N}(t+s, x) = \mathcal{N}(t, x) - \int_{0}^{s}\Re\left[\rho(t+y, x)\ol{\mathcal{E}(t+y, x)}\right]dy.
\end{split}
\end{equation}
Note that it follows that $\mathcal{E}$ is differentiable in the direction $(1, 1)$
and $\mathcal{N}, \rho$ have partial derivative in $t.$ It also follows that the limits
$\lim\limits_{s\to0+}\mathcal{E}(t+s, s),$
$\lim\limits_{s\to0+}\mathcal{E}(s, x+s),$
$\lim\limits_{s\to0+}\mathcal{N}(s, x),$
$\lim\limits_{s\to0+}\rho(s, x)$
exist for all $t>0, x\geq0.$
We say that the initial and boundary conditions \eqref{IBC} are satisfied if 
\[\lim\limits_{s\to0+}\mathcal{E}(t+s, s) = \mathcal{E}_1(t),\qquad \forall t>0,\ \mbox{ and }
\]
\[
\lim\limits_{s\to0+}\mathcal{E}(s, x+s) = \mathcal{E}_1(t),
\qquad
\lim\limits_{s\to0+}\rho(s, x) = \rho_0(x),
\qquad
\lim\limits_{s\to0+}\mathcal{N}(s, x) = \mathcal{N}_0(x),\qquad \forall x\geq0.
\]

\begin{prop}\label{prop_uniqueness}
Let functions $\mathcal{E}_0, \rho_0:[0, +\infty)\to\mathbb{C},\quad $ $\mathcal{N}_0:[0, +\infty)\to\mathbb{R}, \quad \mathcal{E}_1:(0, +\infty)\to\mathbb{C}$ be given, let $\mathcal{N}_0,$ $\rho_0$ satisfy $\mathcal{N}_0(x)^2 + |\rho_0(x)|^2\equiv1$ for all $x\geq0$
and let $\mathcal{E}_j, j=0,1$ satisfy the following property:
\begin{equation}\label{propE01}
\forall\varepsilon>0 \ \exists\delta>0 \mbox{ such that for any interval $\Delta$ of the length smaller than } \delta, \int_\Delta|\mathcal{E}_j(s)|ds < \varepsilon.
\end{equation}
Then there exists at most one subclassical solution of MB system \eqref{MB1a} with initial and boundary conditions \eqref{IBC}.
\end{prop}
\begin{rem}
Property \eqref{propE01} is satisfied for instance if $\int_0^{\infty}|\mathcal{E}_j(s)|ds < \infty,$ $j=0,1.$
\end{rem}
\begin{proof}
{\it Step 1: a priori estimates.} It follows that $\frac{\partial}{\partial t}\(|\rho(t,x)|^2 + \mathcal{N}(t,x)^2\) = 0,$ and hence $|\rho(t,x)|^2 + \mathcal{N}(t,x)^2 = |\rho(0,x)|^2 + \mathcal{N}(0,x)^2 = 1$ for all $t, x\geq0.$
Functions $\mathcal{N}, \rho$ thus possess a priori bounds 
\begin{equation}\label{rhoN_apriori}
|\mathcal{N}(t,x)|\leq 1, \qquad |\rho(t,x)|\leq 1,
\end{equation} 
and from the first of equations \eqref{MBintegralform} we obtain the following a priori bound for $\mathcal{E}:$
\begin{equation}\label{E_apriori}
|\mathcal{E}(t,x)| \leq |\mathcal{E}(t-x, 0)| + x\ \mbox{ for } t>x,\qquad \mbox{ and } |\mathcal{E}(t,x)| \leq |\mathcal{E}(0, x-t)| + x\ \mbox{ for } t\leq x.
\end{equation}

{\it Step 2: rewritting the integral equations in terms of initial and boundary data.}
We need to split the equation for $\mathcal{E}$ in two pieces, depending on whether $t$ is greater than $x$ or not,
\begin{equation}\label{E_integral_2}
\begin{split}
\mathcal{E}(t,x) = \mathcal{E}(0, x-t) + \int_0^t\rho(s, x-t+s)ds, \ t\leq x,
\\
\mathcal{E}(t,x) = \mathcal{E}(t-x, 0) + \int_0^x\rho(t-x+s, s)ds, \ t>x,
\end{split}
\end{equation}
and the equations for $\rho, \mathcal{N}$ take the form
\begin{equation}\label{rhoN_integral_2}
\begin{split}
\rho(t,x) = \rho(0, x) + \int_0^t\mathcal{N}(s, x)\mathcal{E}(s, x)ds,
\\
\mathcal{N}(t,x) = \mathcal{N}(0, x) - \int_0^t\Re\(\rho(s, x)\ol{\mathcal{E}(s, x)}\)ds.
\end{split}
\end{equation}

{\it Step 3: estimates for the differences.} Assume that $\widehat{\mathcal{E}}(t, x), \widehat{\rho}(t, x), \widehat{\mathcal{N}}(t, x)$
and 
$\widetilde{\mathcal{E}}(t, x), \widetilde{\rho}(t, x), \widetilde{\mathcal{N}}(t, x)$
are two different subclassical solutions of the MB equations satisfying the same boundary and initial conditions \eqref{IBC}.
Denote $u(t, x) = \widehat{\mathcal{E}}(t, x) - \widetilde{\mathcal{E}}(t, x),$
$v_1(t, x) = \widehat{\rho}(t, x) - \widetilde{\rho}(t, x),$
$v_2(t, x) = \widehat{\mathcal{N}}(t, x) - \widetilde{\mathcal{N}}(t, x).$
Substituting the two sets of solutions in \eqref{E_integral_2}, \eqref{rhoN_integral_2} and taking the difference between them
we obtain
\[
\begin{split}
u(t, x) = u(0, x-t) + \int_0^tv(s, x-t+s)ds,\ t\leq x,
\\
u(t, x) = u(t-x, 0) + \int_0^xv(t-x+s, s)ds,\ t>x,
\end{split}
\]
\[
\begin{split}
v_1(t, x) = v_1(0, x) + \int_0^t\(v_2(s, x)\widehat{E}(s, x) + \widetilde{\mathcal{N}}(s, x)u(s, x)\)ds,
\\
v_2(t, x) = v_2(0, x) - \int_0^t\Re\(v_1(s, x)\ol{\widehat{E}(s, x)} + \widetilde{\rho}(s, x)\ol{u(s, x)}\)ds.
\end{split}
\]
Since $\widehat{E}, \widehat{\rho}, \widehat{\mathcal{N}}$ and $\widetilde{E}, \widetilde{\rho}, \widetilde{\mathcal{N}}$ satisfy the same boundary and initial conditions, we have
$u(t-x, 0) = 0$ for $t>x$ and $u(0, x-t) = v_j(0, x-t) = 0$ for $t\leq x,$ $j=1,2.$ Applying now the a priori estimates \eqref{rhoN_apriori}, \eqref{E_apriori}, we obtain
\begin{equation}\label{estimates_u}
\begin{split}
|u(t, x)| \leq t\cdot \sup\limits_{(\tau, y)\in[0, t]\times[0, x]}|v_1(\tau, y)|,\quad t\leq x,
\\
|u(t, x)| \leq x\cdot \sup\limits_{(\tau, y)\in[0, t]\times[0, x]}|v_1(\tau, y)|,\quad t>x,
\end{split}
\end{equation}
and the estimates for $v_1, v_2$ will now also depend on whether $t$ is greater or smaller than $x$:
\begin{equation}\label{estimates_v}
\begin{split}
&|v_j(t,x)| \leq \sup\limits_{(\tau, y)\in[0, t]\times[0, x]}|v_{j+1}(\tau, y)| \cdot \int_0^t\(|\widehat{\mathcal{E}}(0, x-s)|+s\)ds
+
t\cdot\sup\limits_{(\tau, y)\in[0, t]\times[0, x]}|u(\tau, y)|,\quad t\leq x,
\\
&|v_j(t,x)| \leq \sup\limits_{(\tau, y)\in[0, t]\times[0, x]}|v_{j+1}(\tau, y)| \cdot \left(\int_0^x\(|\widehat{\mathcal{E}}(0, x-s)|+s\)ds
+
\int_x^t\(|\widehat{\mathcal{E}}(s-x, 0)|+x\)ds\right)
\\&
\hskip10.4cm+
t\cdot\sup\limits_{(\tau, y)\in[0, t]\times[0, x]}|u(\tau, y)|,\quad t > x,
\end{split}
\end{equation}
for $j=1,2$, where we identify $v_3=v_1.$ 
Applying the estimates \eqref{estimates_u}, \eqref{estimates_v} for a point $(\tilde t, \tilde x)\in[0, t]\times[0, x]$ and then taking supremum over $[0, t]\times[0, x],$ we obtain, after some simplifications and after considering separately the cases $t\leq x$ and $t>x$, that for all $t, x\geq0$
\begin{equation}\label{estimates_uv}
\begin{split}
&
\sup\limits_{(\tau, y)\in[0, t]\times[0, x]}|u(\tau, y)|
\leq
t\cdot
\sup\limits_{(\tau, y)\in[0, t]\times[0, x]}|v_1(\tau, y)|,
\\
&
\sup\limits_{(\tau, y)\in[0, t]\times[0, x]}|v_j(\tau, y)|
\leq
\sup\limits_{(\tau, y)\in[0, t]\times[0, x]}|v_{j+1}(\tau, y)|\cdot\(\frac{t^2}{2} + \sup\limits_{\Delta:|\Delta|\leq t}\int_\Delta|\widehat{\mathcal{E}}(0, s)|ds + \int_0^t|\widehat{\mathcal{E}}(s, 0)|ds\)
\\
&\hskip10cm+t\cdot
\sup\limits_{(\tau, y)\in[0, t]\times[0, x]}|u(\tau, y)|,\quad j=1,2,
\end{split}
\end{equation}
where supremum in $\sup\limits_{\Delta:|\Delta|\leq t}\int_\Delta|\widehat{\mathcal{E}}(0, s)|ds$ is taken over all the intervals $\Delta\subset[0, +\infty)$ of length smaller or equal than $t.$

{\it Step 4: concluding estimates.}
Denote
$$M(t,x) = \sup\limits_{(\tau, y)\in[0, t]\times[0, x]}|u(\tau, y)| + \sup\limits_{(\tau, y)\in[0, t]\times[0, x]}|v_1(\tau, y)| + \sup\limits_{(\tau, y)\in[0, t]\times[0, x]}|v_2(\tau, y)|.$$ 
Note that in view of \eqref{rhoN_apriori}, $|v_j(\tau, y)|\leq 2$ and thus the supremum of $|v_j|, j=1,2,$ is finite. In view of the first estimate in \eqref{estimates_uv}, the supremum of $|u|$ over a compact is also finite, and hence $M(t,x)$ is finite.
Adding estimates in \eqref{estimates_uv}, we obtain
\begin{equation}\label{estimate_M}
M(t,x) \leq 
\(3t + t^2 + 2\sup\limits_{\Delta:|\Delta|\leq t}\int_\Delta|\widehat{\mathcal{E}}(0, s)|ds + \int_0^t|\widehat{\mathcal{E}}(s, 0)|ds\)M(t,x).
\end{equation}
Taking $t=t_1>0$ sufficiently small such that the expression in the brackets in the right-hand side of \eqref{estimate_M} is smaller than $1,$ we obtain that $M(t,x)\equiv0$ for all $0\leq t\leq t_1,$ $x\geq0,$ and hence the functions $\widehat{E}, \widehat{\rho}, \widehat{\mathcal{N}}$ and $\widetilde{E}, \widetilde{\rho}, \widetilde{\mathcal{N}}$ coincide in the strip $[0, t_1]\times[0, +\infty).$

{\it Step 5: extending the strip to the whole quarter-plane.} Let $t^*$ be the supremum of all $t,$ such that the two sets of solutions coincide in the strip $[0, t)\times[0, +\infty).$
We want to prove that $t^*=\infty.$ Assuming for the contrary that $t^*$ is finite, we first note that from the continuity of $\rho, \mathcal{N}$ in the direction $(0, 1)$ and the continuity of $\mathcal{E}$ in the direction $(1, 1)$ (as follows from \eqref{MBintegralform}), it follows that the two sets of solutions coincide also at the time $t^*.$
Second, the a priori bounds \eqref{E_apriori} guarantee that the property \eqref{propE01} is satisfied also for $\widehat{\mathcal{E}}(t^*, s).$ 
The Steps 1-4 can now be applied for the quarter-plane $t\geq t^*, x\geq0,$ thus extending the strip where the two sets of solutions coincide. This contradicts the definition of the point $t^*,$ and the above contradiction shows that the two sets of solutions coincide for all $t\geq0, x\geq 0.$
This finishes the proof of the Proposition.
\end{proof}

\section*{Acknowledgement}
O.M. expresses his gratitude to his wife Olga, who helped a lot with improving the quality of the presentation.
V.K. thanks Wolfgang Pauli Institute for financial support  in the context of the WPI thematic program “Quantum Equations and Experiments (2021/22)”.

The research was partially supported by the funds of the Charles University within the follow-up activities of the 4EU+ alliance: ``Support for research and educational cooperation with Ukraine'', project code: 4EU+/UA/F3/09.


	
\end{document}